\documentclass[12pt, letter]{article}
\usepackage{amsmath}
\usepackage{amsthm}
\usepackage{graphicx}
\usepackage{enumerate}
\usepackage{natbib}
\usepackage{url} 
\usepackage{bm}
\usepackage{times}
\usepackage{enumitem}
\usepackage{amssymb,amsfonts,mathrsfs}
\usepackage{bbm}
\usepackage{diagbox}
\usepackage{algorithm, algpseudocode}
\usepackage{tikz-qtree}
\usepackage{booktabs}
\usetikzlibrary{positioning}
\usepackage[per=slash]{siunitx}

\def\T{{ \mathrm{\scriptscriptstyle T} }}

\newtheorem{theorem}{Theorem}
\newtheorem{proposition}{Proposition}

\newtheorem{remark}{Remark}
\newtheorem{assumption}{Assumption}
\newtheorem{definition}{Definition}
\newtheorem{algo}{Algorithm}

\newcommand{\pr}{\mathbb P}

\newcommand{\R}{\mathbb{R}}

\newcommand{\cD}{\mathcal{D}}

\newcommand{\cP}{\mathcal{P}}

\newcommand{\ceil}[1]{\lceil #1 \rceil}

\begin{document}

\def\spacingset#1{\renewcommand{\baselinestretch}%
{#1}\small\normalsize} \spacingset{1}


\title{\bf A Nonparametric Likelihood Approach for Inference in Instrumental Variable Models}

\author{Kwonsang Lee\thanks{Department of Statistics, Sungkyunkwan University, \texttt{kwonsanglee.stat@gmail.com}} \and Bhaswar B. Bhattacharya\thanks{Department of Statistics, University of Pennsylvania, \texttt{bhaswar@wharton.upenn.edu}} \and Jing Qin\thanks{Biostatistics Research Branch, National Institute of Allergy and Infectious Diseases, \texttt{jingqin@niaid.nih.gov}} \and Dylan S. Small\thanks{Department of Statistics, University of Pennsylvania,  \texttt{dsmall@wharton.upenn.edu}}}
\date{}

\maketitle

\begin{abstract}
Instrumental variable methods allow for inference about the treatment effect by controlling for unmeasured confounding in randomized experiments with noncompliance. However, many studies do not consider the observed compliance behavior in the testing procedure, which can lead to a loss of power. In this paper, we propose a novel nonparametric likelihood approach, referred to as the {\it binomial likelihood} method, that incorporates information on compliance behavior while overcoming several limitations of previous techniques and utilizing the advantages of likelihood methods. Our proposed method produces proper estimates of the counterfactual distribution functions by maximizing the binomial likelihood over the space of distribution functions. Using this we propose two versions of a {\it binomial likelihood ratio test} for the null hypothesis of no treatment effect. We show that both versions are more powerful to detect any distributional change than existing methods in finite sample cases, and are asymptotically equivalent to the two-sample Anderson-Darling test. We also develop an efficient algorithm for computing our estimates, and apply the binomial likelihood method to a  study of the effect of Medicaid coverage on mental health using the Oregon Health Insurance Experiment.
\end{abstract}


\spacingset{1.25} 

\section{Introduction}
\label{sec:intro}

The instrumental variables (IV) method is a popular technique for estimating the casual effect of a treatment in the presence of unmeasured confounding \citep{angrist1996, tan2006, baiocchi2014instrumental}. This arises in situations where, even through direct randomization is impossible, an encouragement to take the treatment can be randomized \citep{holland1988}, or there is a ``natural experiment'' such that some people are encouraged to receive the treatment compared to others in a way that is effectively random \citep{angrist1996}. Informally, an instrument is a variable that affects the treatment but is independent of unmeasured confounders and only affects the outcome through affecting the treatment (see Section \ref{sec:assumption} for a more precise definition). Under a monotonicity assumption that the encouraging level of the instrument never causes someone not to take the treatment, the treatment effect can be identified for the compliers, those subjects who would take the treatment if they were encouraged to take the treatment but would not take the treatment if they were not encouraged (see \citet{angrist1996}, \citet{abadie2003semiparametric}, \citet{baiocchi2014instrumental}, \citet{brookhart2007preference}, \citet{cheng2009semiparametric}, \citet{cheng2009efficient}, \citet{hernan2006},  \citet{randomization_inference}, \citet{heterogeneous_treatment_effect}, \citet{ogburn2015doubly}, \citet{tan2006} and the references therein for methods of inference using instrumental variables).

In causal inference, to evaluate the treatment effect on the outcome, Fisher's sharp hypothesis of no effect is often considered, which, in the potential outcome framework \citep{neyman1923, rubin1974estimating}, asserts that the two potential outcomes $Y_i(1)$ and $Y_i(0)$, which are the outcomes individual $i \in \{1, 2, \ldots, n\}$ would experience with or without treatment, respectively, are the same for every individual $i$.  Under the IV assumptions, where the treatment effect can be identified for the compliers, the hypothesis of no effect for compliers can be tested by comparing the distributions of $Y_i(1)$ and $Y_i(0)$ for compliers. Unfortunately, it is difficult to make inference about these distributions since researchers do not know who are the compliers from data. \citet{abadie2002} proposed an approach that indirectly compares the two potential outcome distributions, by using the Kolmogorov--Smirnov test statistic. However, this approach ignores the treatment variable during the testing procedure. Thus, it does not consider the compliance class information of the individuals, which can lead to loss of power, as discussed in \cite{rubin1998more}. 

In this paper, we propose a novel nonparametric likelihood-based approach for comparing the two counterfactual distribution functions, with or without treatment for compliers, that uses the compliance class information and allows for estimation and hypothesis testing in a common holistic framework. This requires a methodological innovation because the usual nonparametric likelihood approach using the empirical likelihood \citep{owen2001} does not work for the IV model because there are infinitely many solutions that maximize the likelihood \citep{geman1982}. Our  proposed \textit{binomial likelihood} (BL)  approach creates a piece of likelihood at each knot (or evaluating point), by using binomially distributed outcomes: outcomes smaller than or equal to the knot, and outcomes larger than the knot. Then, it multiplies together the pieces of these likelihoods across all knots creating a composite likelihood. This is a ``pseudo'' likelihood rather than the true likelihood because the binomial random variables are actually dependent, but are treated as independent in the composite likelihood.  Due to its binomial nature in defining likelihood functions, we specifically call this composite likelihood, the binomial likelihood (BL). Composite likelihood has been found useful in a range of areas including problems in geostatistics, spatial extremes, space-time models, clustered data, longitudinal data, time series and statistical genetics; see \citet{lindsay1988composite}, \citet{heagerty1998composite}, \citet{larribe2011composite}, and \citet{varin2011overview}.

The BL approach can be used for statistical inference similar to the usual likelihood method. For instance, for estimating the distribution functions of the compliers, the {\it maximum binomial likelihood (MBL) estimate} can be obtained by maximizing the BL over the space of distribution functions. Therefore, by definition, the MBL estimates satisfy the necessary conditions for a proper distribution function (increasing and non-negative). This make the BL estimates easily interpretable, and is a major improvement over the naive \textit{plug-in} estimates, which can be non-monotonic and negative. As a consequence, the BL method can be effectively used for making further inferences, such as integrating utility functions or estimating moments of the probability function. Furthermore, similar to classical likelihood ratio tests, the {\it binomial likelihood ratio test (BLRT)} for the null hypothesis of no treatment effect can be constructed by taking the ratio of two BL values that are maximized over the null and the alternative respectively. For computing the MBL estimate and conducting hypothesis testing using the BLRT we develop a computationally efficient iterative algorithm  based on the expectation-maximization (EM) and pool-adjacent-violators (PAV) algorithms. Thus, the BL approach provides the practitioners with a 
comprehensive toolbox for causal inference in non-parametric IV problems. 


The BL method has several attractive limiting and finite-sample properties. To begin with, we show that the MBL estimate for the distribution function of the compliers has the same first-order asymptotics (limiting distribution) as the naive \textit{plug-in} estimates. This shows that the BL estimates, which preserve all the properties of a proper distribution, have no loss in asymptotic efficiency compared to the naive estimates, which can be non-monotone and negative in finite samples.  For hypothesis testing, we show that the BLRT is asymptotically equivalent to the well-known Anderson--Darling two-sample test \citep{pettitt1976}. Since there are no closed form expressions for the BL estimates in general, these asymptotic results are important to the   understanding of the BL approach. The BLRT also has better finite-sample performance 
for detecting distributional changes compared to other baseline methods. The improvement is especially significant in the weak IV setting, exhibiting the importance of incorporating the compliance class information for hypothesis testing in IV models. We also apply the BL approach to study the effect of Medicaid coverage for African American adults on self-reported mental health, as studied by \cite{baicker2013oregon}.

The rest of the article is organized as follows. Basic notation and assumptions of the IV model are discussed in Section~\ref{sec:background}. In this section, we also review the existing plug-in approach for testing the hypothesis of no effect. In Section~\ref{sec:method}, we introduce the BL approach and derive the asymptotic properties of the MBL estimate (Theorem \ref{TH:BLII}). In Section~\ref{sec:blrt}, we develop two versions of the BLRT for testing the null hypothesis, and derive the asymptotic properties of the tests (Theorems \ref{TH:BLRT} and \ref{TH:BLRT_simple}). In Section \ref{SEC:ALGO} we discuss the algorithm for computing the BL estimates and present the numerical results for the BLRT. The analysis of the real data is given in Section~\ref{sec:example}. Proofs of the theorems and additional simulations are given in the supplementary materials.

\section{Framework and Review}
\label{sec:background}

\subsection{Assumptions and Identification with Instrumental Variables}
\label{sec:assumption}

For individual $i$, denote $Z_i$ as the binary IV, $D_i$ as the indicator variable for whether individual $i$ receives the treatment or not, and $Y_i$ as the outcome variable that is continuous in this paper. Using the potential outcome framework \citep{neyman1923, rubin1974estimating}, define $D_i(0)$ as the value that $D_i$ would be if $Z_i$ were to be set to 0, and $D_i(1)$ as the value that $D_i$ would be if $Z_i$ were to be set to 1. Similarly, $Y_i(z, d)$ for $(z,d) \in \{(0,0), (0,1), (1,0), (1,1)\}$, is the value that the outcome $Y_i$ would be if $Z_i=z$ and $D_i=d$. For each individual $i$, the analyst can only observe one of the two potential values $D_i(0)$ and $D_i(1)$, and one of the four potential values $Y_i(0, 0), Y_i(0, 1), Y_i(1, 0), Y_i(1, 1)$. The observed treatment $D_i$ is $D_i = Z_i D_i(1) + (1-Z_i) D_i(0).$ Similarly, the observed outcome $Y_i$ can be expressed as $Y_i = Z_iD_i Y_i(1,1) + Z_i(1-D_i) Y_i(1,0) + (1-Z_i)D_i Y_i(0,1) + (1-Z_i)(1-D_i)  Y_i(0,0)$. An individual's {\it compliance class} is determined by the combination of the potential treatment values $D_i(0)$ and $D_i(1)$, which is denoted by $S_i$: $S_i=$ {\it always-taker} ({\it at}) if $D_i(0)=D_i(1)=1$; $S_i=$ {\it never-taker} ({\it nt}) if $D_i(0)= D_i(1)=0$; $S_i=$ {\it complier} ({\it co}) if $D_i(0)=0, D_i(1)=1$; and $S_i=$ {\it defier} ({\it de}) if $D_i(0)=1, D_i(1)=0$.

For the rest of this paper, the following standard identifying conditions are assumed. The implications of these conditions are briefly explained in the paragraph below; see \citet{angrist1996} for more details on these conditions. 

\begin{assumption} The following identification conditions will be imposed on the instrumental variable model: 
	\label{assump}
	
	\begin{itemize}[leftmargin=1cm]
		\item[(a)] {\it Stable Unit Treatment Value Assumption (SUTVA)} \citep{rubin1986}:  The outcome (treatment) for individual $i$ is not affected by the values of the treatment or instrument (instrument) for other individuals and the outcome (treatment) does not depend on the way the treatment or instrument (instrument) is administered. 
		
		\item[(b)] {\it The instrumental variable $Z_i$ is independent of the potential outcomes $Y_i(z,d)$ and potential treatment  $D_i(z)$.}
		$$
		Z_i \perp\!\!\!\perp \left(Y_i(0,0), Y_i(0,1), Y_i(1,0), Y_i(1,1), D_i(0), D_i(1) \right)
		$$
		
		\item[(c)] {\it Nonzero average causal effect of $Z_i$ on $D_i$}: $\pr(D_i(1) =1) > \pr(D_i(0)=1)$.
		
		\item[(d)] {\it Monotonicity}: $D_i(1) \geq D_i(0)$. 
		
		\item[(e)] {\it Exclusion restriction}: $Y_i(0,d) = Y_i(1, d)$, for $d=0$ or $1$. 
	\end{itemize}
\end{assumption}

Assumption~\ref{assump} enables the causal effect of the treatment for the subpopulation of the compliers to be identified. Condition (a) allows us to use the notation $Y_i(z, d)$ (or $D_i(z)$), which means that the outcome (treatment) for individual $i$ is not affected by the values of the treatment and instrument (instrument) for other individuals. Condition (b) will be satisfied if $Z_i$ is randomized. Condition (c) requires $Z_i$ to have some effect on the average probability of treatment. Condition (d), the monotonicity assumption, means that the possibility of $D_i(0)=1$, $D_i(1)=0$ is excluded, that is, there are no defiers. Condition (e) assures that any effect of $Z_i$ on $Y_i$ must be through an effect of $Z_i$ on $D_i$. Under this assumption, the potential outcome can be written as $Y_i(d)$, instead of $Y_i(z,d)$. 

Let $\phi_1 = \pr(Z = 1)$, $\phi_s = \pr(S = s), s \in \{co, nt, at\}$.  Also, let $F_{co}^{(0)}(t), F_{nt}(t),  F_{co}^{(1)}(t)$, and $F_{at}(t)$ be the cumulative distribution functions of the outcome $Y$ for compliers without treatment, never-takers, compliers with treatment, and always-takers respectively. For $F_{co}^{(0)}(t)$ and $F_{co}^{(1)}(t)$, under Assumption~\ref{assump}, they are identified as the distributions of the potential outcome $Y(0)$ and $Y(1)$ respectively, for example, $F_{co}^{(0)}(t) = \pr(Y(0) \leq t \mid S=co )$. Similarly, we define the distribution functions of $Y$ corresponding to combinations of $Z, D$. Denote $F_{zd}(t)= \pr(Y \leq t \mid Z=z, D=d)$. Although $F_{Y|zd}$ can be more accurate notation than $F_{zd}$ since $Z$ and $D$ are conditioned on, we will instead use simpler notation $F_{zd}$. Any notation involving $F$ followed by a subscript means the distribution function of $Y$ conditioning on the subscript. Also, we define the probabilities $\eta_{zd} = \pr(Z=z, D=d)$ for $z, d \in \{0, 1\}$. Finally, let $H(t)= P(Y \leq t) = \sum_{z, d \in \{0,1\}} \eta_{zd} F_{zd}(t)$, be the mixture distribution of $F_{zd}$. The outcomes $Y_1, Y_2, \ldots, Y_n$ are independent and identically distributed from $H(t)$. Under Assumption~\ref{assump}, as discussed in \citet{abadie2002}, both $F_{co}^{(0)}(t)$ and $F_{co}^{(1)}(t)$ can be identified as  
\begin{equation}
F_{co}^{(0)}(t)= \frac{(\phi_{co}+\phi_{nt})F_{00}(t) - \phi_{nt} F_{10}(t)}{\phi_{co}}, \quad F_{co}^{(1)}(t)= \frac{(\phi_{co}+\phi_{at})F_{11}(t) - \phi_{at} F_{01}(t)}{\phi_{co}}.
\label{eqn:abadie}
\end{equation}
Also, $F_{nt}(t)$ and $F_{at}(t)$ can be identified under Assumption~\ref{assump} as $F_{nt}(t)=F_{01}(t)$ and $F_{at}(t)=F_{01}(t)$.

\subsection{Testing Fisher's Null Hypothesis of No Effect: Review of the Existing Approaches}
\label{subsec:review}

A central question in causal inference is to understand if the treatment has any causal effect on the outcome. To evaluate the treatment effect on the outcome, Fisher's sharp hypothesis of no effect can be considered. Under Assumption~\ref{assump}, it can be tested whether there is any causal treatment effect for compliers. Technically, Fisher's hypothesis can be constructed for compliers as $H_0^{\text{compliers}}: Y_i (1) = Y_i (0)$ for $S_i = co$. However, $H_0^{\text{compliers}}$ cannot be directly tested since only one of the two potential outcomes for each individual can be observed. Instead, we consider a test for equality of distributions using the potential outcome distributions for compliers,
\begin{equation}
H_0: F_{co}^{(0)}(t) = F_{co}^{(1)}(t), \text{ for all } t \in \mathbb{R}.
\label{null_hypothesis}
\end{equation}

The existing approach for testing $H_0$ is based on the fact that $ F_{co}^{(0)}(t) = F_{co}^{(1)}(t)$ implies $F_0(t)=F_1(t)$ where $F_z(t) = \pr(Y \leq t \mid Z=z)$ under Assumption~\ref{assump}. \citet{abadie2002} proposed using the Kolmogorov--Smirnov test $T_{\mathrm{KS}} = \sup_{t \in \R} | \overline F_0(t) - \overline F_1(t)|$, where $\overline F_z(t)={\sum_{i=1}^n \bm 1\{Y_i \leq t, Z_i=z\}}/{\sum_{i=1}^n \bm 1\{Z_i=z\}}$ are the empirical distribution functions for $z=0,1$. This test is the comparison between the outcome distribution of the $Z=0$ group and the outcome distribution of the $Z=1$ group. To show a connection between these two distributions with the compliers' distribution functions $F_{co}^{(0)}(t)$ and $F_{co}^{(1)}(t)$, define the {\it plug-in} estimates obtained using \eqref{eqn:abadie}  as 
\begin{align}
	\breve F_{co}^{(0)}(t)= \frac{(\breve{\phi}_{co}+\breve{\phi}_{nt})\overline F_{00}(t) - \breve{\phi}_{nt}\overline F_{10}(t)}{\breve{\phi}_{co}} ,~  \breve F_{co}^{(1)}(t) = \frac{(\breve{\phi}_{co}+\breve{\phi}_{at})\overline F_{11}(t) - \breve{\phi}_{at} \overline F_{01}(t)}{\breve{\phi}_{co}},
	\label{eqn:plugin}
\end{align}
where $n_{zd}= \sum_{i=1}^{n} \bm 1 \{ Z_i=z, D_i=d\}$, $\breve{\phi}_{nt}={n_{10}}/{(n_{10}+n_{11})}$, $\breve{\phi}_{at}= {n_{01}}/{(n_{00}+n_{01})}$, $\breve{\phi}_{co}=1-\breve{\phi}_{nt}-\breve{\phi}_{at}$, and $\{\overline{F}_{zd}\}_{z, d \in \{0, 1\}}$ are the empirical distribution functions, $\overline{F}_{zd}(t) = (1/n_{zd}) \sum_{i=1}^{n} \bm 1 \{ Y_i \leq t, Z_i=z, D_i=d\}$. Since $|\overline{F}_0(t) - \overline{F}_1(t)| =  |(\breve{F}_{co}^{(0)}(t) - \breve{F}_{co}^{(1)}(t))/\breve{\phi}_{co}|$, $T_{KS}$ is equivalent to the test based on comparison between $\breve{F}_{co}^{(0)}$ and $\breve{F}_{co}^{(1)}$. However, the plug-in estimates have two limitations: (1) violating the non-decreasing condition of distribution functions and (2) being unstable when an IV is weak. First, the violation leads to producing estimates that are often located outside of $[0,1]$. Therefore, they are not proper estimates of $F_{co}^{(0)}$ and $F_{co}^{(1)}$, which is due to not incorporating the observed information on compliance behavior. Furthermore, the test statistic $T_{KS}$ can be unusually large if $\breve{\phi}_{co}$ is small, which occurs when an IV is weak. As discussed in \citet{rubin1998more}, making use of the IV structure can produce a better test statistic, and thus increase power. 

To employ the structure of the instrumental variable model, one simple way is to transform the plug-in estimates to proper distribution functions by using the monotone rearrangement method \citep{chernozhukov2010}, followed by truncation to $[0, 1]$. \citet{chernozhukov2010} showed that the transformed estimates have the same first-order properties (asymptotic distribution) as the plug-in estimates. The rearrangement method produces a quick fix of the plug-in estimates $(\breve{F}_{co}^{(0)}, \breve{F}_{co}^{(1)})$ and provides promising empirical properties, however it is difficult to use the method in hypothesis testing for evaluating the asymptotic properties. 


\begin{remark}\label{rem:Fplugin}The plug-in estimators $\breve F_{co}^{(0)}(t)$ and $\breve F_{co}^{(0)}(t)$ are obtained as \eqref{eqn:plugin}. Other plug-in estimators are $\breve F_{at}(t)=\overline F_{01}(t)$ and $\breve F_{nt}(t)=\overline F_{10}(t)$. We denote the vector of the plug-in estimators of the outcome distribution functions, as $\breve{\bm F}(t)=(\breve F_{co}^{(0)}(t), \breve F_{nt}(t),  \breve F_{co}^{(1)}(t), \breve F_{at}(t))$. We consider the plug-in estimators of the compliance classes as a function of $t$ such that $\breve{\bm \phi}(t)=(\breve \phi_{nt}(t), \breve \phi_{at}(t))$ where $\breve \phi_{nt}(t) = \breve{\phi}_{nt} = {n_{10}}/{(n_{10}+n_{11})}$ and $\breve \phi_{at}(t) =\breve{\phi}_{at} = {n_{01}}/{(n_{00}+n_{01})}$ for all $t$ with terminology slightly abused. 
\end{remark}

\section{The Binomial Likelihood (BL) Approach}
\label{sec:method}

\subsection{Constructing Binomial Likelihood with An Instrumental Variable }
\label{subsec:binomial}

Define $\bm\theta : \R\rightarrow [0, 1]^4$ such that $\bm \theta(t)$=$(\theta_{co}^{(0)}(t)$, $\theta_{nt}(t)$, $\theta_{co}^{(1)}(t)$, $\theta_{at}(t)),$ where $\theta_{co}^{(0)}$, $\theta_{nt}$, $\theta_{co}^{(1)}$, $\theta_{at}$: $\R \rightarrow [0, 1]$ are functional variables representing four different outcome distributions. Instead of using previously defined $\bm F$, we use the parameter set $\bm\theta$ to emphasize the fact that it is a variable to be estimated. Similarly, we can define the parameter $\bm\chi : \R\rightarrow [0,1]^2$ such that $\bm\chi(t) = (\chi_{nt}(t), \chi_{at}(t))$, where $\chi_{nt}(t)$ and $\chi_{at}(t)$ are functional variables representing the proportions of compliance classes. Since the true proportions $\phi_{nt}$ and $\phi_{at}$ do not depend on knots, we take the average across the knots to build estimators for them. We set knots $\bm{t}=(t_1, \ldots, t_m)$ that are the locations to evaluate BL functions later. Then, $\sum_{j=1}^{m} \chi_{nt}(t_j)/m$ and $\sum_{j=1}^{m}\chi_{at}(t_j)/m$ are the estimators of $\phi_{nt}$ and $\phi_{at}$ respectively. Also, we define $\chi_{co}(t_j) = 1-\chi_{nt}(t_j) - \chi_{at}(t_j)$, and then $\sum_{j=1}^{m}(1-\chi_{nt}(t_j) - \chi_{at}(t_j))/m$ is the estimator of $\phi_{co}$. Furthermore, we use $\chi_1$ that is the estimator of $\phi_1$. Finally, we define $\theta_{zd}(t_j)$ that is the estimator of $F_{zd}(t_j)$. For example, $\theta_{00}(t_j) = (\chi_{co}(t_j) \theta_{co}^{(0)}(t_j) + \chi_{nt}(t_j)\theta_{nt}(t_j))/(1-\chi_{at}(t_j))$, $\theta_{01}(t_j) = \theta_{at}(t_j)$, $\theta_{10}(t_j) = \theta_{nt}(t_j)$ and $\theta_{11}(t_j) = ( \chi_{co}(t_j)\theta_{co}^{(1)}(t_j)+\chi_{at}(t_j)\theta_{at}(t_j))/(1-\chi_{nt}(t_j))$.    

Denote the data $\cD_n=(\bm{Z}, \bm{D}, \bm{Y})$ where $\bm{Z} = (Z_1, \ldots, Z_n)^{T}$, $\bm{D} = (D_1, \ldots, D_n)^{T}$, $\bm{Y} = (Y_1, \ldots, Y_n)^{T}$. Also, denote the event $K_{zd}^{ij}=\{ Z_i=z, D_i=d \mid t_j\}$. The probability $\pr(K_{zd}^{ij})$ can be easily computed in terms of the variables $(\bm \chi, \chi_1)$. At each knot $t_j$, we define the {\it point-knot-specific BL function} for a data point $(Z_i, D_i, Y_i)$,
\begin{align*}
	L_{ij}(\bm\theta, \bm\chi, \chi_1| \cD_n) &= \prod_{z, d \in \{0,1\}} \pr(K_{zd}^{ij})^{\bm 1(K_{zd}^{ij})} \times \left(\theta_{zd}(t_j)^{\bm 1(Y_i \leq t_j)} (1-\theta_{zd}(t_j))^{\bm 1(Y_i > t_j)} \right).
\end{align*} 
Then, by aggregating the point-knot-specific BL functions across all data points, we can define the {\it knot-specific BL function} at knot $t_j$,
\begin{align*}
	L_j (\bm\theta, \bm\chi, \chi_1| \cD_n) &= \prod_{i=1}^{n} L_{ij}(\bm\theta, \bm\chi, \chi_1| \cD_n). 
\end{align*}
Finally, we can define the {\it BL function} by taking the geometric mean of the knot-specific BL functions across all knots, 
\begin{align}\label{eq:bl_def}
	L(\bm\theta, \bm\chi, \chi_1| \cD_n) &= \prod_{j=1}^{m} L_{j}(\bm\theta, \bm\chi, \chi_1| \cD_n)^{1/m}. 
\end{align}

The BL function depends on the choice of knots even when the data points are fixed. The knots can be given by researchers, but we propose to use all observed outcomes as knots. More specifically, we use the order statistics $Y_{(j)}$ as knots with $m=n$. This selection procedure provides an automatic way to build the BL function and avoids an arbitrary decision that may cause a favorable conclusion. The contributions of the knot-specific BL functions are, obviously, not independent on data points. Nevertheless, we pretend they are independent. To reduce such dependency, a random sample from $\bm{Y}$ can be chosen as knots in practice. Also, for a large $n$, the size of knots does not need to be $n$. A smaller set of knots can be helpful for reducing computation time. Although we choose $t_j = Y_{(j)}$, we emphasize that, for general knots $\bm{t}$, the BL function can be constructed and also the MBL estimator can be obtained. Theoretical results in the following section are derived for knots $t_j = Y_{(j)}$. As long as the distribution of knots $\bm{t}$ is the same as the distribution of $\bm{Y}$, the theoretical arguments hold.

\begin{remark} The knot-specific BL function at knot $t_j = Y_{(j)}$ for some $j$ becomes zero when any of $\theta_{zd}(Y_{(j)})$ is either 0 or 1. This occasionally occurs at the extreme order statistics. To avoid technicalities in the proofs arising from this, we define the likelihood function \eqref{eq:bl_def} over the central order statistics, that is, for $j \in I_\kappa=[\ceil{n \kappa}, \ceil{n(1-\kappa)}]$ for a small fixed constant $\kappa$. Throughout the proofs in the supplementary materials, the asymptotics will be in the regime where the sample size $n$ grows to infinity, keeping $\kappa$ fixed. We omit dependence on $\kappa$ in the BL for notational brevity. Also, in practice, to avoid computational issues, we can let the knot-specific BL values be 1 when probabilities vanish on the boundary.
\end{remark}

\subsection{The Maximum Binomial Likelihood (MBL) Method}
\label{subsec:mbl}

In Section~\ref{subsec:binomial}, we introduce a BL approach for constructing a nonparametric likelihood function. Given the BL function $L(\bm \theta, \bm \chi, \chi_1 | \cD_n)$, we propose the maximum binomial likelihood (MBL) method to obtain the estimates of $(\bm \theta, \bm \chi, \chi_1)$ by maximizing them over their parameters spaces. To this end, denote $\cP([0, 1]^\R)$ as the space of all distribution functions from $\R\rightarrow [0, 1]$. Let $\bm \vartheta_+ = \{(\theta_{co}^{(0)}, \theta_{nt}, \theta_{co}^{(1)}, \theta_{at}) : \theta_{co}^{(0)}, \theta_{nt}, \theta_{co}^{(1)}, \theta_{at} \in \cP([0, 1]^\R) \}$, and $\bm\varphi_+=\{(\chi_{nt}, \chi_{at}) : \text{ for any } t, (\chi_{nt}(t), \chi_{at}(t)) \in [0,1]^2,  0 \leq \chi_{nt}(t)+\chi_{at}(t) \leq 1\}$ be the parameter spaces for $\bm\theta$ and $\bm\chi$. 

\begin{definition}\label{definition:BLE}
	The {\it MBL} estimate $(\hat{\bm F}, \hat{\bm \phi}, \hat \phi_1)$ is defined as
	\begin{equation}\label{eq:BLE}
	(\hat{\bm F}, \hat{\bm \phi}, \hat \phi_1)=\arg \max_{\bm \theta\in \bm \vartheta_+, \bm \chi \in \bm\varphi_+, \chi_1 \in [0, 1]} L(\bm \theta, \bm\chi, \chi_1|\cD_n),
	\end{equation} 
	where $\hat {\bm F}=(\hat F_{co}^{(0)}, \hat F_{nt}, \hat F_{co}^{(1)}, \hat F_{at})$ and $\hat{\bm\phi}=(\hat{\phi}_{nt}, \hat{\phi}_{at})$ are defined at the knots $\bm{t} = (t_1, \ldots, t_m)$.
\end{definition}

\begin{remark} The complete parameter space $\bm \vartheta_+ \times \bm\varphi_+ \times [0, 1]$ of the three parameters $(\bm F, \bm \phi, \phi_1)$ will be hereafter referred to as the {\it restricted parameter space}. To ensure that \eqref{eq:BLE} is well-defined, we extend $\hat {\bm F}$ between the knots by using coordinate-wise right-continuous interpolation and extrapolation beyond the knots by 0 or 1. Also, $\sum_{j=1}^{m} \hat{\phi}_{nt}(t_j)$ and $\sum_{j=1}^{m}\hat{\phi}_{at}(t_j)$ are the estimators of $\phi_{nt}$ and $\phi_{at}$. 
\end{remark}

The full expression of the binomial log-likelihood function $\ell(\bm\theta, {\bm \chi}, \chi_1 |\cD_n) = \log L(\bm\theta, {\bm \chi}, \chi_1 |\cD_n)$ is long and unwieldy. However, we can rewrite it in a compact and instructive form, by grouping and rearranging the terms. It follows (see proof of Proposition \ref{proposition:bl} below for details) that  $\ell(\bm\theta, \bm\chi, \chi_1 | \cD_n) = \ell_{\bm{Y}, \bm{D} | \bm{Z}}(\bm\theta, \bm\chi) + \ell_{\bm{Z}}(\chi_1 )$, where $ \ell_{\bm{Z}}(\chi_1) =\frac{1}{n} \{ (n_{00} + n_{01}) \log (1-\chi_1) + $$(n_{10} + n_{11}) \log \chi_1 \}$, and 
\begin{align*}
	\ell_{\bm{Y}, \bm{D} | \bm{Z}}(\bm\theta, {\bm \chi})=\sum_{z, d \in \{0, 1\}}\ell_{zd}(\bm \theta, {\bm \chi}), 
\end{align*}
with $\ell_{zd}(\bm \theta, {\bm \chi})$, for $z, d \in \{0, 1\}$, defined as follows: 

\begin{align*}
	\ell_{00}(\bm \theta, {\bm \chi})&=\frac{1}{m}\sum_{j=1}^{m} n_{00}\left\{ \log\left(1- \chi_{at}(t_j)\right)+ J(\overline F_{00}(t_j), \theta_{00}(t_j) \right\}, \\
	\ell_{10}(\bm \theta, {\bm \chi})&= \frac{1}{m}\sum_{j=1}^{m} n_{10} \left\{\log  \chi_{nt}(t_j) + J(\overline F_{10}(t_j), \theta_{10}(t_j)) \right\}, \\
	\ell_{01}(\bm \theta, {\bm \chi})&= \frac{1}{m}\sum_{j=1}^{m} n_{01} \left\{\log  \chi_{at}(t_j) + J(\overline F_{01}(t_j), \theta_{01}(t_j)) \right\}, \\
	\ell_{11}(\bm \theta, {\bm \chi})&= \frac{1}{m}\sum_{j=1}^{m} n_{11} \left\{ \log \left(1- \chi_{nt}(t_j)\right) + J(\overline F_{11}(t_j), \theta_{11}(t_j) \right\},
\end{align*}
where the function $J(x, y)=x\log y+(1-x) \log(1-y)$. 

\begin{proposition}\label{proposition:bl} Let $(\hat{\bm F}, \hat{\bm \phi}, \hat \phi_1)$ be the binomial likelihood estimates as defined in \eqref{eq:BLE}. Then $\hat \phi_1=\frac{n_{10}+n_{11}}{n}$ that is equal to the plug-in estimate $\breve{\phi}_1$, and 
	\begin{align*} 
		(\hat{\bm F}, \hat{\bm \phi})=\arg \max_{\bm \theta\in \bm \vartheta_+, \bm \chi \in \bm\varphi_+} \ell_{\bm{Y}, \bm{D} | \bm{Z}}(\bm\theta, {\bm \chi}).
	\end{align*}
\end{proposition}
\begin{proof}
	See Section~A in the Supplementary Material. 
\end{proof}

This proposition shows that the MBL estimate of $\phi_1$ is the proportion of individuals with instrument (that is, $Z=1$) in the observed sample. Furthermore, the MBL estimates of $\bm F$ and $\bm\phi$ can be obtained by maximizing the function $\ell_{\bm{Y}, \bm{D} | \bm{Z}}(\bm\theta, {\bm \chi})$.

\begin{remark}
	Maximizing the BL function over the {\it unrestricted parameter space} $\bm \vartheta \times \bm\varphi$, where $\bm \vartheta = \{(\theta_{co}^{(0)}, \theta_{nt}, \theta_{co}^{(1)}, \theta_{at}):  \theta_{co}^{(0)}, \theta_{nt}, \theta_{co}^{(1)}, \theta_{at} \in \R^\R\}$ with $\R^\R$ the set of all functions from $\R \rightarrow \R$ and $\bm\varphi = \{(\chi_{nt}, \chi_{at}) : \text{ for any } t, (\chi_{nt}(t), \chi_{at}(t)) \in \R^2\}$, produces the plug-in estimates $(\breve {\bm F}, \breve {\bm \phi}) =\arg \max_{\bm\theta \in \bm\vartheta, \bm\chi \in \bm\varphi}  \ell_{\bm{Y}, \bm{D} | \bm{Z}}(\bm\theta, {\bm \chi})$ (see Lemma~1 in the Supplementary Material for the proof).
	\label{rmk:plugin}
\end{remark}

\subsection{Asymptotic Properties of The MBL Estimates}
\label{subsec:BLproperties}

In this section we discuss the asymptotic properties of the MBL estimates $(\hat{\bm F}, \hat{\bm\phi})$, and how they compare with the plug-in estimates $(\breve{\bm F}, \breve{\bm\phi})$. Assume the knots $t_j = Y_{(j)}$, for $1 \leq j \leq n$. 

\begin{assumption}\label{assumptionF}
	We assume the following: 
	
	\begin{enumerate}[leftmargin=1cm]
		
		\item[(a)] The proportion parameter vector $\bm \phi$ belongs to the interior of the parameter space $[0, 1]_+^2$. 
		
		\item[(b)] The distribution functions $F_{zd}$ are continuous, strictly increasing, and have the same support. 
		
		\item[(c)] For all $K \subset \R$ compact, $s, t \in K$, there exists constants $0< C_1 \leq C_2 < \infty$ (depending on $K$) such that $C_1|s-t| \leq |F_{zd}(s)-F_{zd}(t)| \leq C_2|s-t|$. 
	\end{enumerate}
\end{assumption}

In particular, Assumption \ref{assumptionF} holds whenever $F_{zd}$ are differentiable and the derivatives are uniformly bounded above and below, that is, $C_1 \leq F'_{zd}(t) \leq C_2$, for all $t \in K$, and $K \subset \R$  compact. Under this assumption we show that the MBL estimates and the plug-in estimates have mean squared errors converging to zero, after rescaling by $\sqrt n$. Recall that $H(t)=\sum_{z, d \in \{0, 1\}}\eta_{zd}F_{zd}(t)$ is the true population outcome distribution of $Y$.

\begin{theorem}\label{TH:BLII}  For any fixed $0 <\kappa <1/2$, let $I_\kappa=[\ceil{n \kappa}, \ceil{n(1-\kappa)}]$ and $J_\kappa=[H^{-1}(\kappa), H^{-1}(1-\kappa)]$. Then, the MBL estimates $(\hat{\bm F}, \hat{\bm \phi})$ and the plug-in estimates $(\breve {\bm F}, \breve {\bm \phi})$ satisfy 
	\begin{align*}
		\frac{1}{|I_{\kappa}|} \sum_{j \in I_\kappa} ||\sqrt n\{\bm{\hat{F}}(Y_{(j)})-\breve{\bm F}(Y_{(j)})\}||^2_2=o_P(1),
	\end{align*}
	and 
	\begin{align}\label{eq:BLestdiffintegral}
		\int_{J_\kappa} ||\sqrt n\{\bm{\hat{F}}(t)-\breve {\bm F}(t)\}||^2_2 {\rm d} H= o_P(1),
	\end{align}
	where the $o_P(1)$ term goes to zero as $n\rightarrow \infty$.  Moreover, $\frac{1}{|I_{\kappa}|} \sum_{j \in I_\kappa} ||\sqrt n\{ \hat{\bm\phi}(Y_{(j)}) -\bm {\breve \phi}(Y_{(j)}) \}||_2^2 =o_P(1)$. Also, it implies that the two estimators $\frac{1}{|I_{\kappa}|} \sum_{j \in I_\kappa} \hat{\bm\phi}_{s}(Y_{(j)})$ and $\frac{1}{|I_{\kappa}|} \sum_{j \in I_\kappa} \breve{\bm\phi}_{s}(Y_{(j)})$ of the population $\phi_{s}$ for $s \in \{nt, at\}$ satisfy 
	\begin{align*}
		\sqrt n \left(\frac{1}{|I_{\kappa}|} \sum_{j \in I_\kappa} \hat{\bm\phi}_{s}(Y_{(j)}) - \frac{1}{|I_{\kappa}|} \sum_{j \in I_\kappa} \breve{\bm\phi}_{s}(Y_{(j)}) \right) = o_P(1)
	\end{align*}
\end{theorem}
\begin{proof}
	See Section~B in the Supplementary Material. 
\end{proof}

The theorem shows that $\hat{\bm F}$ and $\breve{\bm F}$ (also, $\hat {\bm\phi}$  and $\breve{\bm\phi}$) have the same first-order behavior, and hence the same limiting distribution, which can be derived using the Brownian bridge approximation of the empirical distribution functions; see Corollary~1 (Section C) in the Supplementary Material. Interestingly, \citet{chernozhukov2010} showed the monotone rearrangement estimates also have the same first-order behavior as the plug-in estimates, which together with Theorem \ref{TH:BLII}, implies that the MBL estimates have the same first-order properties as the rearrangement estimates.

\begin{remark} The proof of Theorem \ref{TH:BLII} can be easily modified to show finite dimensional convergence, that is, for every $s\geq 1$ and given $ t_1 < t_2 <\cdots < t_s $,  $||\sqrt n(\hat{\bm F}(t_j)-\breve {\bm F}(t_j))||^2_2 = o_P(1)$. This would imply that the finite dimensional distributions of the plug-in estimate process $\sqrt n(\breve{\bm F}(t)-\bm F(t))$ and the MBL estimate process $\sqrt n(\hat{\bm F} (t)-\bm F(t))$ are asymptotically the same. We present this result in terms of mean squared errors as in \eqref{eq:BLestdiffintegral}, because it emerges naturally from the asymptotic properties of the BL function, and can be directly applied to the analysis of the BLRT that is introduced in Section \ref{sec:blrt}.
\end{remark}

\section{Extension of the BL Approach: Hypothesis Testing}
\label{sec:blrt}

\subsection{Binomial Likelihood Ratio Test (BLRT): Full Version}
\label{subsec:blrt_full}

The BL approach can be extended to constructing a likelihood ratio-type test in a similar way that the ML approach can be extended to constructing a likelihood ratio test. We take two times the difference in two binomial log-likelihood values; one is obtained with the constraint $F_{co}^{(0)}(t)=F_{co}^{(1)}(t)$ (that is, under the null) and the other is obtained without this constraint (that is, under the alternative). This gives a new test for the null hypothesis $H_0: F_{co}^{(0)}(t)=F_{co}^{(1)}(t)$, and hereafter, we call it the {\it binomial likelihood ratio test (BLRT)}. 

Define the {\it restricted null parameter space} as $\bm \vartheta_{+, 0} = \{(\theta_{co}, \theta_{nt}, \theta_{at}):  \theta_{co}, \theta_{nt}, \theta_{at} \in \cP([0, 1]^\R) \}$, where $\cP([0, 1]^\R)$ is the set of distribution functions from $\R\rightarrow [0, 1]$. Then, the BLRT statistic is obtained by 
\begin{equation}\label{eq:blrt}
T_n= 2 \left( \max_{\bm \theta \in \bm \vartheta_{+}, \bm\chi \in \bm\varphi_+ }  \ell_{\bm{Y}, \bm{D} | \bm{Z}}(\bm\theta, {\bm \chi})-\max_{\bm \theta \in \bm \vartheta_{+, 0}, \bm\chi \in \bm\varphi_+ } \ell_{\bm{Y}, \bm{D} | \bm{Z}}(\bm\theta, {\bm \chi}) \right). 
\end{equation}

Let $(\hat{\bm \psi}, \hat{\bm\xi}) = \arg\max_{\bm \theta \in \bm \vartheta_{+, 0}, \bm\chi \in \bm\varphi_+}  \ell_{\bm{Y}, \bm{D} | \bm{Z}}(\bm\theta, {\bm \chi})$. The asymptotic properties of $\hat{\bm \psi}$ can be derived as we did for $\hat{\bm F}$ in Theorem~\ref{TH:BLII}. Also, we can derive the asymptotically equivalent plugin-type estimators that have not been studied before. It is worth noting that the explicit form of the equivalent estimators of $(\hat{\bm \psi}, \hat{\bm\xi})$ is provided in Section~D, the Supplementary Material.  

\begin{theorem}\label{TH:BLRT} Fix $0 < \kappa < 1/2$ and recall that $H(t)=P(Y \leq t)$. Let $T_n$ be the binomial likelihood ratio test statistic as defined in \eqref{eq:blrt}. Denote $\overline{J}_{\kappa} = [\overline{H}^{-1}(\kappa), \overline{H}^{-1}(1-\kappa)]$. Then, 
	\begin{align*}
		T_n &= \frac{n_0 n_1}{n} \int_{\overline{J}_{\kappa}} \frac{\bar{F}_{0}(t) - \bar{F}_{1}(t)}{\bar{H}(t) (1-\bar{H}(t))} {\rm d} \bar{H}(t) +o_P(1), 
	\end{align*} 
	where $\bar{H}(t) = (n_0 \bar{F}_0(t) + n_1 \bar{F}_1(t))/n$ is the empirical distribution function of $Y$. 
\end{theorem}
\begin{proof}
	See Section~E in the Supplementary Material. 
\end{proof}

This theorem gives an asymptotically equivalent representation of the BLRT statistic as the two-sample Anderson-Darling test statistic \citep{pettitt1976}. It can be used to construct the rejection region and compute the critical value for a given significance level. Moreover, this shows that the test based on $T_n$ is consistent against all fixed alternatives, because of the universal consistency of the two-sample Anderson-Darling test \citep{scholz1987}. 

However, in finite-sample settings, the critical value obtained from the asymptotic distribution of $T_n$ can be conservative. In the theorem above, to derive the asymptotic properties, we use the equivalent plug-in estimators of $(\hat{\bm \psi}, \hat{\bm\xi})$ instead of using $(\hat{\bm \psi}, \hat{\bm\xi})$ directly. However, they do not lie in the restricted parameter space, which leads to a gap between the equivalent and actual BL values. This gap fades out as $n$ increases, but it can be critical when we evaluate finite-sample performance. This issue will be further discussed in the simulation section. 

The BLRT is developed for testing the null hypothesis $H_0: F_{co}^{(0)}(t) = F_{co}^{(1)}(t)$ that assumes no treatment effect for compliers. The BLRT can be further extended to testing other hypotheses like $H_0^{g}: F_{co}^{(0)}(g(t)) = F_{co}^{(1)}(t)$ for some $g$. To test $H_0^{g}$, a simple modification is required. A new outcome variable $Y_i^*$ can be generated: $Y_i^* = g(Y_i)$ if $D_i=1$, and $Y_i^* = Y_i$ otherwise. Then, $(Z_i, D_i, Y_i^*)$ can be used for the BLRT as \eqref{eq:blrt}, and this test based on the new dataset conducts a hypothesis test for $H_0^{g}$. Among many choices of $g$, $g(t) = t- \mu$ can be considered to check whether there is any location shift between the two distributions. The assumption of the location shift means that there is a constant treatment effect $\mu$ for all compliers.  Therefore, this test can be used for examining treatment effect heterogeneity. If $H_0^{location}: F_{co}^{(0)}(t-\mu) = F_{co}^{(1)}(t)$ is rejected for all $\mu$, then there is evidence that treatment effects are heterogeneous.   

\subsection{Binomial Likelihood Ratio Test: Simple Version}
\label{subsec:simple}

As we discussed in section~\ref{subsec:review}, under Assumption~\ref{assump}, testing the null hypothesis $H_0: F_{co}^{(0)}(t) = F_{co}^{(1)}(t)$ is equivalent to testing the null hypothesis $H_0^{simple}: F_0(t) = F_1(t)$. Based on this, we propose a simple version of the BLRT by comparing $ F_0(t) = F_1(t)$ instead of $F_{co}^{(0)}(t) = F_{co}^{(1)}(t)$. This test does not use the information of compliance classes by ignoring the treatment $\bm{D}$, but uses only $\bm{Z}$ and $\bm{Y}$. We define the parameter $\bm\theta$ such that $\bm\theta(t) = (\theta_0(t), \theta_1(t))$, where $\theta_0(t), \theta_1(t): \R \to [0,1]$ are functional variables representing the outcome distributions for $Y|Z=0$ and $Y|Z=1$. Then, given that $\bm{Z}$ is conditioned on, the {\it simple version binomial log-likelihood function} $\ell^{simple}_{\bm{Y}|\bm{Z}}(\bm\theta)$ is 
\begin{align}\label{bl:simple_version}
	\ell^{simple}_{\bm{Y}|\bm{Z}}(\bm\theta) &= \frac{1}{m}\sum_{i=1}^{n} \sum_{j=1}^{m} \bm 1(Z_i=0, Y_i \leq t_j) \log\theta_0(t_j)
	+ \bm 1(Z_i=0, Y_i \leq t_j) \log(1-\theta_0(t_j)) \nonumber\\
	&\quad\quad\quad\quad +\bm 1(Z_i=1, Y_i \leq t_j) \log\theta_1(t_j)
	+ \bm 1(Z_i=1, Y_i \leq t_j) \log(1-\theta_1(t_j))\nonumber\\
	&= \frac{1}{m} \sum_{j=1}^{m} n_0  J(\overline{F}_0(t_j), \theta_0(t_j)) + n_1  J(\overline{F}_1(t_j), \theta_1(t_j)),
\end{align}
where $J(x,y) = x \log y + (1-x)\log(1-y)$. The {\it simple version BLRT statistic} $T_n^{simple}$ is defined as
$$
T_n^{simple} = 2 \left(\max_{\theta_0, \theta_1 \in \cP([0, 1]^\R)}\ell^{simple}_{\bm{Y}|\bm{Z}}(\bm\theta) - \max_{\theta_0=\theta_1 \in \cP([0, 1]^\R)}\ell^{simple}_{\bm{Y}|\bm{Z}}(\bm\theta)   \right).
$$ 
Since this test does not use any information of the compliance class behaviors, it does not require estimation of the proportions, and estimation of the outcome distribution for the compliance classes. The following gives the asymptotic approximation of $T_n^{simple}$. 

\begin{theorem}\label{TH:BLRT_simple}
	The test statistic $T_n^{simple}$ has an explicit form as
	$$
	T_n^{simple} = 2 \left( \ell^{simple}_{\bm{Y}|\bm{Z}}(\overline{F}_0, \overline{F}_1)  - \ell^{simple}_{\bm{Y}|\bm{Z}}(\overline{H}, \overline{H})\right).
	$$
	If we assume that the knots are $\bm{t} = (Y_{(1)}, \ldots, Y_{(n)})$, then the test statistic $T_n^{simple}$ is asymptotically equivalent to the two-sample Anderson-Darling test statistic, that is, 
	$$
	T_n^{simple} = \frac{n_0 n_1}{n} \int_{\overline{J}_{\kappa}} \frac{\bar{F}_{0}(t) - \bar{F}_{1}(t)}{\bar{H}(t) (1-\bar{H}(t))} {\rm d} \bar{H}(t) +o_P(1). 
	$$
\end{theorem}

\begin{proof}
	See Section~E in the Supplementary Material. 
\end{proof}

Theorem~\ref{TH:BLRT_simple} shows that, as in the case of the full version BLRT, the simple version BLRT is asymptotically equivalent to the two-sample Anderson-Darling test, and, therefore, is consistent against all fixed alternatives, as well. The difference is that $T_n^{simple}$ has a closed form and does not need the EM-PAV algorithm that will be introduced in the next section. However, $T_n^{simple}$ does not involve any estimation procedure of outcome distributions for compliance classes, and, hence, cannot be applied for estimation purposes.

\section{Computation and Simulation}
\label{SEC:ALGO}

\subsection{EM-PAV Algorithm for Computing the MBL Estimates}

There are no closed form solutions to the MBL estimates. However, the estimates can be computed efficiently by using a combination of the expectation-maximization (EM) algorithm and the pool-adjacent-violator(PAV) algorithm. We call it the \textit{EM-PAV} algorithm. 
To begin with, we introduce the {\it complete-data} $\overline \cD_n$, which includes the compliance class $\bm S$, $\overline \cD_n=(\bm Z, \bm S, \bm D, \bm Y)^{\T}$. If $Z_i$ and $S_i$ are known, then $D_i$ is determined; for example, if $Z_i=0$ and $S_i=co$, then $D_i=0$. Denote the event $\overline{K}_{zs}^{ij}=\{ Z_i=z, S_i=s \mid t_j\}$, where $s\in \{co, at, nt\}$. 

Given the complete data, we can define a {\it point-knot-specific complete-data binomial likelihood function} for the data point $(Z_i, S_i, D_i, Y_i)$ at knot $t_j$, 
\begin{align*}
	\overline{L}_{ij}(\bm\theta, \bm\chi, \chi_1 | \overline{\cD}_n) = \prod_{z \in \{0,1\}} \prod_{ s \in \{co, nt, at\}} \pr(\overline{K}_{zs}^{ij})^{\bm 1(\overline{K}_{zs}^{ij})} \times \left(\theta_{s}(t_j)^{\bm 1(Y_i \leq t_j)} (1-\theta_{s}(t_j))^{\bm 1(Y_i > t_j)}\right).
\end{align*}
The {\it complete-data binomial likelihood} is obtained by combining all point-knot-specific complete-data likelihood functions in the same way to define the BL, 
\begin{align*}
	\overline L(\bm \theta, \bm \chi, \chi_1 \mid \overline \cD_n)= \prod_{j=1}^{m} \left\{ \prod_{i=1}^{n}  \overline{L}_{ij}(\bm\theta, \bm\chi, \chi_1 | \overline{\cD}_n) \right\}^{1/m}.
\end{align*}

As in the BL, the dependence on $\chi_1$ in the complete-data likelihood is separable, that is, 
\begin{equation*}\label{eq:cbinom}
	\log \overline L(\bm \theta, \bm \chi, \chi_1 \mid \overline \cD_n)=\ell(\chi_1)+\log \overline L(\bm \theta, \bm \chi \mid \overline \cD_n)
\end{equation*}
where $\ell(\chi_1)=(n_{00}+n_{01}) \log (1-\chi_1) +  (n_{10}+n_{11}) \log \chi_1 $, and $\log \overline L(\bm \theta, \bm \chi \mid \overline \cD_n)$ does not depend on $\chi_1$. Hereafter, we will refer to $\log \overline L(\bm \theta, \bm \chi \mid \overline \cD_n)$ as the {\it complete-data binomial log-likelihood}. To find the maximizer of $\log \overline L(\bm \theta, \bm \chi \mid \overline \cD_n)$, our algorithm is initiated by specifying the initial values $(\bm\theta_{(0)}, \bm\chi_{(0)})$ that lie in the parameter space $\bm\vartheta_+ \times \bm\varphi_+$. Then the following steps are repeated until the values converge: 

\begin{algo}{(EM-PAV algorithm)} 
	Let $\bm{\hat \theta}_{(k)}=(\bm{\hat\theta}_{co, (k)}^{(0)}, \bm{\hat\theta}_{nt, (k)}, \bm{\hat\theta}_{co, (k)}^{(1)}, \bm{\hat\theta}_{at, (k)})$ and $\bm{\hat \chi}_{(k)}=(\hat{\chi}_{nt, (k)}, \hat\chi_{at, (k)})$ be the outputs after the $k$th step of the iteration. The following shows the $(k+1)$th step.
	
	({\it Expectation Step}) Given these outputs $(\bm{\hat \theta}_{(k)}, \bm{\hat \chi}_{(k)})$and the observed data $\cD_n$, the expected complete-data binomial log-likelihood is
	\begin{align*}
		Q_k(\bm\theta,  \bm\chi|\bm{\hat \theta}_{(k)},  \bm{\hat \chi}_{(k)})= E_{\bm{\hat \theta}_{(k)},  \bm{\hat \chi}_{(k)}}\left(\log \overline L(\bm \theta, \bm \chi|\overline \cD_n) \, | \,  \cD_n \right).
	\end{align*} 
	The expectation can be easily calculated; see Section~F2 in the Supplementary Material for computational details.
	
	({\it Maximization Step}) To begin with, define 
	\begin{align*}
		(\breve{\bm\theta}_{(k+1)}, \breve{\bm \chi}_{(k+1)})=\arg\max_{\bm\theta \in \bm \vartheta,  \bm\chi \in \bm\varphi}Q_k(\bm\theta,  \bm\chi|\hat{\bm\theta}_{(k)},  \hat{\bm \chi}_{(k)}).
	\end{align*}
	Note that $
	\breve{\bm\theta}_{(k+1)}=
	(\breve{\bm \theta}_{co, (k+1)}^{(0)} \breve{\bm \theta}_{nt, (k+1)}, \breve{\bm \theta}_{co, (k+1)}^{(1)}, \breve{\bm \theta}_{at, (k+1)}),$ where $\breve{\bm \theta}_{co, (k+1)}^{(0)}$ is evaluated at knots $Y_{(j)}$, and similarly for other estimates. Observe that $(\breve{\bm\theta}_{(k+1)}, \breve{\bm \chi}_{(k+1)})$ is the unrestricted maximizer of $Q_k(\bm\theta,  \bm\chi|\hat{\bm\theta}_{(k)},  \hat{\bm \chi}_{(k)})$. 
	These estimates can be computed explicitly; see Section~F3 in the Supplementary Material. It can be shown that $\breve{\bm \chi}_{(k+1)}=(\breve \chi_{nt, (k+1)} , \breve \chi_{at, (k+1)} )$ is actually in the restricted space $\bm\varphi_+$, that is, $\breve \chi_{nt, (k+1)}(t_j) , \breve \chi_{at, (k+1)}(t_j) \in [0, 1]$ and $0 \leq \breve\chi_{nt, (k+1)}(t_j) + \breve\chi_{at, (k+1)}(t_j) \leq 1$ for any knot $t_j$. Define $\hat{\bm \chi}_{(k+1)}=\breve{\bm \chi}_{(k+1)}$. In general, however, $\breve{\bm\theta}_{(k+1)} \notin \bm \vartheta_+$, because $\breve{\bm\theta}_{(k+1)}$ may not satisfy the non-decreasing condition of distribution functions. To ensure the monotonicity constraint we apply the PAV algorithm to the estimate $\breve{\bm\theta}_{(k+1)}$,
	
	\begin{align*}
		\hat{\bm\theta}_{(k+1)} = \text{PAV}_{\bm w}(\breve{\bm \theta}_{co, (k+1)}^{(0)}, \breve{\bm \theta}_{nt, (k+1)} , \breve{\bm \theta}_{co, (k+1)}^{(1)}, \breve{\bm \theta}_{at, (k+1)}),
	\end{align*}
	where the operation PAV$_{\bm w}$ is applied coordinate-wise and the weight vector is $\bm w_{(k+1)} = (\bm w_{co, (k+1)}^{(0)},$ $\bm w_{nt,  (k+1)}, \bm w_{co,  (k+1)}^{(1)}, \bm w_{at,  (k+1)})$, where the weights are defined in Section~F in the Supplementary Material. 
	\label{algo:empava}
\end{algo}

The following proposition establishes the correctness of this algorithm. 

\begin{proposition}\label{ppn1} Let $\hat{\bm\theta}_{(k+1)}, \hat{\bm \chi}_{(k+1)}$ be as defined above. Then $$(\hat{\bm\theta}_{(k+1)}, \hat{\bm \chi}_{(k+1)})=\arg\max_{\bm\theta \in \bm \vartheta_+,  \bm\chi \in \bm\varphi_+}Q_k(\bm\theta,  \bm\chi|\hat{\bm\theta}_{(k)},  \hat{\bm \chi}_{(k)}).$$
\end{proposition}

\begin{proof}
	See Section~F in the Supplementary Material. 
\end{proof}

\subsection{Simulation: Performance of BLRT}
\label{subsec:sim_blrt}

\begin{table}
	\begin{center}
		\caption{\label{power_table1} Size and power of the different tests with a significance level 0.05 }
		\begin{tabular}{lllccccc}
			&&&\multicolumn{5}{c}{$N(-\mu, 1)$ vs. $N(\mu, 1)$} \\
			$(\mu_{nt}, \mu_{at})$ &IV & $\mu$ & $T_n$ (boot.) & $T_n$ (asymp.) & $T_n^{simple}$ & $T_{AD}$ & $T_{KS}$ \\
			$(-1, 1)$ & Strong  & 0  & 0.054 & 0.042  & 0.056 & 0.047 & 0.047\\
			& & 0.3 & 0.313 & 0.277 & 0.270 & 0.287 & 0.249 \\
			& & 0.6 & 0.839 & 0.802 & 0.796 & 0.777 & 0.779 \\
			& & 0.9 & 0.983 & 0.968 & 0.963 & 0.971 & 0.954 \\[0.2cm]
			&Weak  & 0 & 0.044 & 0.027 & 0.053 & 0.052 & 0.050\\
			& & 0.3 & 0.200 & 0.149 & 0.146 & 0.120 & 0.134\\
			& & 0.6 & 0.486 & 0.393 & 0.385 & 0.356 & 0.362\\
			& & 0.9 & 0.721 & 0.637 & 0.624 & 0.646 & 0.584\\[0.2cm]
			$(-2, 2)$ & Strong  & 0 & 0.041 & 0.022 & 0.047 & 0.050 & 0.046\\
			& & 0.3 & 0.258 & 0.188 & 0.171 & 0.170 & 0.237 \\
			& & 0.6 & 0.668 & 0.590 & 0.536 & 0.529 & 0.663 \\
			& & 0.9 & 0.918 & 0.877 & 0.841 & 0.831 & 0.928\\[0.2cm]
			& Weak  & 0 & 0.045 & 0.026 & 0.042 & 0.043 & 0.033\\
			& & 0.3 & 0.142 & 0.095 & 0.080 & 0.081 & 0.101\\
			& & 0.6 & 0.319 & 0.236 & 0.177 & 0.185 & 0.267\\
			& & 0.9 & 0.504 & 0.423 & 0.336 & 0.339 & 0.457\\
		\end{tabular}
	\end{center}
\end{table}

To assess the performance of the two versions of the proposed BLRT, we compare them to the Kolmogorov--Smirnov test with $T_{KS}$ in a simulation study. Note that both $T_n^{simple}$ and $T_{KS}$ do not use the variable $\bm D$, but use $\bm Z$ and $\bm Y$. The null distribution of $T_{KS}$ can be obtained by permuting $\bm Z$ multiple times while the Anderson-Darling distribution $A^2$ is used as the asymptotic null distribution of $T_{n}^{simple} $. The distribution $A^2$ is the limiting distribution of the two-sample Anderson-Darling test statistic $T_{AD}$ \citep{pettitt1976}. Also, $A^2$ is used as the limiting distribution of the full version BLRT $T_n$ as well.  

In the simulation study, assume that all four potential outcome distributions are normal distributions with variance 1, but with different means:  $F_{co}^{(0)} \sim N(\mu_{co}^{(0)}, 1)$, $F_{co}^{(1)} \sim N(\mu_{co}^{(1)}, 1)$, $F_{nt} \sim N(\mu_{nt}, 1)$ and $F_{at} \sim N(\mu_{at}, 1)$. Two simulation factors are considered: (1) how far the distributions are from each other, (2) how strong the IV is. To see the impact of the first factor, we consider two simulation settings with $(\mu_{nt}, \mu_{at})=(-1,1)$ (close) and $(\mu_{nt}, \mu_{at})=(-2,2)$ (far). We evaluate these settings with $(\mu_{co}^{(0)}, \mu_{co}^{(1)}) = (-\mu, \mu)$ for various $\mu$ values. In addition, to assess the second factor, we consider the weak IV setting with the proportions $(\phi_{co}, \phi_{nt}, \phi_{at}) = (0.2, 0.4, 0.4)$ and the strong IV setting with $(\phi_{co}, \phi_{nt}, \phi_{at}) = (1/3, 1/3, 1/3)$. We consider four simulations settings, and, in each simulation setting, various values of $\mu$ and $n$ are considered.

\begin{figure}
	\centering
	\includegraphics[width=150mm]{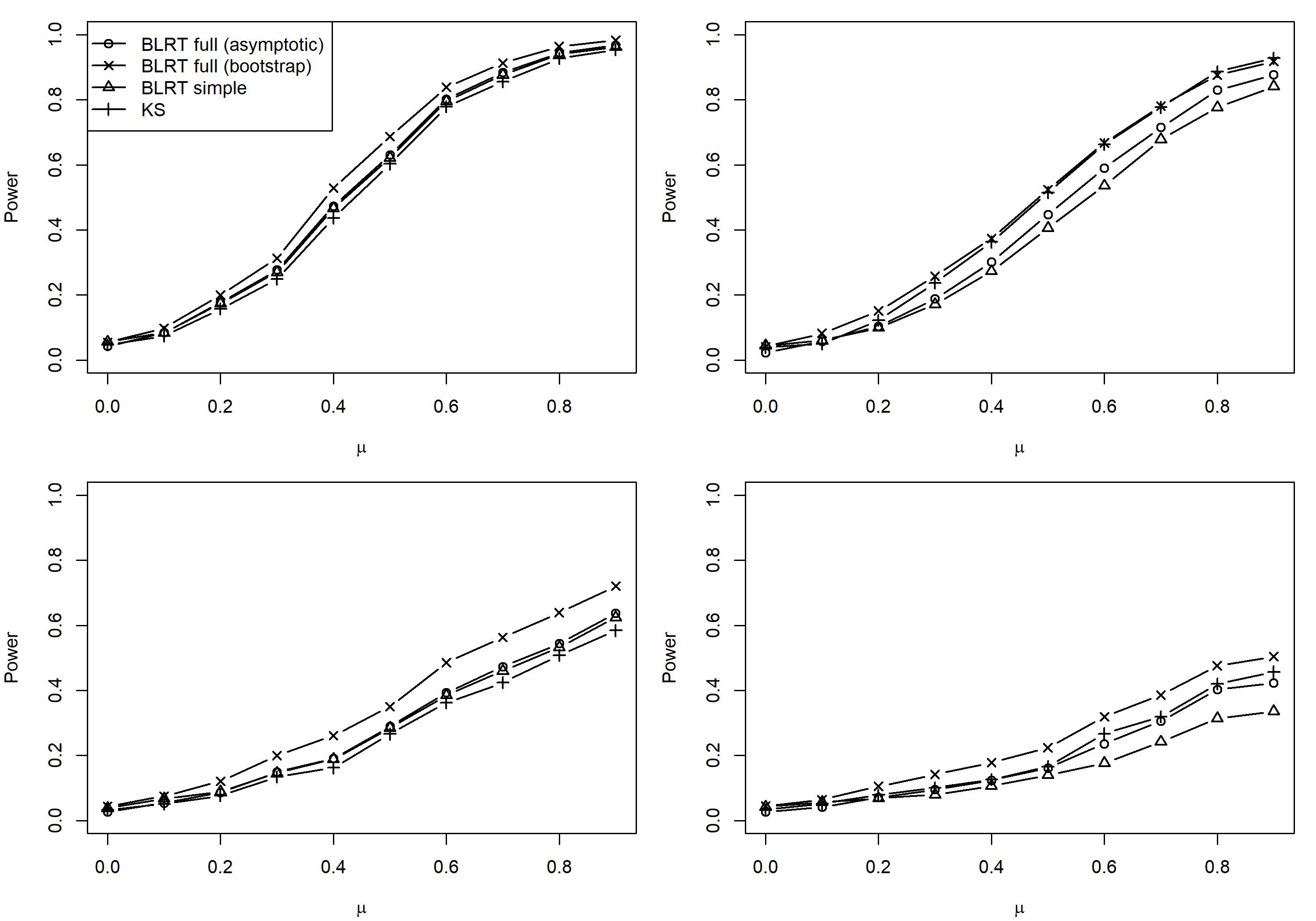}
	\caption{\small{Power of $T_n$, $T_n^{simple}$ and $T_{KS}$. Power is calculated given a significance level $\alpha=0.05$. Upper left: $(\mu_{nt}, \mu_{at})=(-1,1)$ and strong IV, Upper right:  $(\mu_{nt}, \mu_{at})=(-2,2)$ and strong IV, Lower left: $(\mu_{nt}, \mu_{at})=(-1,1)$ and weak IV, lower right:  $(\mu_{nt}, \mu_{at})=(-2,2)$ and weak IV.}}
	\label{power_plot1}
\end{figure}

Table~\ref{power_table1} shows estimated size and power of $T_n$, $T_n^{simple}$, $T_{AD}$ and $T_{KS}$ from 1000 simulated datasets. This table reports the four simulation settings for $n=300$ and $\mu=(0, 0.3, 0.6, 0.9)$. Other values of $\mu$ are not reported in this table, but are plotted in Figure~\ref{power_plot1}. More simulation results are reported in the Supplementary Materials (Section G). The first row of each simulation setting shows the simulated size. For power comparisons, one of the main findings is that the shape difference between $F_0(t)$ and $F_1(t)$ is important for the performances of the tests. When $(\mu_{nt}, \mu_{at}) = (-1,1)$, $F_0(t)$ and $F_1(t)$ differ at tails, and $T_n$ and $T_n^{simple}$ outperforms $T_{KS}$. However, when $(\mu_{nt}, \mu_{at}) = (-2, 2)$, $F_0(t)$ and $F_1(t)$ differ mostly at the middle, and $T_{KS}$ outperforms the others. For another finding, when an IV is weak, meaning that $\phi_{co}$ is small, power is reduced, but at the same time, the shape difference between $F_0(t)$ and $F_1(t)$ is less centered since $F_{nt}(t)$ and $F_{at}(t)$ dominate the shape. Therefore, as $\phi_{co}$ decreases, $T_n$ can capture the distributional difference more and produce better performance than $T_{KS}$. This can be found in Figure~\ref{power_plot1} by comparing the upper right plot and lower right plot; the asymptotic $T_n$ is less powerful than $T_{KS}$ when an IV is strong, but they have the almost same power in the weak IV setting. In summary, when the difference of the distributions in two samples is not concentrated in the middle, $T_n$ and $T_n^{simple}$ can be powerful.

As we pointed out in the previous section, the simulation results indicate that using the limiting distribution $A^2$ for $T_n$ is conservative. The simulated sizes do not reach the nominal level 0.05. We conducted additional simulations for various $n$ values when there is no effect at all. We conducted simulations for different sample sizes $n= (500, 1000, 1500, 2000)$ and the estimated sizes from 10,000 simulations for each $n$ are $(0.027, 0.029, 0.032, 0.034)$ when $(\mu_{nt}, \mu_{at})=(-2,2)$. As we expected, the size approaches to the correct nominal level $\alpha=0.05$ as $n$ increases. However, the convergence for $T_{n}$ is not satisfactory for a moderately large $n$. This conservativeness essentially lowers the performance of $T_n$ in finite samples. 

To boost the finite-sample performance, we can consider the bootstrapping method that simulates the true null distribution of $T_n$ under the null hypothesis for given $n$. Bootstrapping can be done using the estimates $\hat{\bm{\psi}}$ and $\hat{\bm{\xi}}$ that are obtained under the null hypothesis. For the $b$th procedure of bootstrapping, first fix $\bm{Z}$ and sample the compliance class membership $\bm{S}^{(b)}$ using  $\hat{\bm\xi}$. Second, determine $\bm{D}^{(b)}$ based on $\bm{Z}$ and $\bm{S}^{(b)}$; for instance, if $Z_i=0$ and $S_i^{(b)}=co$, then $D_i^{(b)}=0$. Third, take a sample $\bm{Y}^{(b)}$ based on $\bm{Z}$ and $\bm{S}^{(b)}$ using the estimate $\hat{\bm \psi}$. Finally, repeat the entire process for $1 \leq b \leq B$ to obtain the bootstrapped samples $\{(\bm Z, \bm D^{(b)}, \bm Y^{(b)})\}_{1 \leq b \leq B}$. Table~\ref{power_table1} reports simulated size and power based on $B=1000$  bootstrapped samples for each simulated dataset. The column of $T_n$ (boot.) in Table~\ref{power_table1} shows the estimated size and power from the bootstrap procedure. All the values are improved from the asymptotic-version values. The bootstrap-version $T_n$ can reduce the performance gap in cases where $T_{KS}$ is superior, and in some cases, can overtake $T_{KS}$. 

\subsection{Simulation: Performance of the MBL Method}
\label{subsec:sim_mbl}

In this section, we evaluate the MBL estimates by comparing it with the plug-in estimates \eqref{eqn:plugin} and the estimates obtained from the rearrangement method proposed by \citet{chernozhukov2010}. 

We consider the four situations in simulation studies. In the first three situations, all distributions are Gamma distributions: $F_{co}^{(0)}=F_{co}^{(1)} \sim Gamma(1.2^2, 1), F_{nt} \sim Gamma(1^2, 1)$ and $F_{at} \sim Gamma(1.4^2, 1)$. The compliance class proportions are (1) $(\phi_{co}, \phi_{nt}, \phi_{at})=(0.10, 0.45, 0.45)$, (2) $(\phi_{co}, \phi_{nt}, \phi_{at})=(0.2, 0.4, 0.4)$, (3) $(\phi_{co}, \phi_{nt}, \phi_{at})=(1/3, 1/3, 1/3)$. In the fourth situation, all distributions are normal distributions: $F_{co}^{(0)}=F_{co}^{(1)} \sim N(0, 1), F_{nt} \sim N(-1, 1)$ and $F_{at} \sim N(1, 1)$ with  $(\phi_{co}, \phi_{nt}, \phi_{at})=(0.10, 0.45, 0.45)$. The sample size is $n=1000$. To compute the average performance, we consider 1000 simulated datasets. For each dataset, we compute the $L_2$ distance between the estimated function $\hat{F}$ and the true function $F$, $L_2 (\hat{F}, F) = \int (\hat{F}-F)^2 \mathrm d F$.

\begin{table}
\begin{center}
\caption{\small{Average performance of the three estimation methods}}
		\begin{tabular}{rlrrr}
			Situation & Method & Bias & SE & 1000MSE \\
			1 & Plug-in  & 0.0514 & 0.0923 & 11.16\\ 
			& MBL & 0.0288 & 0.0330 & 1.92\\
			& Rearrangement & 0.0359 & 0.0740 & 6.76\\[0.1cm]
			2 & Plug-in  & 0.0098 & 0.0101 & 0.20\\ 
			& MBL & 0.0088 & 0.0089 & 0.16\\
			& Rearrangement& 0.0087 & 0.0094 & 0.16\\[0.1cm]
			3 & Plug-in & 0.0032 & 0.0032 & 0.02\\
			& MBL & 0.0031 & 0.0031 & 0.02\\
			& Rearrangement & 0.0030 & 0.0031 & 0.02\\[0.1cm]
			4 & Plug-in & 0.0612 & 0.1745 & 34.20\\
			& MBL & 0.0276 & 0.0328 & 1.84\\
			& Rearrangement & 0.0255 & 0.0404 & 2.28\\
	\end{tabular}
\end{center}
	\label{tab:mbl}
\end{table}

Table~\ref{tab:mbl} shows the average performance of three considered estimation methods. Biases, standard errors and mean squared errors are reported. 
The MBL method has the least bias in situation 1, but the rearrangement method has the least bias in situations 2, 3 and 4. However, the MBL method has the least standard errors in every situation. Moreover, it has the best mean squared error in all the situations, although it has similar performance to the rearrangement method when an instrument is not weak.

\section{Oregon Health Insurance Experiment: The Effect of Medicaid Coverage on Mental Health }
\label{sec:example}

We consider the 2008 Oregon health insurance experiment data that is publicly available from \url{https://www.nber.org/oregon/1.home.html}. To investigate the effect of Medicaid on health outcomes, Oregon opened a waiting list for a limited number of spots in its Medicaid program for low-income, uninsured, able-bodied adults between 19-64, which had previously been closed to new enrollment.  From the waiting list, people selected by random lottery drawings, won the opportunity for themselves, and any household member, to apply for Oregon Health Program (OHP) Standard.  However, not
all persons selected by the lottery enrolled in Medicaid, either because they did not apply or
because they were deemed ineligible. The lottery process and OHP standard are described in more detail in \cite{finkelstein2012oregon}. This random assignment embedded in the lottery allows us to study the effect of Medicaid coverage in a random encouragement design. An indicator of winning a lottery is the instrumental variable. Also, enrollment to the Medicaid program is the non-randomized treatment variable. Approximately 2 years after the lottery, health outcomes are measured for persons who responded to the follow-up survey. In our example, we use a self-reported mental health outcome by using scores on the Medical Outcome Study 8-Item Short-Form Survey (SF-8). The scores range from 0 to 100, with higher scores indicating better self-reported health-related quality of life. The scale is normalized to yield a mean of 50 and a standard deviation of 10 in the general U.S. population. Details of other health outcomes and data collection have been provided in \cite{baicker2013oregon}. 

From the data, we consider a sample of 1,117 African Americans  (single-person households) who signed themselves up for the lottery. Among them, 546 people (48.9\%) were selected by the lottery drawings. The probability of Medicaid coverage is 0.511 in the lottery winning group and 0.226 in the other group. The plug-in estimates of the proportions for compliance classes are $(\breve{\phi}_{co}, \breve{\phi}_{nt}, \breve{\phi}_{at}) = (0.285, 0.489, 0.226)$. Lottery selection increased the probability of Medicaid coverage by 28.5 points among the single-person  African American households. The MBL estimates of the proportions are the same as the plug-in estimates up to three decimal places. The two-stage least squares (2SLS) estimate is 4.88 (95\% CI, 0.01 to 9.75) with $p$-value 0.050. The magnitude of improvement was approximately half of the standard deviation of the mental-component score. Furthermore, we can restrict our attention to a subsample of African Americans aged between 19-34 ($N = 378$). The 2SLS estimate is 9.92 (95\% CI, 0.98 to 18.86) with $p$-value 0.030. The estimated proportions of compliance classes are $(0.284, 0.495, 0.221)$ which are almost identical to the estimated proportions for the total African American  population. 

\begin{figure}[!h]
	\includegraphics[width=150mm]{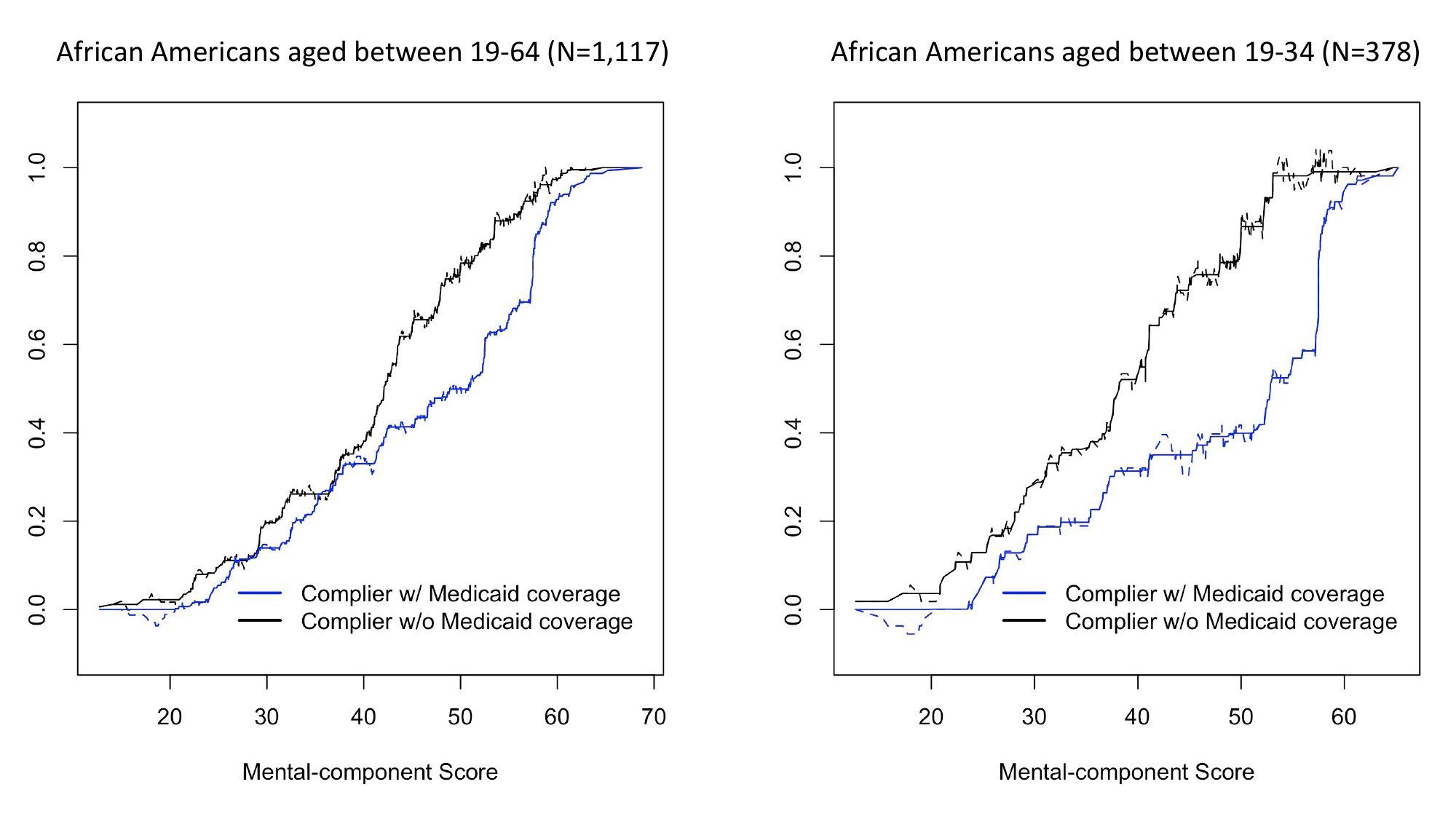}
	\caption{\small{Estimated distribution functions of the mental-component scores for compliers in the African American population. Higher mental-component scores indicates better self-reported mental health. The dotted blue and black lines are the plug-in estimates and the solid blue and black lines are the BL estimates of the distribution functions of the complier with Medicaid coverage and the complier without medicaid coverage, respectively.}}
	\label{fig2}
\end{figure}

Figure~\ref{fig2} shows the estimated distribution functions of the potential outcomes of mental-component scores for compliers when enrolled in the Medicaid program and when not enrolled. The left plot shows the plug-in estimates described in Section \ref{subsec:review} and the MBL estimates for African Americans aged between 19-64 (the full sample), and the right plot shows them for African Americans aged between 19-34. In both of the plots, we see that the estimated distribution function for complier without Medicaid coverage is almost always above the other. The gap between the two functions is wider at higher mental-component scores. Unlike the plug-in estimate, the MBL estimate satisfies the non-decreasing condition, and, as a result, there is a unique value of estimated scores corresponding to a specific quantile level. This feature can be useful for those who want to estimate the treatment effect at a certain quantile level using the estimated distribution functions. For example, the MBL method estimates that Medicaid coverage led to an increase of 8.70 points in the median score on the mental component for compliers. However, from the plug-in method, there are two values that correspond to the value 0.5 of the distribution for complier with Medicaid coverage, making it unclear how to compare the medians of the two distribution functions. For the young African American population (aged between 19-34), Medicaid coverage increased the median mental-component score by 14.68 points from 38.13 to 52.81. 

Furthermore, the BLRT can be conducted for testing the null hypothesis $H_0: F_{co}^{(0)} = F_{co}^{(1)}$. We apply both versions of the proposed binomial likelihood test, $T_{n}$ and $T_n^{simple}$, and compare them with $T_{KS}$. Using the asymptotic null distribution $A^2$, the $P$-values are computed as $(0.021, 0.020, 0.031)$ for $(T_n, T_n^{simple}, T_{KS})$. Moreover, the $P$-value of $T_n$ can be computed by the bootstrap procedure with $B=10,000$ described in Section~\ref{subsec:sim_blrt}, and the $P$-value is 0.021, which agrees with the asymptotic version of the $P$-value. Similarly, for the young African American  population, the estimated $P$-values are $(0.012, 0.012, 0.030)$ for $(T_n, T_n^{simple}, T_{KS})$. For a smaller sample size, the proposed BLRT produced a smaller $P$-value. For all considered tests, we reject the null hypothesis of equality of distributions at a significance level $\alpha=0.05$. 

Finally, using the BLRT, we also test another hypothesis $H_0^{location}: F_{co}^{(0)}(t-\mu) = F_{co}^{(1)}(t)$ for testing treatment effect heterogeneity. All possible values of $\mu$ are examined for the African American  population and the subpopulation of African Americans  aged between 19-34. For the African American  population, $H_0^{location}$ is not rejected for values between $0.76 \leq \mu \leq 10.59$. Also, for African Americans  aged between 19-34, $H_0^{location}$ is not rejected for values between $1.92 \leq \mu \leq 21.38$. These results show that there is no evidence of treatment effect heterogeneity.

\section{Discussion}
\label{sec:discussion}

We propose a non-parametric composite likelihood approach, referred to as the binomial likelihood (BL) method, for making causal inferences about the distributional treatment effect in a randomized experiment with an instrumental variable. The BL approach provides a non-parametric inferential tool similar to the classical parametric likelihood. The maximum binomial likelihood (MBL) method provides estimates of the outcome distributions, which are proper distribution functions and the binomial likelihood ratio test (BLRT) is a powerful technique to detect distributional changes when the outcome distributions are close to each other, especially when the IV is weak. 

Several extensions and generalizations of the BL method are possible. For instance, while constructing the BL functions one requires specification of the knots as evaluating points. We recommended to use all the observed outcomes for the knots, but a different specification can be considered depending on the question of interest. For example, in the Oregon Health Insurance experiment, if one wants to examine whether there is any effect for people above mental score 40, then only outcomes above 40 can be chosen as knots. In such a setting, a rejection implies that there is evidence of distributional changes for this specific subpopulation. 

The results in this paper show that the BL approach works well in randomized encouragement experiments where a compliance class is latent. A possible future direction could be to study the performance of the BL approach in general mixture models when there is a latent variable.

\section*{Acknowledgement}
The authors thank Abhishek Chakrabortty and Shirshendu Ganguly for helpful discussions. 

%
%

\bibliographystyle{apalike}
\bibliography{bl_library}

\end{document}


\def\spacingset#1{\renewcommand{\baselinestretch}%
{#1}\small\normalsize} \spacingset{1}


\title{\bf A Nonparametric Likelihood Approach for Inference in Instrumental Variable Models: Supplementary Materials}
\author{}
\date{}

\maketitle

\spacingset{1.25}

\section{Preliminaries}
\label{sec:prelim}

In this section, we start with reviewing functions that are defined in the main manuscript. Also, we define new functions that are necessary for later proofs. First, we defined the potential outcome distribution for compliance classes. $\bm F=(F_{co}^{(0)}, F_{co}^{(1)}, F_{nt}, F_{at})$ is a collection of distribution functions for compliers without treatment, compliers with treatment, never-takers, and always-takers respectively. Also, we defined the proportions of compliance classes as $\bm \phi = (\phi_{nt}, \phi_{at})$. The proportion of compliers $\phi_{co}$ was defined as $1-\phi_{nt} - \phi_{at}$. $\chi_1$ is the proportion of $Z=1$. These functions represent the true values. Table~\ref{tab:summary} summarizes other defined functions. 

We define new functions $\M_n (\bm \theta, \bm \chi) $ and $\M (\bm \theta, \bm \chi)$ for fixed knots $\bm{t} = (t_1, \ldots, t_m)$ based on the binomial log-likelihood $\ell_{\bm{Y}, \bm{D} | \bm{Z}}(\bm \theta, \bm \chi)$, 
\begin{align*}
	\M_n(\bm \theta, \bm \chi) &= \frac{1}{n} \ell_{\bm{Y}, \bm{D} | \bm{Z}}(\bm \theta, \bm \chi) = \frac{1}{n} \sum_{i=1}^{n} \ell_{Y_i, D_i | Z_i}(\bm \theta, \bm \chi) \\
	\M(\bm \theta, \bm \chi)  &=  \E [ \ell_{Y_i, D_i | Z_i}(\bm \theta, \bm \chi) ].
\end{align*} 
 Since $\M_n$ is just $n$ times samller than the binomial log-likelihood $\ell_{\bm{Y}, \bm{D} | \bm{Z}}(\bm \theta, \bm \chi) $, the maximum binomial likelihood estimator $(\hat{\bm \theta}, \hat{\bm \chi})$ that maximizes $\ell_{\bm{Y}, \bm{D} | \bm{Z}}(\bm \theta, \bm \chi) $ can be obtained by maximizing $\M_n$. To apply the theory of M-estimators, we use mostly $\M_n$ instead of $\ell_{\bm{Y}, \bm{D} | \bm{Z}}(\bm \theta, \bm \chi) $ in later sections. 
 
\begin{table}[t]
	\centering
	\caption{Summary of defined functions}\label{tab:summary}
	\begin{tabular}{ll}
		\toprule
		Functions & Descriptions \\
		$\bm{t} = (t_1, \ldots, t_m)$ & Knots where the likelihood functions are evaluated \\
		$\bm \theta = (\theta_{co}^{(0)}, \theta_{co}^{(1)}, \theta_{nt}, \theta_{at})$ & A collection of functional variables that estimates $\bm F$ \\
		$\bm \chi = (\chi_{nt}, \chi_{at})$ & A collection of functional variables that estimates $\bm \phi$ \\
		$F_{zd}(t)$ & $\pr (Y \leq t \>|\> Z=z, D=d)$ \\
		$\theta_{zd}$ & A functional variable that estimates $F_{zd}$\\
		$\eta_{zd}$ & $\pr(Z=z, D=d)$ \\
		$n_{zd}$ & $\sum_{i=1}^{n} \bm{1} (Z_i=z, D_i=d)$ \\
		$\overline{F}_{zd}$ & The empirical distribution function of $\bm{Y}$ in the cell $Z=z, D=d$, \\
		& equivalently $\overline{F}_{zd}(t)=(1/n_{zd}) \cdot \sum_{i=1}^{n} \bm{1} (Z_i=z, D_i=d, Y_i \leq t)$ \\
		$\breve{\bm F} = (\breve{F}_{co}^{(0)}, \breve{F}_{co}^{(1)}, \breve{F}_{nt}, \breve{F}_{at})$ & The plug-in estimator of $\bm F$ \\ 
		$\breve{\bm \phi} = (\breve{\phi}_{nt}, \breve{\phi}_{at})$ & The plug-in estimator of $\bm \phi$ \\
		$H(t)$ & $\sum_{z, d \in \{0, 1\}} \eta_{zd} F_{zd}(t)$ \\
		$\overline{H}(t)$ & The emprirical distribution of $\bm{Y}$, $\sum_{z, d \in \{0,1\}} (n_{zd}/n) \cdot \overline{F}_{zd}(t)$ \\
		$\bm\vartheta \times \bm\varphi $ & A {\it unrestricted} parameter space for $(\bm \theta, \bm \chi)$ \\
		$\bm\vartheta_{0} \times \bm\varphi $ & A {\it unrestricted} parameter space for $(\bm \theta, \bm \chi)$ under the null\\
		$\bm\vartheta_{+} \times \bm\varphi_{+} $ & A {\it restricted} parameter space for $(\bm \theta, \bm \chi)$ \\
		$\bm\vartheta_{+, 0} \times \bm\varphi_{+} $ & A {\it restricted} parameter space for $(\bm \theta, \bm \chi)$ under the null \\
	\end{tabular}
\end{table}

\subsection{Basic inequalities}
\label{sec:pfblI1}

In this section we collect some basic inequalities and properties of the objective function $\M_n$. We begin with a few preliminary observations which will be used later in our proofs: 

\begin{obs}[A1] Fix $a\in [0, 1]$. Then for every $x\in [0, 1]$, $J(a, x)=a\log x+(1-a)\log(1-x)\leq a\log a+(1-a)\log (1-a)=I(a)$.
	\label{obs:ineq}
\end{obs}

\begin{proof}The inequality is trivially satisfied when $a\in\{0, 1\}$. Therefore, assume that $a\in (0, 1)$, and define a random variable $W$ which takes values $\frac{x}{a}$ and $\frac{1-x}{1-a}$ with probabilities $a$ and $1-a$, respectively. Note that $\E W=1$. Then by Jensen's inequality, $\E(\log W)=a\log \frac{x}{a}+(1-a)\log\frac{1-x}{1-a}\leq \log \E W=0$, which completes the proof of the result.
\end{proof}

\begin{obs}[A2]\label{obs:cdiff} Fix $a\in [0, 1]$. Then for every $x\in [0, 1]$, $a \log\frac{x}{a} + (1-a) \log \frac{1-x}{1-a} \leq  - \frac{1}{2}(x-a)^2$. 
\end{obs}

\begin{proof} For a given $a \in (0,1)$, let $f_a(x) = a \log\frac{x}{a} + (1-a) \log \frac{1-x}{1-a}$.   By a second order Taylor expansion around the point $a$, $f_a(x)=\frac{1}{2}(x-a)^2 f_a''(\gamma_{x, a})$ where $\gamma_{x, a} \in [x \wedge a, x \vee a]$\footnote{For $a, j \in \R$, define $a\wedge b=\min\{a, b\}$ and $a\vee b=\max\{a, b\}$.} and $f_a''(x)=-\frac{a}{x^2} - \frac{1-a}{(1-x)^2}$. Note that, for $a\in (0, 1)$ fixed, the function $f_a''(x)$ is convex. It is easy to check that the minimum is attained at $x_0 = (\frac{a}{1-a})^{\frac{1}{3}}$, and $f''_a(x)\geq f''_a(x_0)> 1$. This implies, $f_a(x)\leq - \frac{1}{2}(x-a)^2$.
\end{proof}

\begin{obs}[A3] \label{obs:orderstatistics}
	Fix $0< t < 1$. Suppose $Y_1, Y_2, \ldots, Y_n$ are i.i.d. samples from the distribution $H(t)$. Then, for $z, d \in \{0, 1\}$ $$F_{zd}(Y_{\ceil{nt}})\xrightarrow{P} H_{zd}^{-1}(t),$$ where $H_{zd}(t)= (H \circ F_{zd}^{-1})(t)$.
\end{obs}

\begin{proof} Without of generality, take $z=0$ and $d=0$. Then the distribution of $$W_1=F_{00}(Y_1), ~ W_2=F_{00}(Y_2), \ldots, ~ W_n=F_{00}(Y_n)$$ are i.i.d. samples with distribution function $H_{00}(t)=\eta_{00}t+\eta_{01}(F_{01}\circ F_{00}^{-1})(t)+\eta_{10}(F_{10}\circ F_{00}^{-1})(t)+\eta_{11}(F_{11}\circ F_{00}^{-1})(t)$. This implies, for any $0<t< 1$, $F_{00}(Y_{(\ceil{nt})})=W_{(\ceil{nt})}\xrightarrow{P}  H^{-1}_{00}(t)$, where the last step uses the convergence of sample quantiles to the corresponding population quantiles \citet{walker1968}. 
\end{proof}



\subsection{Proof of Proposition 1}
\label{sec:pfbl} 

To begin recall the definition of $K_{zd}^{ij}=\{Z_i=z, D_i=d \>|\> t_j \}$. Then, for all $i$ and $j$, 
\begin{align}\label{eq:prob}
\pr(K_{zd}^{ij}) = \left\{
\begin{array}{ccc}
(1-\chi_1) \cdot (1-\chi_{at}(t_j))  &    \text{for} & z=0, d=0,   \\
(1-\chi_1) \cdot \chi_{at}(t_j)   &   \text{for} & z=0, d=1,  \\
\chi_1 \cdot \chi_{nt}(t_j)   &  \text{for} & z=1, d=0, \\
\chi_1 \cdot (1-\chi_{nt}(t_j))  &  \text{for}  & z=1, d=1.      
\end{array}
\right.
\end{align}

Next, to describe the binomial log-likelihood function $\log L(\bm \theta, \bm \chi, \chi_1 | \cD_n)$, we first define the functions $S_{zd}(\cD_{n,i}=(Z_i, D_i, Y_i), t_j)$ where $\cD_{n, i}$ is single data point of $(Z_i, D_i, Y_i)$ for $z, d \in \{0, 1\}$ that are defined as 
{\small
\begin{align}\label{eq:Suv}
S_{00}(\cD_{n,i}, t_j) & = \bm 1\{K_{00}^{ij}\} \log \left\{ 1-\chi_{at}(t_j) \right \} \nonumber \\
& + \bm 1\{K_{00}^{ij}, Y_i \leq t_j\}  \log\left\{ \frac{1-\chi_{nt}(t_j)-\chi_{at}(t_j)}{1-\chi_{at}(t_j)} \theta_{co}^{(0)}(t_j) + \frac{\chi_{nt}(t_j)}{1-\chi_{at}(t_j)} \theta_{nt}(t_j) \right\},  \nonumber \\
&+\bm 1\{K_{00}^{ij}, Y_i > t_j\}  \log\left\{ \frac{1-\chi_{nt}(t_j)-\chi_{at}(t_j)}{1-\chi_{at}(t_j)} (1-\theta_{co}^{(0)}(t_j)) + \frac{\chi_{nt}(t_j)}{1-\chi_{at}(t_j)} (1-\theta_{nt}(t_j)) \right \},  \nonumber \\ 
S_{01}(\cD_{n,i}, t_j) & =\bm 1\{K_{01}^{ij}\}  \log  \chi_{at}(t_j) +\bm 1  \{K_{01}^{ij}, Y_i  \leq t_j\} \log \theta_{at}(t_j)  + \bm 1\{K_{01}^{ij}, Y_i > t_j\} \log (1-\theta_{at}(t_j)), \nonumber \\
S_{10}(\cD_{n,i}, t_j) &= \bm 1\{K_{10}^{ij}\}  \log  \chi_{nt}(t_j)   + \bm 1  \{K_{10}^{ij}, Y_i \leq t_j\} \log\theta_{nt}(t_j) +\bm 1\{K_{10}^{ij}, Y_i > t_j\} \log(1-\theta_{nt}(t_j)), \nonumber \\
S_{11}(\cD_{n,i}, t_j) & = \bm 1\{K_{11}^{ij}\} \log \left\{   1- \chi_{nt}(t_j)\right\}   \nonumber \\
&+\bm 1\{K_{11}^{ij}, Y_i \leq t_j\}  \log\left\{ \frac{1-\chi_{nt}(t_j)-\chi_{at}(t_j)}{1-\chi_{nt}(t_j)} \theta_{co}^{(1)}(t_j) + \frac{\chi_{at}(t_j)}{1-\chi_{nt}(t_j)} \theta_{at}(t_j)\right \}  \nonumber \\
&+\bm 1\{K_{11}^{ij}, Y_i > t_j\} \log \left\{  \frac{1-\chi_{nt}(t_j)-\chi_{at}(t_j)}{1-\chi_{nt}(t_j)} (1-\theta_{co}^{(1)}(t_j)) + \frac{\chi_{at}(t_j)}{1-\chi_{nt}(t_j)} (1-\theta_{at}(t_j))\right\}. 
\end{align}
}
\normalsize
With \eqref{eq:prob} and \eqref{eq:Suv}, the binomial log-likelihood can be re-written as follows:
\begin{align}\label{eq:binom}
\log L(\bm \theta, \bm\chi, \chi_1|\cD_n)=  \ell_{\bm{Z}}(\chi_1) +  \frac{1}{m}\sum_{i=1}^{n}  \sum_{j=1}^{m}  \sum_{z, d \in\{0, 1\}} S_{zd}(\cD_{n,i}, t_j),
\end{align}
where $ \ell_{\bm{Z}}(\chi_1) =\frac{1}{n} \left\{ (n_{00} + n_{01}) \log (1-\chi_1) +  (n_{10} + n_{11}) \log \chi_1 \right\}$. Note that the second-term in the RHS of \eqref{eq:binom} above does not depend on $\chi_1$. Hence, 
$$\hat \phi_1=\arg\max_{\chi_1 \in [0, 1]}\ell_{\bm{Z}}(\chi_1)=\frac{n_{10}+n_{11}}{n},$$ 
as required.

Furthermore, the defined quantities $\ell_{zd} (\bm \theta, \bm \chi)$ in the main manuscript can be represented by $\ell_{zd}(\bm \theta, \bm \chi) = \frac{1}{m} \sum_{i=1}^{n}\sum_{j=1}^{m} S_{zd}(\cD_{n,i}, t_j)$. Also, since the sub-loglikelihood function $\ell_{\bm{Y}, \bm{D} | \bm{Z}}(\bm \theta, \bm \chi)$ is defined as $\sum_{z,d \in \{0,1\}} \ell_{zd}(\bm \theta, \bm \chi)$, the second term on the right hand side in \eqref{eq:binom} is  $\ell_{\bm{Y}, \bm{D} | \bm{Z}}(\bm \theta, \bm \chi)$. Therefore,  $(\hat{\bm F}, \hat{\bm \phi})=\arg \max_{\bm \theta\in \bm \vartheta_+, \bm \chi \in \bm \varphi_{+}} \ell_{\bm{Y}, \bm{D} | \bm{Z}}(\bm \theta, \bm \chi)$.

\subsection{Unrestricted maximization of the sample objective function}
\label{sec:Mnmax}

Here we show that maximizing the sample objective  function $\M_n$  over the unrestricted parameter space gives the plug-in estimates, justifying the choice of $\M_n$ as an approximate surrogate for the actual likelihood. As a consequence,  it follows that the plug-in and the maximum binomial likelihood estimates of the compliance classes are equal with probability 1. In the following proofs for theoretical arguments, we choose knots based on actual observations of $\bm{Y}$. As we discussed in the main manuscript, to avoid the vanishing issue, we use the truncation parameter $\kappa$ for theoretical purposes. We set $I_\kappa=[\ceil{n \kappa}, \ceil{n(1-\kappa)}]$ for a fixed small constant $\kappa$. We can simply consider this as a special choice of knots by defining $ |I_{\kappa}|=m$ and $\bm{t} = (t_{1}, \ldots, t_{m}) = (Y_{(\ceil{n\kappa})}, \ldots, Y_{(\ceil{n(1-\kappa)})})$. 

\begin{lemma}\label{lemma:Mnopt} Let $\hat{\bm \phi}$ be the maximum binomial likelihood estimate and $\breve{\bm \phi}$ be the plug-in estimate. Then 
	\begin{align}\label{eq:Mn2}
	\arg\max_{\bm \theta \in \bm \vartheta, \bm \chi \in \bm\varphi} \M_n(\bm \theta, \bm\chi)=(\breve{\bm G},  \breve{\bm \phi}),
	\end{align}
	where $\breve{\bm G}\in \breve \cF=\{\breve{\bm G}\in \bm\vartheta: \breve{\bm G}(Y_{(j)})=\breve{\bm F}(Y_{(j)}) \text{ for } j \in I_\kappa\}$ and $\breve{\bm \phi} = (\breve{\phi}_{nt}(Y_{(j)}), \breve{\phi}_{at}(Y_{(j)}))$ with $\breve{\phi}_{nt}(Y_{(j)}) \equiv \breve{\phi}_{nt}= n_{10}/(n_{10}+n_{11})$ and $\breve{\phi}_{nt}(Y_{(j)}) \equiv \breve{\phi}_{at}= n_{01}/(n_{00}+n_{01})$ for all $j \in I_{\kappa}$. 
\end{lemma}

\begin{proof} Define $n_{00}+n_{01}=n_0$ and $n_{10}+n_{11}=n_1$. Then by Observation \ref{obs:ineq}, 
	\begin{align*}
	\M(\bm \theta, \bm\chi) &\leq \frac{1}{|I_{\kappa}|} \sum_{j \in I_{\kappa}} \left\{ \frac{n_{0}}{n} I\left(\frac{n_{10}}{n_1}\right) + \frac{n_{1}}{n} I\left(\frac{n_{01}}{n_0}\right)  + \sum_{z, d \in \{0, 1\}}\frac{n_{zd}}{n} I(\overline F_{zd}(Y_{(j)}))\right\}.
	\end{align*}
	Moreover, the equality is attained when 
	$$\begin{pmatrix}
	\theta_{00}(Y_{(j)}) \\
	\theta_{10}(Y_{(j)}) \\ 
	\theta_{01}(Y_{(j)}) \\
	\theta_{11}(Y_{(j)})
	\end{pmatrix}=
	\begin{pmatrix}
	\overline F_{00}(Y_{(j)}) \\
	\overline F_{10}(Y_{(j)}) \\
	\overline F_{01}(Y_{(j)}) \\ 
	\overline F_{11}(Y_{(j)}))
	\end{pmatrix}$$ 
	and $(\chi_{nt}(Y_{(j)}), \chi_{at}(Y_{(j)})) = (n_{10}/n_1, n_{01}/n_0)$ for all $j \in I_\kappa$ . This implies that the equality is attained when $(\chi_{nt}(Y_{(j)}), \chi_{at}(Y_{(j)})) = (\breve \phi_{nt}, \breve \phi_{at})$ and $\theta_{at}(Y_{(j)})=\overline F_{01}(Y_{(j)})=\breve F_{at}(Y_{(j)})$, $\theta_{nt}(Y_{(j)})=\overline F_{10}(Y_{(j)})=\breve F_{at}(Y_{(j)})$ and 
	\begin{align}
	\theta_{co}^{(0)}(Y_{(j)})&=\frac{\overline F_{00}(Y_{(j)})- \frac{\breve \phi_{nt}}{\breve \phi_{co}+\breve \phi_{nt}}\overline F_{10}(Y_{(j)})}{\frac{\breve \phi_{co}}{\breve \phi_{co}+\breve \phi_{nt}}}=\breve F_{co}^{(0)}(Y_{(j)}) \\
	\theta_{co}^{(1)}(Y_{(j)})&=\frac{\overline F_{11}(Y_{(j)})- \frac{\breve \phi_{at}}{\breve \phi_{co}+\breve \phi_{at}}\overline F_{01}(Y_{(j)})}{\frac{\breve \phi_{co}}{\breve \phi_{co}+\breve \phi_{at}}}=\breve F_{co}^{(1)}(Y_{(j)})
	\end{align}
	where $\breve{\phi}_{co} = 1- \breve{\phi}_{nt} - \breve{\phi}_{at}$ for $j \in I_\kappa$. This completes the proof of \eqref{eq:Mn2}. 
\end{proof}

\section{Proof of Theorem 1}
\label{sec:pfthblII}

In this section we prove Theorem 1, which shows that the plug-in estimate $\breve{\bm F}$ and the maximum binomial likelihood estimate $\hat {\bm F}$ have the same limiting distribution.

\subsection{Comparing the (sample) objective functions}

From Lemma \ref{lemma:Mnopt}, we know that $\M_n(\breve{\bm F}, \breve{\bm \phi}) -\M_n(\hat{\bm F}, \hat{\bm \phi}) \geq 0$. One of the main steps towards the proof of  $\frac{1}{|I_{\kappa}|} \sum_{j \in I_\kappa} ||\sqrt n\{\bm{\hat{F}}(Y_{(j)})-\breve{\bm F}(Y_{(j)})\}||^2_2=o_P(1)$, is to show that this difference is small, more precisely, 
\begin{align}\label{eq:MnTSI} 
\M_n(\breve{\bm F}, \breve{\bm \phi}) - \M_n(\hat{\bm F}, \hat{\bm \phi}) = o_P(n^{-\frac{5}{4}}).
\end{align} 
To this end, we have the following definition:

\begin{definition}\label{definition:tildeF}
	Define $\tilde{\bm F}=(\tilde F_{co}^{(0)}, \tilde F_{nt}, \tilde F_{co}^{(1)}, \tilde F_{at})' \in \cP([0, 1]^\R)^4$, as follows:  $$\tilde F_{nt}=\breve F_{10}=\overline F_{10} \quad and \quad \tilde F_{at}=\breve F_{01}=\overline F_{01},$$ are the corresponding empirical distribution functions, and 
	\begin{align}\label{eq:incF}
	\tilde F_{co}^{(0)}=\arg\min_{\theta \in \cP([0, 1]^\R)}\sum_{j \in I_\kappa} (\breve F_{co}^{(0)}(Y_{(j)})-\theta(Y_{(j)}))^2, \nonumber \\
	\tilde F_{co}^{(1)}=\arg\min_{\theta \in \cP([0, 1]^\R)}\sum_{j \in I_\kappa} (\breve F_{co}^{(1)}(Y_{(j)})-\theta(Y_{(j)}))^2.
	\end{align}
	Note that this only defines $\tilde F_{co}^{(0)}$ and $\tilde F_{co}^{(1)}$ at the knots $\{Y_{(j)}\}_{j \in I_\kappa}$. To ensure \eqref{eq:incF} is well-defined we extend $\tilde F_{co}^{(0)}$ and $\tilde F_{co}^{(1)}$ between the knots by right-continuous interpolation, and extrapolate it beyond the knots to 0 and 1. Moreover, define 
	\begin{align}\label{eq:tildeFuv}
	\tilde F_{00}&= \breve \lambda_0 \tilde F_{co}^{(0)}+(1-\breve \lambda_0) \tilde F_{nt},  \nonumber \\
	\tilde F_{11}&= \breve \lambda_1 \tilde F_{co}^{(1)}+(1-\breve \lambda_1)\tilde F_{at}, 
	\end{align}
	where $\breve \lambda_0 = \frac{1-\breve{\phi}_{nt} - \breve{\phi}_{at}}{1-\breve{\phi}_{at}}$ and $\breve \lambda_1 = \frac{1-\breve{\phi}_{nt} - \breve{\phi}_{at}}{1-\breve{\phi}_{nt}}$. Also, define $\tilde F_{10}=\tilde F_{nt}$ and $\tilde F_{01}=\tilde F_{at}$. 
\end{definition}

Note, since the true proportion $\bm \phi$ is in the interior of $\bm \varphi_{+}$, which is an open subset of $\bm \varphi$, there exists an $\varepsilon >0$ such that $B(\bm \phi, \varepsilon)=\{\bm \chi \in \bm \varphi: ||\bm \chi(t)-\bm \phi(t)||_2 < \varepsilon \> \text{for all } \> t \} \subset \bm \varphi_+$. Moreover, there exists $n \geq N(\varepsilon, \delta)$ such  that $\pr(\breve {\bm \phi} \notin B(\bm \phi, \varepsilon) ) < \delta$. Therefore, $\pr( \breve {\bm \phi} \text{ not in the interior of } \bm \varphi_+) \leq \pr(\breve {\bm \phi} \notin B(\bm \phi, \varepsilon) )  <\delta$. To this end, let $$\sB_1=\left\{\breve {\bm \phi} \text{ is in the interior of } \bm \varphi_+ \right\} \bigcap \left\{\breve {\bm F} \text{ is coordinate-wise in the interior of } \R^{[0, 1]}) \right\}.$$ From the discussion above it is clear that the $\pr(\sB_1^c)\rightarrow 0$. Therefore, for the remainder of this section all events will be on the set $\sB_1$.

To begin with this gives, 
$$ \M_n(\hat{\bm F}, \hat{\bm \phi}) \geq  \M_n(\tilde{\bm F}, \breve{\bm \phi}),$$
since $(\tilde{\bm F}, \hat{\bm \phi}) \in \bm\vartheta_+\times \bm \varphi_+$ and $(\hat{\bm F}, \hat{\bm \phi})$ maximizes $\M_n$ over $\bm\vartheta_+ \times \bm \varphi_+$, by definition. Therefore, to show \eqref{eq:MnTSI} it suffices to prove that 
\begin{align}\label{eq:MnTSII} 
\M_n(\breve{\bm F}, \breve{\bm \phi}) - \M_n(\tilde{\bm F}, \breve{\bm \phi}) =o_P(n^{-\frac{5}{4}}).
\end{align} 

\begin{obs}[B4] \label{obs:md2} For both estimate $\breve{\bm F}$ and $\tilde{\bm F}$, we have 
	$$\M_n(\breve{\bm F}, \breve{\bm \phi}) - \M_n(\tilde{\bm F}, \breve{\bm \phi})=O_P(1) \left[\frac{1}{|I_{\kappa}|}\sum_{z \in \{0,1\}}\sum_{j \in I_\kappa} (\breve{F}_{co}^{(z)}(Y_{(j)})-  \breve{F}_{co}^{(z)}(Y_{(j)}))^2\right],$$
	whenever  $|\tilde F_{zz}(Y_{(\ceil{n\kappa})})- \overline F_{zz}(Y_{(\ceil{n\kappa})})|=o_P(1)$ and $|\tilde F_{zz}(Y_{(\ceil{n(1-\kappa)})})- \overline F_{zz}(Y_{(\ceil{n(1-\kappa)})})|=o_P(1)$.  
\end{obs}

\begin{proof} 
	\begin{align}\label{diff2}
	&\M_n(\breve{\bm F},  \breve {\bm \phi})  - \M_n(\tilde{\bm F}, \breve {\bm \phi}) \nonumber \\
	&=\frac{1}{|I_{\kappa}|}\sum_{z, d \in \{0, 1\}}\sum_{j \in I_\kappa} \frac{n_{zd}}{n} (J(\overline F_{zd}(Y_{(j)}), \breve  F_{zd}(Y_{(j)}))-J(\overline F_{zd}(Y_{(j)}), \tilde  F_{zd}(Y_{(j)}))) \nonumber \\
	&=\frac{1}{|I_{\kappa}|}\sum_{z \in \{0, 1\}} \sum_{j \in I_\kappa}\frac{n_{zz}}{n} (J(\overline F_{zz}(Y_{(j)}), \breve  F_{zz}(Y_{(j)}))-J(\overline F_{zz}(Y_{(j)}), \tilde  F_{zz}(Y_{(j)}))) \tag*{(since $\tilde F_{10}=\breve F_{10}$ and $\tilde F_{01}=\breve F_{01}$)}\nonumber \\
	&=\frac{1}{|I_{\kappa}|}\sum_{z \in \{0, 1\}} \sum_{j \in I_\kappa}\frac{n_{zz}}{n} \left( J(\overline F_{zz}(Y_{(j)}), \overline F_{zz}(Y_{(j)})) - J(\overline F_{zz}(Y_{(j)}), \tilde  F_{zz}(Y_{(j)}))\right) \nonumber \\
	&=\frac{1}{|I_{\kappa}|} \sum_{z \in \{0, 1\}} \frac{n_{zz}}{n} \sum_{j \in I_\kappa}T_{z}(Y_{(j)}),
	\end{align}
	where 
	\begin{align*}
	T_{z}(Y_{(j)})&=\overline F_{zz}(Y_{(j)}) \log \frac{ \overline F_{zz}(Y_{(j)}) }{\tilde  F_{zz}(Y_{(j)})} +(1-\overline F_{zz}(Y_{(j)}) )\log\frac{1- \overline  F_{zz}(Y_{(j)})}{1- \tilde F_{zz}(Y_{(j)})}.
	\end{align*}

	Now, a two-term Taylor expansion of the function $f_a(x)=-a \log\frac{x}{a} - (1-a) \log \frac{1-x}{1-a}$,  at $x=a$, gives 
	\begin{align}\label{eq:tb}
	T_{z}(Y_{(j)})&=\frac{(\tilde{F}_{zz}(Y_{(j)})-\overline F_{zz}(Y_{(j)}))^2}{2} \bigg\{ \frac{\overline F_{zz}(Y_{(j)})}{(\omega_{zz}(Y_{(j)}))^2} - \frac{1-\overline F_{zz}(Y_{(j)})}{ (1-\omega_{zz}(Y_{(j)}))^2} \bigg\},
	\end{align}
	and $\omega_{zz}(Y_{(j)}) \in [\overline F_{zz}(Y_{(j)})\wedge \tilde F_{zz}(Y_{(j)}), \tilde F_{zz}(Y_{(j)})\vee \overline F_{zz}(Y_{(j)})]$. 
	
	Note that $\omega_{zz}(Y_{(j)}) \geq \overline F_{zz}(Y_{(\ceil{n\kappa})})\wedge \tilde F_{zz}(Y_{(\ceil{n\kappa})})$ and $\overline F_{zz}(Y_{(j)}) \leq \overline F_{zz}(Y_{(\ceil{n(1-\kappa)})})$. Therefore, 
	$$\frac{\overline F_{zz}(Y_{(j)})}{(\omega_{zz}(Y_{(j)}))^2}  \leq \frac{\overline F_{zz}(Y_{(\ceil{n(1-\kappa)})})}{\overline F_{zz}(Y_{(n\kappa)})\wedge \tilde F_{zz}(Y_{(n\kappa)})} = O_P(1),$$ since $\overline F_{zz}(Y_{(\ceil{n\kappa})})\pto H_{zz}^{-1}(\kappa)$, $\overline F_{zz}(Y_{(\ceil{n(1-\kappa)})})\pto H_{zz}^{-1}(1-\kappa)$ using Observation \ref{obs:orderstatistics}, and $|\tilde F_{zz}(Y_{(\ceil{n\kappa})})- \overline F_{zz}(Y_{(\ceil{n\kappa})})|=o_P(1)$ by assumption. Similarly, 
	$$\frac{1-\overline F_{zz}(Y_{(j)})}{ (1-\omega_{zz}(Y_{(j)}))^2} =O_P(1).$$ 
	Therefore, 
	\begin{align*}
	\sum_{j \in I_\kappa} |T_{z}(Y_{(j)})| & \leq \sum_{j \in I_\kappa} \frac{(\tilde{F}_{zz}(Y_{(j)})-  \overline{F}_{zz}(Y_{(j)}))^2}{2} \left\{\frac{\overline F_{zz}(Y_{(j)})}{(\omega_{zz}(Y_{(j)}))^2} + \frac{1-\overline F_{zz}(Y_{(j)})}{ (1-\omega_{zz}(Y_{(j)}))^2}\right\} \nonumber \\
	& =O_P(1) \sum_{j \in I_\kappa} (\tilde{F}_{zz}(Y_{(j)})-  \overline{F}_{zz}(Y_{(j)}))^2 \nonumber \\
	&=O_P(1)\sum_{j \in I_\kappa} (\tilde{F}_{co}^{(z)}(Y_{(j)})-  \breve{F}_{co}^{(z)}(Y_{(j)}))^2,
	\end{align*}
	where the last step uses $\tilde F_{00}= \breve \lambda_0 \tilde F_{co}^{(0)}+(1-\breve \lambda_0) \tilde F_{nt}$, $\overline F_{00}= \breve \lambda_0 \breve F_{co}^{(0)}+(1-\breve \lambda_0) \breve F_{nt}$, and $\tilde F_{nt}=\breve F_{nt}$, and similarly for $\tilde F_{11}$ and $\overline F_{11}$. 
\end{proof}

Therefore, to show that $\M_n(\breve{\bm F}, \breve{\bm \phi}) - \M_n(\tilde{\bm F}, \breve{\bm \phi}) =o_P(1/n)$, it suffices to show that $$\sum_{j \in I_\kappa} (\tilde{F}_{co}^{(z)}(Y_{(j)})-  \breve{F}_{co}^{(z)}(Y_{(j)}))^2=o_P(1),$$
for $z \in \{0, 1\}$. This is the content of the following proposition, gives an error rate of $o_P(n^{-\frac{1}{4}})$.

\begin{proposition} \label{proposition:pavabl} For $z \in \{0, 1\}$, 
	\begin{align}\label{eq:pavabl}
	\sum_{j \in I_\kappa}\left(\tilde{F}_{co}^{(z)}(Y_{(j)}) - \breve{F}_{co}^{(z)}(Y_{(j)})\right)^2=o_P(n^{-\frac{1}{4}}).
	\end{align} 
	This implies, $\M_n(\breve{\bm F}, \breve{\bm \phi}) - \M_n(\hat{\bm F}, \hat{\bm \phi}) =o_P(n^{-\frac{5}{4}})$. 
\end{proposition}

\subsection{Proof of Proposition \ref{proposition:pavabl}} To begin with, note that \eqref{eq:pavabl} implies 
$$|\tilde F_{zz}(Y_{(\ceil{n\kappa})})- \overline F_{zz}(Y_{(\ceil{n\kappa})})|=o_P(1), \quad |\tilde F_{zz}(Y_{(\ceil{n(1-\kappa)})})- \overline F_{zz}(Y_{(\ceil{n(1-\kappa)})})|=o_P(1).$$ Therefore, using Observation \ref{obs:md2} gives $\M_n(\breve{\bm F}, \breve{\bm \phi}) - \M_n(\tilde{\bm F}, \breve{\bm \phi})=o_P(n^{-\frac{5}{4}})$, which implies $\M_n(\breve{\bm F}, \breve{\bm \phi}) - \M_n(\hat{\bm F}, \hat{\bm \phi}) =o_P(n^{-\frac{5}{4}})$, from \eqref{eq:MnTSII}.  Therefore, the rest of this section is devoted to the proof of \eqref{eq:pavabl}.

We will prove the result for $z=0$. The other case for $z=1$ can be done similarly. To begin with note that on the event $\sB_1$, $|\breve{F}^{(0)}_{co}(t)| \in [0, 1]$ for all $t \in \R$, and, therefore, maximum in \eqref{eq:incF} can be taken over $\mathbb I([0,1]^\R)$, the set of increasing functions from $\R \rightarrow [0, 1]$. Then the well-know result of \citet{robertson1988} gives 
\begin{equation}\label{min_max}
\tilde{F}_{co}^{(0)}(Y_{(j)}) = \min_{\ell \geq j} \max_{k \leq j} \frac{\sum_{i=k}^{\ell}\breve{F}_{co}^{(0)}(Y_{(i)})}{\ell-k+1}
\end{equation}

Now, define $f(k, \ell) = \frac{ \sum_{i=k}^\ell \breve{F}_{co}^{(0)}(Y_{(i)}) }{\ell-k+1}-\breve{F}_{co}^{(0)}(Y_{(j)})$. We will use the following lemma.

\begin{lemma}\label{lemma1}
	Let $j \in I_\kappa$. For any function $f: [n] \times [n] \rightarrow \R$, the following inequality holds 
	\begin{equation}
	\left( \min_{\ell \geq j} \max_{k \leq j} f(k, \ell) \right)^2 \leq \left(\max_{k \leq j}f_+(k, j) \right)^2 + \left(\max_{\ell \geq j} f_-(j, \ell) \right)^2, 
	\end{equation}
	where $f_+ =\max(f,0)$ and $f_- = (-f)_+$. 
\end{lemma}

\begin{proof} We have 
	\begin{align*}
	\min_{\ell \geq j} \max_{k \leq j} f (k, \ell) \leq  \max_{k \leq j} f (k, j) \leq \max_{k \leq j} f_+(k, j).
	\end{align*}
	Similarly, 
	\begin{align*}
	\min_{\ell \geq j} \max_{k \leq j} f (k, \ell) \geq  \min_{\ell \geq j} f (j, \ell)= - \max_{\ell \geq j} -f(j, \ell) \geq - \max_{\ell \geq j} f_-(j, \ell)
	\end{align*}
	Therefore, the lemma follows.
\end{proof}

To prove \eqref{eq:pavabl}, we have to show $\pr_{\sB_1}\left(\sum_{j \in I_\kappa}(\tilde{F}_{co}^{(0)}(Y_{(j)}) - \breve{F}_{co}^{(0)}(Y_{(j)}))^2 > \delta n^{-\frac{1}{4}}\right)=o(1)$, for any $\delta >0$, where $\pr_{\sB_1}(\cdot)=\pr(\cdot \cap \sB_1)$. To begin with, by the triangle inequality and Lemma \ref{lemma1}, we get 
\begin{align}\label{eq:dfineq1}
& \pr_{\sB_1}\left(\sum_{j \in I_\kappa}\left(\tilde{F}_{co}^{(0)}(Y_{(j)}) - \breve{F}_{co}^{(0)}(Y_{(j)})\right)^2 > \delta n^{-\frac{1}{4}} \right) \nonumber \\
& \leq \sum_{j \in  I_\kappa}\pr_{\sB_1}\left(\left(\tilde{F}_{co}^{(0)}(Y_{(j)}) - \breve{F}_{co}^{(0)}(Y_{(j)})\right)^2 >\delta n^{-\frac{5}{4}} \right) \nonumber \\
& = \sum_{j \in  I_\kappa} \pr_{\sB_1}\left(\left( \min_{\ell \geq j} \max_{k \leq j} f (k, \ell) \right)^2 > \delta n^{-\frac{5}{4}} \right) \nonumber \\
& \leq \sum_{j \in  I_\kappa} \pr_{\sB_1}\left(\left(\max_{k \leq j}f_+(k, j) \right)^2 > \frac{\delta n^{-\frac{5}{4}}}{2} \right)+\sum_{j \in  I_\kappa} \pr_{\sB_1}\left(\left(\max_{\ell \geq j}f_-(j, \ell) \right)^2 > \frac{\delta n^{-\frac{5}{4}}}{2} \right) \nonumber \\
& \leq \sum_{j \in  I_\kappa} \pr_{\sB_1}\left(\max_{k \leq j}f_+(k, j)  >\delta_0 n^{-\frac{5}{8}} \right)+ \sum_{j \in  I_\kappa} \pr_{\sB_1}\left( \max_{\ell \geq j}f_-(j, \ell)  >\delta_0 n^{-\frac{5}{8}}\right), 
\end{align}
where $\delta_0=\sqrt{\delta/2}$. 

\begin{lemma} Let $\varepsilon=\delta_0/n^{5/8}$. Then for any $j \in I_\kappa$, 
	\begin{align}\label{ineq1}
	\pr_{\sB_1}\left(\max_{\ell \in[j, \ceil{n(1-\kappa)}]} f_-(j, \ell)  \geq \varepsilon \right) \leq \sum_{\ell=j}^{\ceil{n(1-\kappa)}} \pr_{\sB_1}\left(\breve{F}_{co}^{(0)}(Y_{(j)})- \breve{F}_{co}^{(0)}(Y_{(\ell)}) \geq \varepsilon \right).
	\end{align}
	Similarly, 
	\begin{align}\label{ineq2}
	\pr_{\sB_1}\left(\max_{k \in[\ceil{n \kappa}, j]} f_+(k, j)  \geq \varepsilon \right) \leq \sum_{k=\ceil{n \kappa}}^{j} \pr_{\sB_1}\left(\breve{F}_{co}^{(0)}(Y_{(k)})- \breve{F}_{co}^{(0)}(Y_{(j)}) \geq \varepsilon \right).
	\end{align}
\end{lemma}

\begin{proof} Note that $$f_-(j, \ell) =\max(0, \breve{F}_{co}^{(0)}(Y_{(j)})-\frac{\breve{F}_{co}^{(0)}(Y_{(j)})+ \cdots + \breve{F}_{co}^{(0)}(Y_{(\ell)})}{\ell-j+1})=\max\left(0, \frac{\sum_{i=j}^{\ell}(\breve{F}_{co}^{(0)}(Y_{(j)})-\breve{F}_{co}^{(0)}(Y_{(i)})) }{\ell-j+1}\right).$$ Then 
	\begin{align*}
	&\pr_{\sB_1}\left(\max_{\ell \in[j, \ceil{n(1-\kappa)}]} f_{-}(j, \ell) \leq \varepsilon \right) \\
	&\pr_{\sB_1}\left(\max_{\ell \in[j, \ceil{n(1-\kappa)}]} \frac{\sum_{i=j}^{\ell}(\breve{F}_{co}^{(0)}(Y_{(j)})-\breve{F}_{co}^{(0)}(Y_{(i)})) }{\ell-j+1} \leq \varepsilon \right) \\
	&= \pr_{\sB_1}\left(\frac{1}{2}\sum_{i=j}^{j+1}(\breve{F}_{co}^{(0)}(Y_{(j)})-\breve{F}_{co}^{(0)}(Y_{(i)})) \leq \varepsilon, 
	\ldots, \frac{\sum_{i=j}^{n}\breve{F}_{co}^{(0)}(Y_{(j)})-\breve{F}_{co}^{(0)}(Y_{(i)})) }{n-j+1} \leq \varepsilon \right) \\
	& \geq \pr_{\sB_1}\left(\breve{F}_{co}^{(0)}(Y_{(j)}) - \breve{F}_{co}^{(0)}(Y_{(b+1)}) \leq 2\varepsilon, 
	\ldots, \breve{F}_{co}^{(0)}(Y_{(j)}) - \breve{F}_{co}^{(0)}(Y_{(\ceil{n(1-\kappa)})} \leq \varepsilon \right) \\
	& \geq \pr_{\sB_1}\left(\breve{F}_{co}^{(0)}(Y_{(j)}) - \breve{F}_{co}^{(0)}(Y_{(b+1)}) \leq \varepsilon, 
	\ldots, \breve{F}_{co}^{(0)}(Y_{(j)}) - \breve{F}_{co}^{(0)}(Y_{(\ceil{n(1-\kappa)})}) \leq \varepsilon \right) \\
	&= \pr_{\sB_1}\left(\max_{\ell \in[j, \ceil{n(1-\kappa)}]} (\breve{F}_{co}^{(0)}(Y_{(j)})-\breve{F}_{co}^{(0)}(Y_{(\ell)}) \leq \varepsilon \right).
	\end{align*}
	The above inequality and a union bound implies \eqref{ineq1}.

	Similarly, 
	\begin{align*}
	&\pr_{\sB_1}\left(\max_{k \in[\ceil{n \kappa}, j]} f_{+}(k, j) \leq \varepsilon \right) \\
	&\pr_{\sB_1}\left(\max_{k \in[\ceil{n \kappa}, j]} \frac{\sum_{i=k}^{j}(\breve{F}_{co}^{(0)}(Y_{(k)})-\breve{F}_{co}^{(0)}(Y_{(j)})) }{j-k+1} \leq \varepsilon \right) \\
	&= \pr_{\sB_1}\left(\frac{1}{2} \sum_{i=j-1}^{j}(\breve{F}_{co}^{(0)}(Y_{(i)})-\breve{F}_{co}^{(0)}(Y_{(j)})) \leq \varepsilon, 
	\ldots, \frac{\sum_{i=\ceil{n \kappa}}^{j}\breve{F}_{co}^{(0)}(Y_{(j)})-\breve{F}_{co}^{(0)}(Y_{(i)})) }{j} \leq \varepsilon \right) \\
	& \geq \pr_{\sB_1}\left(\breve{F}_{co}^{(0)}(Y_{(j-1)}) - \breve{F}_{co}^{(0)}(Y_{(j)}) \leq 2\varepsilon, 
	\ldots, \breve{F}_{co}^{(0)}(Y_{(\ceil{n \kappa})}) - \breve{F}_{co}^{(0)}(Y_{(j)} \leq \varepsilon \right) \\
	&= \pr_{\sB_1}\left(\max_{k \in[\ceil{n \kappa}, j]} (\breve{F}_{co}^{(0)}(Y_{(k)})-\breve{F}_{co}^{(0)}(Y_{(j)}) \leq \varepsilon \right).
	\end{align*}
	This implies \eqref{ineq2} by a union bound.
\end{proof}

In light of \eqref{eq:dfineq1} and above lemma, to prove Proposition \ref{proposition:pavabl}, we need bounds on the upper tail of the difference $\breve{F}_{co}^{(0)}(Y_{(j)})- \breve{F}_{co}^{(0)}(Y_{(\ell)})$, for $\ell \geq j$ and lower tail of $\breve{F}_{co}^{(0)}(Y_{(k)})- \breve{F}_{co}^{(0)}(Y_{(j)})$, for $k \leq j$. To this end, let $\breve \eta_{zd}=n_{zd}/n$, for $\{z, d\} \in \{0, 1\}$, and $\breve \eta_{-}=\min_{z, d}\breve\eta_{zd}$. 

\begin{lemma}\label{lm:diffprob} Let $\varepsilon=\delta_0/n^{5/8}$. Then for any $j \in I_\kappa$, such that 
	\begin{align}\label{c1}
	\sum_{\ell=j}^{\ceil{n(1-\kappa)}} \pr_{\sB_1}\left(\breve{F}_{co}^{(0)}(Y_{(j)})- \breve{F}_{co}^{(0)}(Y_{(\ell)}) \geq \varepsilon |(\bm Z, \bm D) \right) \leq O(1/n^3)+ 2 n \exp\left\{-\frac{ \breve \eta_{-}^{3/2} \breve \lambda_0^3 \delta_0^2  \sqrt n }{48 \log n}\right\},
	\end{align}
	where the constant in the $O(1/n^3)$-term is non-random. Similarly, 
	\begin{align}\label{c2}
	\sum_{k=\ceil{n \kappa}}^j \pr_{\sB_1}\left(\breve{F}_{co}^{(1)}(Y_{(k)})- \breve{F}_{co}^{(1)}(Y_{(j)}) \geq \varepsilon |(\bm Z, \bm D) \right) \leq O(1/n^3)+ 2 n \exp\left\{-\frac{ \breve \eta_{-}^{3/2} \breve \lambda_0^3 \delta_0^2  \sqrt n }{48 \log n}\right\}.
	\end{align}
\end{lemma}

The proof of the above lemma is given below. Using this, the proof of Proposition \ref{proposition:pavabl} can be easily completed as follows: Note that $$n^2 \exp\left\{-\frac{ \breve \eta_-^{3/2} \breve \lambda_0^3 \delta_0^2  \sqrt n }{48 \log n}\right\} \dto 0, \quad \E\left(n^2 \exp\left\{-\frac{ \breve \eta_-^{3/2} \breve \lambda_0^3 \delta_0^2  \sqrt n }{48 \log n}\right\} \right) \rightarrow 0,$$ by the dominated convergence theorem. Then 
first taking expectation over $\bm Z, \bm D$ gives 
\begin{align}
\sum_{\ell=j}^{\ceil{n(1-\kappa)}} \pr_{\sB_1}\left(\breve{F}_{co}^{(0)}(Y_{(j)})- \breve{F}_{co}^{(0)}(Y_{(\ell)}) \geq \varepsilon \right) \leq O(1/n^3)+ 2 n \E\exp\left\{-\frac{ \breve \eta_-^{3/2} \breve \lambda_0^3 \delta_0^2  \sqrt n }{48 \log n}\right\} .
\end{align}
Therefore, from \eqref{ineq2}, 
\begin{align}\label{eq:f1}
\pr_{\sB_1}\left(\max_{k \in[\ceil{n \kappa}, j]} f_+(k, j)  >\delta_0 n^{-\frac{5}{8}} \right) \leq O(1/n^2)+ 2 n^2 \E \exp\left\{-\frac{ \breve \eta_-^{3/2} \breve \lambda_0^3 \delta_0^2  \sqrt n }{48 \log n}\right\} \rightarrow 0.
\end{align}
Similarly, from \eqref{ineq1} and \eqref{c2}, it can be shown that 
\begin{align}\label{eq:f2}
\pr_{\sB_1}\left( \max_{k \in[j, \ceil{n (1-\kappa)}]} f_-(k, j)  >\delta_0 n^{-\frac{5}{8}} \right) \rightarrow 0.
\end{align}
Adding \eqref{eq:f1} and \eqref{eq:f2} and using \eqref{eq:dfineq1}, completes the proof of \eqref{eq:pavabl}, since $\pr(\sB_1^c)\rightarrow 0$. \\

\noindent \textbf{\textit{Proof of Lemma \ref{lm:diffprob}}} Throughout the proof, all events will be conditional of $(\bm Z, \bm D)$, and, for notational brevity, we will omit the conditioning event in all the expressions. Let $$L=H^{-1}(\kappa)-1 \quad \text{and} \quad R=H^{-1}(1-\kappa)+1,$$ and let $\sB_2$ be the event that $\{Y_{(\ceil{n\kappa})},  Y_{(\ceil{n(1-\kappa)})},  \in [L, R]\}$. Using $\E Y_{(\ceil{n\kappa})} \rightarrow H^{-1}(\kappa)$, $\E(Y_{(\ceil{n\kappa})}-\E Y_{(\ceil{n\kappa})})^6=O(1/n^3)$ (see \citet{sen1959}), and the Chebyshev's inequality it follows that 
\begin{align}\label{eq:probcomp2}
\pr(Y_{(\ceil{n\kappa})} \notin [L, R])  \leq \frac{\E(Y_{(\ceil{n\kappa})}-\E Y_{(\ceil{n\kappa})})^6}{(L-\E Y_{(\ceil{n\kappa})})^6}=O(1/n^3).
\end{align}
Similarly, $\pr(Y_{(\ceil{n(1-\kappa)})} \notin [L, R])=O(1/n^3)$, which gives $\pr(\sB_2^c)=O(1/n^3)$. 

Now, partition $[L, R]$ in a grid $L=t_0< t_1< \ldots < t_M=R$, of size $ \frac{4\log n}{C_1 n}$, that is, $t_a=L+a \left(\frac{4 \log n }{C_1 n} \right)$, for $0\leq a \leq M=\frac{C_1(R-L) n}{4 \log n}$. Define $N_a=\sum_{s=1}^n \bm 1\{Y_s \in (t_a, t_{a+1}]\}$, the number of observations in the interval $(t_a, t_{a+1}]$ (where $C_1$ and $C_2$ are the constants in Assumption 2 with $K=[L, R]$). Let $$\sB_3=\left\{\min_{0 \leq a \leq M} N_a \geq 1\right\} \cap \left\{\max_{0 \leq a \leq M} N_a \leq \frac{12 C_2}{C_1} \log n  \right\}.$$

\begin{lemma}\label{lm:probcomp3}
	$\pr(\sB_3^c)= O(1/n^3).$
\end{lemma}

\begin{proof} Note that $\pr(Y_j \notin (t_a, t_{a+1}])=1- (H(t_{a+1})- H(t_a)) \leq 1- \frac{4  \log n}{n}$, by Assumption 2 in the main text. Then 
	\begin{align}\label{eq:B1}
	\pr\left (\min_{0 \leq a \leq M} N_a=0 \right)=\sum_{a=0}^ M \pr(N_a=0) \lesssim_\kappa n\left(1-\frac{4 \log n}{n}\right)^n=O(1/n^3).
	\end{align}
	
	Next, note that $\E(N_a)=n (F(t_{a+1})-F(t_a)) \in [ 4 \log n,  \frac{4 C_2}{C_1}\log n]$, by Assumption 2. Therefore, by the union bound followed by a Chernoff bound,\footnote{Suppose $X_1, X_2, \ldots, X_n$ are independent random variables taking values in $\{0, 1\}$. Let $X=\sum_{i=1}^n X_i$ denote their sum and let $\mu = \E(X)$. Then $\pr(X \geq (1+\delta)\mu) \leq e^{-\frac{\delta^2 \mu}{3}}$, for $0 < \delta < 1$ and $\pr(X \geq (1+\delta)\mu) \leq e^{-\frac{\delta \mu}{2}}$, for $\delta \geq 1 $.} gives 
	\begin{align}\label{eq:B2}
	\pr\left(\max_{0 \leq a \leq M} N_a > \frac{12 C_2}{C_1} \log n\right)&=\sum_{a=0}^M \pr\left(N_a -\E(N_a)> \frac{8 C_2}{C_1} \log n \right) \nonumber \\
	& \lesssim n e^{-\frac{4 C_2 \log n}{C_1}}=n e^{-4 \log n }=O(1/n^3).
	\end{align}
	Combining \eqref{eq:B1} and \eqref{eq:B2} the proof of the lemma follows. 
\end{proof}

Finally, let $\sB_4=\{|\breve \eta_{zd}/\eta_{zd}-1|\leq 1: \text{ for all } z, d \in \{0, 1\} \}$, and set $\sB_0=\sB_2\cap \sB_3 \cap \sB_4$. From \eqref{eq:probcomp2} and  Lemma \ref{lm:probcomp3}, 
\begin{align}\label{eq:probcomplement}
\pr_{\sB_1}(\sB_0^c)\leq \pr(\sB_0^c) \leq  \pr(\sB_2^c)+ \pr(\sB_3^c) + \pr(\sB_3^4)=O(1/n^3).
\end{align} 
Therefore, it suffices to consider events on $\sB=\sB_0\cap\sB_1$. Now, fix $j \in I_\kappa$. For any $\ell \geq j$ denote by $I_{p(\ell)}=(t_{p(\ell)},~ t_{p(\ell)+1}]$ the interval which contains $Y_{(\ell)}$, and $\overline F_{zd}((a, b])=\overline F_{zd}(a)-F_{zd}(j)$, for $z, d \in \{0, 1\}$. Then, by triangle inequality, on the set $\sB$, 
\begin{align}\label{eq:tineq1}
|\overline F_{00}((t_{p(j)}, t_{p(\ell)+1}])-\overline F_{00}((Y_{(j)}, Y_{(\ell)}])| & \leq |\overline F_{00}(t_{p(j)})-\overline F_{00}(Y_{(j)})|+|\overline F_{00}(t_{p(\ell)+1})-\overline F_{00}(Y_{(\ell)})| \nonumber \\
&=O\left(\frac{\log n}{n_{00}}\right)=O\left(\frac{\log n}{n}\right).
\end{align}

Now, take $\varepsilon=\delta_0/\sqrt n$. Then recalling the definition of $\breve{F}_{co}^{(0)}(Y_{(j)}) = \frac{\overline{F}_{00}(Y_{(j)})-(1-\breve \lambda_0)\overline{F}_{10}(Y_{(j)})}{\breve \lambda_0}$, and using triangle inequality gives, 
\begin{align*}
\pr_{\sB}& \left(\breve{F}_{co}^{(0)}(Y_{(j)})- \breve{F}_{co}^{(0)}(Y_{(\ell)}) \geq \varepsilon \right) \\
& = \pr_{\sB}\left(\frac{\overline F_{00}((Y_{(j)}, Y_{(\ell)}])}{\breve \lambda_0} -\frac{(1-\breve \lambda_0) (\overline F_{01}((Y_{(j)}, Y_{(\ell)}])) }{\breve \lambda_0} \geq \varepsilon \right) \\
& \leq \pr_{\sB}\left(\frac{\overline F_{00}((t_{p(j)}, t_{p(\ell)+1}])}{\breve \lambda_0} -\frac{1-\breve \lambda_0(\overline F_{01}((t_{p(j)}, t_{p(\ell)+1}]))}{\breve \lambda_0}  \geq \frac{\varepsilon}{2}  \right) \tag*{(by \eqref{eq:tineq1})}\\
& \leq T_1+T_2,
\end{align*}
where 
$$T_1=\pr\left(|\overline F_{00}((t_{p(j)}, t_{p(\ell)+1}])-F_{00}((t_{p(j)}, t_{p(\ell)+1}])| \geq \frac{\breve \lambda_0(\frac{\varepsilon}{2} -F_{co}((t_{p(j)}, t_{p(\ell)+1}])}{2}  \right) $$
and 
$$T_2=\pr\left(|\overline F_{01}((t_{p(j)}, t_{p(\ell)+1}])-F_{01}((t_{p(j)}, t_{p(\ell)+1}])| \geq \frac{\breve \lambda_0(\frac{\varepsilon}{2} -F_{co}((t_{p(j)}, t_{p(\ell)+1}])}{2 (1-\breve \lambda_0)}  \right).$$

Now, we will bound $T_1$. To begin with note that $$-n_{00}\overline F_{00}((t_{p(j)}, t_{p(\ell)+1}]) \sim \dBin(n_{00}, - F_{00}((t_{p(j)}, t_{p(\ell)+1}])).$$ Moreover, by Assumption 2, $ - F_{co}((t_{p(j)}, t_{p(\ell)+1}]) \geq C_1(t_{p(\ell)+1}-t_{p(j)})\geq K (\ell-j)/n$, for some constant $K >0$. Then for $|\ell-j|> \frac{4}{K \breve \lambda_0 \sqrt{\breve \eta_{00}}} \sqrt n \log n$, where $\breve \eta_{00}=n_{00}/n$, we have  
\begin{align*}
&\pr\left(|\overline F_{00}((t_{p(j)}, t_{p(\ell)+1}])-F_{00}((t_{p(j)}, t_{p(\ell)+1}])| \geq \frac{\breve \lambda_0(\frac{\varepsilon}{2} -F_{co}((t_{p(j)}, t_{p(\ell)+1}])}{2}  \right) \\
& \leq \pr\left(|\overline F_{00}((t_{p(j)}, t_{p(\ell)+1}])-F_{00}((t_{p(j)}, t_{p(\ell)+1}])| \geq -\frac{\breve \lambda_0 F_{co}((t_{p(j)}, t_{p(\ell)+1}])}{2}  \right)  \\
& \leq \pr\left(|n_{00}\overline F_{00}((t_{p(j)}, t_{p(\ell)+1}])-n_{00}F_{00}((t_{p(j)}, t_{p(\ell)+1}])| \geq \frac{K \breve \lambda_0 n_{00}(\ell-j)}{2n}  \right)  \\
& \leq 2 e^{-\frac{\breve \eta_{00} K^2\breve \lambda_0^2 (\ell-j)^2}{2n}}=O(1/n^8),
\end{align*}
where the last step follows by the Hoeffding's inequality.

Now, suppose $|\ell-j| \leq \frac{4}{K \breve \lambda_0 \sqrt{\breve \eta_{00}}} \sqrt n \log n$. Let $t=-\frac{\breve \lambda_0 \varepsilon}{2 F_{00}((t_{p(j)}, t_{p(\ell)+1}])}$. Then 
\begin{align*}
& \pr\left(|\overline F_{00}((t_{p(j)}, t_{p(\ell)+1}])-F_{00}((t_{p(j)}, t_{p(\ell)+1}])| \geq \frac{\breve \lambda(\varepsilon - F_{co}((t_{p(j)}, t_{p(\ell)+1}])}{2}  \right) \\
& \leq \pr\left(|\overline F_{00}((t_{p(j)}, t_{p(\ell)+1}])-F_{00}((t_{p(j)}, t_{p(\ell)+1}])| \geq \frac{\breve \lambda_0 \varepsilon}{2}  \right) \\
& \leq \pr\left(|n_{00}\overline F_{00}((t_{p(j)}, t_{p(\ell)+1}])-n_{00}F_{00}((t_{p(j)}, t_{p(\ell)+1}])| \geq -t n_{00} F_{00}((t_{p(j)}, t_{p(\ell)+1}])\right) \\ 
& \leq 2\exp\left\{\frac{ t^2 \breve \eta_{00} n F_{00}((t_{p(j)}, t_{p(\ell)+1}]) }{3}\right\} \tag*{(by Chernoff bound)}\\
& \leq 2\exp\left\{\frac{\breve \eta_{00} \breve \lambda_0^2 \delta_0^2  }{12 F_{00}((t_{p(j)}, t_{p(\ell)+1}]) }\right\} \\
& \leq 2\exp\left\{-\frac{ \breve \eta_{00} \breve \lambda_0^2 \delta_0^2  n}{12 K (\ell-j) }\right\} \tag*{(since $ - F_{co}((t_{p(j)}, t_{p(\ell)+1}]) \geq K (\ell-j)/n$)} \\
& \leq 2\exp\left\{-\frac{ \breve \eta_{-}^{3/2} \breve \lambda_0^3 \delta_0^2  \sqrt n }{48 \log n}\right\} \tag*{(recall $ \breve \eta_{-}=\min_{z, d}\breve \eta_{zd}$).}
\end{align*}
This implies $T_1 \leq O(1/n^8)+2\exp\left\{-\frac{ \breve \eta_{-}^{3/2} \breve \lambda_0^3 \delta_0^2  \sqrt n }{48 \log n}\right\}$, and similarly, for $T_2$. These combined with $\pr_{\sB_1}(\sB_0^c) =O(1/n^3)$ completes the proof of Lemma \ref{lm:diffprob}. \hfill $\Box$.

\subsection{Completing the proof of Theorem 1}
\label{sec:pfBLplugin}

In this section we complete the proof of Theorem 1. To begin, we compute the difference $\M_n(\breve{\bm F}, \breve {\bm \phi}) - \M_n(\hat{\bm F}, \hat {\bm \phi})$, 
\begin{align}
& \M_n(\breve{\bm F}, \breve {\bm \phi}) - \M_n(\hat{\bm F}, \hat {\bm \phi}) \nonumber \\
&= \frac{1}{|I_{\kappa}|} \left[ \sum_{z, d \in \{0,1\}} \sum_{j \in I_\kappa} \frac{n_{zd}}{n} \left\{ \overline F_{zd}(Y_{(j)}) \log \frac{\overline F_{zd}(Y_{(j)})}{\hat F_{zd}(Y_{(j)})} + (1-\overline F_{zd}(Y_{(j)}) \log \frac{1-\hat F_{zd}(Y_{(j)})}{1-\hat F_{zd}(Y_{(j)})} \right\}\right] \nonumber \\
& +\frac{1}{|I_{\kappa}|} \sum_{j \in I_\kappa}  \left\{ \frac{n_{00}}{n} \log \frac{1-\breve \phi_{at}}{1-\hat \phi_{at}(Y_{(j)})}+\frac{n_{01}}{n} \log \frac{\breve \phi_{at}}{\hat \phi_{at}(Y_{(j)})} +  \frac{n_{10}}{n} \log \frac{\breve \phi_{nt}}{\hat \phi_{nt}(Y_{(j)})}+ \frac{n_{11}}{n}  \log \frac{1-\breve \phi_{nt}}{1-\hat \phi_{nt}(Y_{(j)})} \right \}, \nonumber \\
&=\frac{1}{|I_{\kappa}|} \left\{ \sum_{z, d \in \{0, 1\}} \frac{n_{zd}}{n} \sum_{j \in I_\kappa}T_{zd}(Y_{(j)}) \right\} +  \frac{1}{|I_{\kappa}|} \sum_{j \in I_\kappa}   \left\{ \frac{n_{0}}{n} R_{nt}(Y_{(j)}) + \frac{n_{1}}{n} R_{at}(Y_{(j)}) \right\},
\end{align}
where  
\begin{align*}
T_{zd}(Y_{(j)})&=\overline F_{zd}(Y_{(j)}) \log \frac{ \overline F_{zd}(Y_{(j)}) }{\hat F_{zd}(Y_{(j)})} +(1-\overline F_{zd}(Y_{(j)}) )\log\frac{1- \overline  F_{zd}(Y_{(j)})}{1- \hat F_{zd}(Y_{(j)})}, \\
R_{at}(Y_{(j)})& = (1-\breve \phi_{at}) \log \frac{1-\breve \phi_{at}}{1-\hat \phi_{at}(Y_{(j)})}+\breve \phi_{at} \log \frac{\breve \phi_{at}}{\hat \phi_{at}(Y_{(j)})}, \\
R_{nt}(Y_{(j)})&= (1-\breve \phi_{nt}) \log \frac{1-\breve \phi_{nt}}{1-\hat \phi_{nt}(Y_{(j)})}+\breve \phi_{nt} \log \frac{\breve \phi_{nt}}{\hat \phi_{nt}(Y_{(j)})}.
\end{align*}

Now,  using $a \log\frac{a}{x} + (1-a) \log \frac{1-a}{1-x} \geq  \frac{1}{2}(x-a)^2$ (Observation \ref{obs:cdiff}) gives $T_{zd}(Y_{(j)}) \gtrsim (\hat F_{zd}(Y_{(j)}) - \overline F_{zd}(Y_{(j)}))^2$, $R_{at}(Y_{(j)})\gtrsim (\hat \phi_{at}(Y_{(j)})-\breve \phi_{at})^2$ and $R_{nt}(Y_{(j)})\gtrsim (\hat \phi_{nt}(Y_{(j)})-\breve \phi_{nt})^2$.  Therefore, 
\begin{align*}
\M_n(\breve{\bm F}, \breve {\bm \phi}) - \M_n(\hat{\bm F}, \hat {\bm \phi})  & \gtrsim \frac{1}{|I_{\kappa}|}\sum_{z, d \in\{0, 1\}} \frac{n_{zd}}{n} \sum_{j \in I_\kappa} (\hat F_{zd}(Y_{(j)})-  \overline{F}_{zd}(Y_{(j)}))^2 \\
&+ K_n \frac{1}{|I_{\kappa}|}\sum_{j \in I_\kappa} \left( \hat {\bm \phi}(Y_{(j)})-\breve{\bm \phi} (Y_{(j)})\right)^2,
\end{align*}
for some constant $K_n \pto K >0$. Therefore, using $\M_n(\breve{\bm F}, \breve {\bm \phi}) - \M_n(\hat{\bm F}, \hat {\bm \phi}) \leq \M_n(\breve{\bm F}, \breve {\bm \phi}) - \M_n(\tilde{\bm F}, \breve {\bm \phi})$,  since $(\tilde{\bm F}, \breve {\bm \phi}) \in \bm \vartheta_+ \times \bm \varphi_+$  on the set $\sB_1$  gives
\begin{align}\label{eq:t1}
&\frac{1}{|I_{\kappa}|}\sum_{z, d \in\{0, 1\}} \frac{n_{zd}}{n} \sum_{j \in I_\kappa} (\hat F_{zd}(Y_{(j)})-  \overline{F}_{zd}(Y_{(j)}))^2 + K_n \frac{1}{|I_{\kappa}|}\sum_{j \in I_\kappa} \left( \hat {\bm \phi}(Y_{(j)})-\breve{\bm \phi} (Y_{(j)})\right)^2\\
 & \lesssim  \M_n(\breve{\bm F}, \breve {\bm \phi}) - \M_n(\hat{\bm F}, \hat {\bm \phi}) \nonumber \\
& \lesssim \M_n(\breve{\bm F}, \breve {\bm \phi}) - \M_n(\tilde{\bm F}, \breve {\bm \phi}) \nonumber \\ 
&=o_P(n^{-\frac{5}{4}}), 
\end{align}
by Proposition \ref{proposition:pavabl}. 

Therefore, \eqref{eq:t1} implies $\frac{1}{|I_{\kappa}|}\sum_{j \in I_\kappa} \left\| \sqrt n \left(\hat {\bm \phi}(Y_{(j)})-\breve{\bm \phi}(Y_{(j)}) \right) \right\|_2^2=o_P(n^{-\frac{1}{4}})$. Furthermore, by Jensen's inequality, this implies 
\begin{align*}
\sqrt n \left( \frac{1}{|I_{\kappa}|}\sum_{j \in I_\kappa}\hat {\phi}_{nt}(Y_{(j)})- \frac{1}{|I_{\kappa}|}\sum_{j \in I_\kappa} \breve{\phi}_{nt}(Y_{(j)}) \right) &=o_P(1) \\
\sqrt n \left( \frac{1}{|I_{\kappa}|}\sum_{j \in I_\kappa}\hat {\phi}_{at}(Y_{(j)})- \frac{1}{|I_{\kappa}|}\sum_{j \in I_\kappa} \breve{\phi}_{at}(Y_{(j)}) \right) &=o_P(1).
\end{align*}
 Moreover, 
\begin{align}\label{eq:diffFatnt}
&\frac{1}{|I_{\kappa}|}\sum_{j \in I_\kappa} \left\{ \sqrt n(\hat F_{nt}(Y_{(j)})-  \breve {F}_{nt}(Y_{(j)}))\right\}^2=o_P(n^{-\frac{1}{4}}), \nonumber \\
&\frac{1}{|I_{\kappa}|}\sum_{j \in I_\kappa} \left\{\sqrt n (\hat F_{at}(Y_{(j)})-  \breve {F}_{at}(Y_{(j)})) \right\}^2=o_P(n^{-\frac{1}{4}}),
\end{align}
since  $\hat F_{10}=\hat F_{nt}$, $\overline F_{10}=\breve F_{nt}$, $\hat F_{01}=\hat F_{at}$, $\overline F_{01}=\breve F_{at}$, and $n_{zd}/n \pto \eta_{zd}$. Next, define $\hat{\lambda}_{0}(t) = \frac{1-\hat{\chi}_{nt}(t) - \hat{\chi}_{at}(t)}{1-\hat{\chi}_{at}(t)}$ and $\hat{\lambda}_{1}(t) = \frac{1-\hat{\chi}_{nt}(t) - \hat{\chi}_{at}(t)}{1-\hat{\chi}_{nt}(t)}$. Observe that on $\sB_1$, $|\breve F_{00}(t)| \leq 1$, and 
\begin{align}\label{eq:fdiffco}
& (\hat F_{co}^{(0)}(t)-  \breve {F}_{co}^{(0)}(t))^2 \nonumber \\
& =\left(\frac{\hat F_{00}(t)-(1-\hat \lambda_0 (t)) \hat F_{nt}(t)}{\hat \lambda_0(t)}-\frac{\breve F_{00}(t)-(1-\breve \lambda_0) \breve F_{nt}(t)}{\breve \lambda_0}\right) ^2\nonumber \\
&\lesssim \left(\frac{\hat F_{00}(t)}{\hat \lambda_0(t)}-\frac{\breve F_{00}(t)}{\breve \lambda_0} \right)^2 +\left(\frac{(1-\hat \lambda_0(t)) \hat F_{nt}(t)}{\hat \lambda_0(t)}-\frac{(1-\breve \lambda_0) \breve F_{nt}(t)}{\breve \lambda_0}\right)^2 \nonumber \\
&= O_P(1) (\hat F_{00}(t)-\breve F_{00}(t))^2 + O_P(1) (\hat F_{nt}(t)- \breve F_{nt}(t))^2 +\left(\frac{1}{\hat \lambda_0(t)}-\frac{1}{\breve \lambda_0}\right)^2.
\end{align}
We take the sum of \eqref{eq:fdiffco} over $t=Y_{(j)}$, $j \in I_\kappa$ and use the fact  $\left(\frac{1}{\hat \lambda_0(t)}-\frac{1}{\breve \lambda_0}\right)^2=o_P(n^{-\frac{5}{4}}) $ for all $t$. From \eqref{eq:t1}, \eqref{eq:diffFatnt} and $\frac{1}{|I_{\kappa}|}\sum_{j \in I_\kappa} \left\| \sqrt n \left(\hat {\bm \phi}(Y_{(j)})-\breve{\bm \phi}(Y_{(j)}) \right) \right\|_2^2=o_P(n^{-\frac{1}{4}})$, we have
$$
\frac{1}{|I_{\kappa}|}\sum_{j \in I_\kappa} \left\{ \sqrt n (\hat F_{co}^{(0)}(Y_{(j)})-  \breve {F}_{co}^{(0)}(Y_{(j)}) \right\}^2=o_P(n^{-\frac{1}{4}}).
$$
Similarly, we can show that $\frac{1}{|I_{\kappa}|}\sum_{j \in I_\kappa} \left\{ \sqrt n (\hat F_{co}^{(1)}(Y_{(j)})-  \breve {F}_{co}^{(1)}(Y_{(j)}) \right\}^2=o_P(n^{-\frac{1}{4}})$. The following lemma summarizes these findings together with \eqref{eq:diffFatnt}:

\begin{lemma}\label{lm:BLpluginsqrtn} The maximum binomial likelihood estimates $(\hat {\bm  F}, \hat {\bm \phi})$ and the plug-in estimates $(\breve {\bm F},  \breve {\bm \phi})$ satisfy 
	\begin{align*}
	\frac{1}{|I_{\kappa}|}\sum_{j \in I_\kappa} \left\| \sqrt n \left\{\hat {\bm \phi}(Y_{(j)})-\breve{\bm \phi}(Y_{(j)}) \right\}\right\|^2_2 &=o_P(n^{-\frac{1}{4}}) \\
	\frac{1}{|I_{\kappa}|}\sum_{j \in I_\kappa} \left\| \sqrt n \left\{ \bm{\hat{F}}(Y_{(j)})-\breve {\bm F}(Y_{(j)}) \right\}\right\|^2_2 &= o_P(n^{-\frac{1}{4}}).
	\end{align*}
\end{lemma}

This shows the maximum binomial likelihood estimates and the plug-in estimates are close in average squared error with respect to the empirical distribution $\overline{H}(t)=\sum_{z, d \in \{0, 1\}} \frac{n_{zd}}{n} \overline F_{zd}(t)$. To complete the proof of Theorem 1, we need to show that the average with respect to the empirical distribution can be replaced by the average (integral) with respect to the population distribution function $H(t)=\sum_{z, d \in \{0, 1\}} \eta_{zd} \overline F_{zd}(t)$.

\begin{lemma}\label{lm:BLpluginintegral} The maximum binomial likelihood estimates $\hat {\bm  F}$ and the plug-in estimates $\breve {\bm F}$ satisfy 
	$$ \int_{J_\kappa} ||\sqrt n\{\bm{\hat{F}}(t)-\breve {\bm F}(t)\}||^2_2 \mathrm d H= o_P(1),$$ 
	where the $o_P(1)$ terms goes to zero as $n \rightarrow \infty$.
\end{lemma}

\begin{proof}For $j \in I_\kappa$ and $Y_{(j)} \leq t < Y_{(b+1)}$, $\breve {\bm F}(t)=\breve {\bm F}(Y_{(j)})$. Moreover, $\bm{\hat{F}}(t) \leq \bm{\hat{F}}(Y_{(b+1)})$, where the inequality holds coordinate-wise. Then,  for $Y_{(j)} \leq t < Y_{(b+1)}$, 
	\begin{align*}
	||\bm{\hat{F}}(t)-\breve {\bm F}(t)||_2 & \leq ||\bm{\hat{F}}(Y_{(j)})-\breve {\bm F}(Y_{(j)})||_2 + ||\bm{\breve{F}}(Y_{(j)})-\breve {\bm F}(Y_{(b+1)})||_2 \\
	& \leq ||\bm{\hat{F}}(Y_{(j)})-\breve {\bm F}(Y_{(j)})||_2 + O(1/n).
	\end{align*}
	Therefore, 
	\begin{align}\label{eq:yb}
	\sum_{j \in I_\kappa} \int_{Y_{(j)}}^{Y_{(b+1)}} & ||\sqrt n\{\bm{\hat{F}}(t)-\breve {\bm F}(t)\}|| ^2_2  \mathrm d H \nonumber \\
	& \lesssim_\kappa \Delta n \sum_{j \in I_\kappa} ||\bm{\hat{F}}(Y_{(j)})-\breve {\bm F}(Y_{(j)})||^2_2+ \Delta,
	\end{align}
	where $\Delta=\sup_{j \in [n]}(H(Y_{(b+1)})-H(Y_{(j)}))\stackrel{D}=\sup_{j \in [n]}(U_{(b+1)}-U_{(j)})=O_P(\log n/n)$, by \citet{holst1980}. Then, by \eqref{eq:yb} 
	\begin{align*}
	\int_{Y_{(\ceil{n \kappa})}}^{Y_{(\ceil{n (1-\kappa)})}}  ||\sqrt n\{\bm{\hat{F}}& (t)-\breve {\bm F}(t)\}||^2_2 \mathrm d H \\
	& =\sum_{j \in I_\kappa} \int_{Y_{(j)}}^{Y_{(b+1)}} ||\sqrt n\{\bm{\hat{F}}(t)-\breve {\bm F}(t)\}||^2_2 \mathrm d H  \\
	& \leq O_P(\log n)   \sum_{j \in I_\kappa} ||\{\bm{\hat{F}}(Y_{(j)})-\breve {\bm F}(Y_{(j)})\}||^2_2+o_P(1) \\
	&=o_P(1),
	\end{align*} 
	where the last step follows from Lemma \ref{lm:BLpluginsqrtn}. 
	
	To complete the proof we need to take care of the boundary effects. As before, by triangle inequality
	\begin{align}\label{eq:kappaI}
	\int_{H^{-1}(\kappa)}^{Y_{\ceil{n\kappa}}}  & ||\sqrt n\{\bm{\hat{F}} (t)-\breve {\bm F}(t)\}||^2_2 \mathrm d H  \nonumber \\
	& \lesssim ||\sqrt n\{\bm{\hat{F}} (Y_{\ceil{n\kappa}})-\breve {\bm F}(Y_{\ceil{n\kappa}})\}||^2_2 (Y_{\ceil{n\kappa}}-\kappa)+o_P(1) \nonumber \\
	&=o_P(1),
	\end{align}
	where the last step uses $\sqrt n\{\bm{\hat{F}} (Y_{\ceil{n\kappa}})-\breve {\bm F}(Y_{\ceil{n\kappa}})\}=o_P(1)$ (which follows by a simple modification of the proof of Theorem 1) and $\sqrt n (Y_{\ceil{n\kappa}}-\kappa)=O_P(1)$. Similarly, it can be shown that 
	\begin{align}\label{eq:kappaII}
	\int_{Y_{\ceil{n(1-\kappa)}}}^{H^{-1}(1-\kappa)} ||\sqrt n\{\bm{\hat{F}} (t)-\breve {\bm F}(t)\}||^2_2 \mathrm d H=o_P(1).
	\end{align} 
	
	The proof now follows by combining \eqref{eq:kappaI} and \eqref{eq:kappaII} with \eqref{eq:yb}.
\end{proof}

\section{Limiting distribution of the maximum likelihood estimates}
\label{sec:pfBLdist}

As the limiting distribution of the empirical distributions $\overline F_{zd}$ are well-known, Theorem 1 can be used to derive the limiting distribution of the maximum binomial likelihood estimates $\hat{\bm F}$.

\begin{corollary}\label{corollary:BLnulldist} Fix $0 <\kappa <1/2$. Then  for any continuous function $h: \R \rightarrow \R$, 
	\begin{align}\label{eq:BLestlimit}
	\int_{J_\kappa}  h(t) \cdot \sqrt n\{\hat{\bm F}(t)-\E(\breve{\bm F}(t)|\bm Z, \bm D)\} \mathrm d H \dto  \int_{J_\kappa} h(t) \cdot \bm G(t) \mathrm d H ,
	\end{align}
	where
	\begin{align*}
	\bm G(t)= \begin{pmatrix}
	\frac{1}{\phi_{co}} \left\{\sqrt{\frac{\phi_{co}+\phi_{nt}}{\phi_0}}B_{00}(F_{00}(t)) - \sqrt{\frac{\phi_{nt}}{\phi_1}}B_{10}(F_{10}(t)) \right\} \\
	\frac{B_{01}(F_{01}(t))}{\sqrt{\phi_0\phi_{at}}}\\
	\frac{1}{\phi_{co}} \left\{\sqrt{\frac{\phi_{co}+\phi_{at}}{\phi_1}}B_{11}(F_{11}(t)) - \sqrt{\frac{\phi_{at}}{\phi_0}}B_{01}(F_{01}(t)) \right\} \\
	\frac{B_{10}(F_{10}(t))}{\sqrt{\phi_0\phi_{nt}}} ,
	\end{pmatrix},
	\end{align*}
	and $B_{00}(\cdot), B_{01}(\cdot), B_{10}(\cdot)$, and $B_{11}(\cdot)$ are independent standard Brownian bridges, and the integrals in \eqref{eq:BLestlimit} are defined coordinate-wise. 
	\label{COR:BLDIST}
\end{corollary}


\noindent \textbf{ \textit{Proof of Corollary~\ref{COR:BLDIST}.}} The joint distribution of the process $\sqrt n (\overline{F}_{zd}(t) - F_{zd}(t))_{z, d \in \{0, 1\}}$ can be easily derived from empirical process theory. To this end, let $D[0, 1]$ be the space of all right-continuous functions on $[0,1]$ with left limits equipped with the supremum norm metric. A sequence of  random functions $\{X_n(\cdot)\}_{n \geq 1}$ in $D[0, 1]$ converges to $X(\cdot) \in D[0, 1]$, denoted by $X_n(t) \stackrel{w}\Rightarrow X(t)$, if $\E(f(X_n)) \rightarrow \E f(X)$, for all bounded continuous function $f: D[0, 1] \rightarrow \R$. Now, considering $\sqrt n (\overline{F}_{zd}(t) - F_{zd}(t))_{z, d \in \{0, 1\}}$ as a random element of $D[0, 1]^4$ equipped with the product topology, we have the following result:

\begin{lemma}\label{lm:Fhatdist} Let $B_{00}, B_{01}, B_{10}$, and $B_{11}$ be independent Brownian bridges. Then 
	\begin{align} 
	\sqrt n \begin{pmatrix}
	\overline{F}_{00}(t) - F_{00}(t) \\
	\overline{F}_{01}(t) - F_{01}(t) \\
	\overline{F}_{11}(t) - F_{11}(t) \\
	\overline{F}_{10}(t) - F_{10}(t) 
	\end{pmatrix}    \stackrel{w}\Rightarrow   \begin{pmatrix}
	\frac{B_{00}(F_{00}(t))}{\sqrt{\eta_{00}}} \\
	\frac{B_{01}(F_{01}(t))}{\sqrt{\eta_{01}}} \\
	\frac{B_{11}(F_{11}(t))}{\sqrt{\eta_{11}}} \\
	\frac{B_{10}(F_{10}(t))}{\sqrt{\eta_{10}}}
	\end{pmatrix}    
	\end{align} 
\end{lemma}

\begin{proof} Note that $\E(\breve{\bm F}(t)|\bm Z, \bm D)\}$ is the mean of the plug-in estimate conditional on the sigma-algebra generated by $(\bm Z, \bm D)$. Conditioning on this sigma-algebra, $\{n_{zd}\}_{z, d \in \{0, 1\}}$ are fixed, and $\E(\overline F_{zd}(t)|\bm Z, \bm D)=\mathbb{P}(Y_1 \leq t|Z_1=z, D_1=d)=F_{zd}(t)$. Moreover, if $s<t$, $\Cov(\overline F_{zd}(s), \overline F_{zd}(t)|\bm Z, \bm D)=\frac{1}{n_{zd}} F_{zd}(s) (1-F_{zd}(t))$, and for $zd \ne z'd'$, $\Cov(\overline F_{zd}(s), \overline F_{z'd'}(t)|\bm Z, \bm D)=0$ since, 
	\begin{align}
	&\E\overline F_{zd}(t) \overline F_{z'd'}(t) \nonumber \\
	&=\frac{1}{n_{zd} n_{z'd'}}\sum_{\substack{i, i' \>\text{with}\> \\ Z_i=z, D_i=d, Z_{i'}=z', D_{i'}=d'}} \pr(Y_i \leq s, Y_{i'} \leq t |Z_i=z, D_i=d, Z_{i'}=z', D_{i'}=d') \nonumber \\
	&=F_{zd}(s) F_{z'd'}(t),
	\end{align}
	whenever $zd \ne z'd'$.

	Now, it is well-known that for each $z$ and $d$, 
	$$
	\sqrt {n_{zd}} (\overline F_{zd}(t)- F_{zd}(t))) \stackrel{w}\Rightarrow (B_{zd}(F_{zd}(t))_{\{z, d\}\in \{0, 1\}}
	$$ and therefore, 
	$$(\sqrt {n_{zd}} (\overline F_{zd}(t)- F_{zd}(t)))_{\{z, d\}\in \{0, 1\}} \Rightarrow  (B_{zd}(F_{zd}(t))_{\{z, d\}\in \{0, 1\}}.
	$$ Then, $$(\sqrt {n} (\overline F_{zd}(t)- F_{zd}(t)))_{\{z, d\}\in \{0, 1\}} \stackrel{w}\Rightarrow  \left(\frac{B_{zd}(F_{zd}(t))}{\sqrt{\eta_{zd}}}\right)_{\{z, d\}\in \{0, 1\}}$$ and the result follows.
\end{proof}

For $\bm \chi(t) =(\chi_{nt}(t), \chi_{at}(t)) \in \R^2$, define 
$$C(\bm \chi(t))=\begin{pmatrix}
\frac{1-\chi_{at}(t)}{ 1-\chi_{nt}(t)-\chi_{at}(t)} & 0 & 0& -\frac{{\chi}_{nt}(t)}{ 1-\chi_{nt}(t)-\chi_{at}(t)}  \\
0 & 1 & 0 & 0 \\
0 & -\frac{{\chi}_{at}(t)}{ 1-\chi_{nt}(t)-\chi_{at}(t)} & \frac{1-\chi_{nt}(t)}{  1-\chi_{nt}(t)-\chi_{at}(t)}& 0  \\
0 & 0 & 0 & 1
\end{pmatrix}.
$$
Note that $\E(\breve{\bm F}(t)|\bm Z, \bm D)=C(\breve {\bm \phi}) (F_{00}(t), F_{01}(t), F_{11}(t), F_{10}(t))' $ and $$\bm G(t)=C(\bm \phi) \left(\frac{B_{00}(F_{00}(t))}{\sqrt{\eta_{00}}}, \frac{B_{01}(F_{01}(t))}{\sqrt{\eta_{01}}}, \frac{B_{11}(F_{11}(t))}{\sqrt{\eta_{11}}}, \frac{B_{10}(F_{10}(t))}{\sqrt{\eta_{10}}}\right)',$$ where $\bm G(\cdot)$ is as defined in the statement of Corollary \ref{corollary:BLnulldist}. Now, using the above lemma and the Donsker's invariance principle, and noting that $C(\breve {\bm \phi})\pto C(\bm \phi)$, it follows that 
\begin{align}\label{eq:plugindist}
\int_{J_\kappa} & h(t) \cdot \sqrt n (\breve{\bm F}(t)-\E(\breve{\bm F}(t)|\bm Z, \bm D)) \mathrm d  H \nonumber \\
& = \int_{J_\kappa} h(t) \cdot C(\breve {\bm \phi}) \cdot \sqrt{n} \begin{pmatrix}
\overline{F}_{00}(t) - F_{00}(t) \\
\overline{F}_{01}(t) - F_{01}(t) \\
\overline{F}_{11}(t) - F_{11}(t) \\
\overline{F}_{10}(t) - F_{10}(t) 
\end{pmatrix}  \mathrm d  H  \nonumber \\
&\dto  \int_{J_\kappa} h(t) \cdot \bm G(t)  \mathrm d  H,
\end{align}
for any continuous function $h: \R\rightarrow \R$. This implies 
\begin{align}
& \int_{J_\kappa}  h(t) \cdot \sqrt n\{\hat{\bm F}(t)-\E(\breve{\bm F}(t)|\bm Z, \bm D)\}  \mathrm d H \nonumber \\
&= \int_{J_\kappa}   h(t) \cdot \sqrt n\{\breve{\bm F}(t)-\E(\breve{\bm F}(t)|\bm Z, \bm D)\}  \mathrm d H +  \int_{J_\kappa}  h(t) \cdot \sqrt n\{\hat{\bm F}(t)-\breve{\bm F}(t) \}  \mathrm d H   \nonumber \\
& \dto  \int_{J_\kappa}  \bm G(t)  \mathrm d  H,
\end{align}
using \eqref{eq:plugindist} for the first term, and second term is $o_P(1)$ by applying the Cauchy-Schwarz inequality followed by Theorem 1.

\section{The maximum binomial likelihood estimates under the null}
\label{sec:under_null}

In this section we analyze the maximum binomial likelihood estimate of the distribution functions of the compliance classes under the null. To this end, define
$$(\hat{\bm \psi}, \hat{\bm \xi})=\arg\max_{(\bm \theta, \bm \chi) \in \bm \vartheta_{+, 0}\times \bm \varphi_+ } \M_n(\bm \theta, \bm \chi), $$
where $\bm \psi = (\psi_{co}, \psi_{nt}, \psi_{at})$ and $\bm \xi = (\xi_{nt}, \xi_{at})$. 

Next, define the {\it population objective  function},\footnote{Note that we are slightly abusing terminology here, because the population objective function depends on the sample $\{Y_1, Y_2, \ldots, Y_n\}$. Ideally, one should define $\M(\cdot, \cdot)$ as an integral with respect to the population distribution function $H$. However, for technical reasons, it is more convenient for us to define $\M(\cdot, \cdot)$ with respect to the empirical measure instead.}
under the null hypothesis, as follows: 
\begin{align}\label{eq:Mnull}
\M(\bm\theta, \bm \chi)=T_{00}(\bm\theta, \bm \chi)  +  T_{10}(\bm\theta, \bm \chi)  + T_{10}(\bm\theta, \bm \chi) + T_{11}(\bm\theta, \bm \chi),
\end{align}
where
\footnotesize{
\begin{align*}
T_{00}(\bm\theta, \bm \chi)&=\frac{1}{|I_\kappa|}\sum_{j \in I_\kappa} \frac{n_{00}}{n} \left\{ \log\left(1- \chi_{at}(Y_{(j)})\right)+ J(F_{00}(Y_{(j)}),  \frac{ (1-\chi_{nt}(Y_{(j)})-\chi_{at}(Y_{(j)}))\theta_{co}(Y_{(j)}) + \chi_{nt}(Y_{(j)}) \theta_{nt}(Y_{(j)}))}{1-\chi_{at}(Y_{(j)})} \right\}, \\
T_{10}(\bm\theta, \bm \chi)&= \frac{1}{|I_\kappa|}\sum_{j \in I_\kappa} \frac{n_{10}}{n} \left\{\log \chi_{nt}(Y_{(j)}) + J(F_{10}(Y_{(j)}), \theta_{nt}(Y_{(j)})) \right\}, \\
T_{01}(\bm\theta, \bm \chi)&= \frac{1}{|I_\kappa|}\sum_{j \in I_\kappa} \frac{n_{01}}{n} \left\{\log \chi_{at}(Y_{(j)}) + J(F_{01}(Y_{(j)}), \theta_{at}(Y_{(j)})) \right\}, \\
T_{11}(\bm\theta, \bm \chi)&= \frac{1}{|I_\kappa|} \sum_{j \in I_\kappa} \frac{n_{11}}{n} \left\{ \log \left(1-\chi_{nt}(Y_{(j)})\right) + J(F_{11}(Y_{(j)}), \frac{ (1-\chi_{nt}(Y_{(j)})-\chi_{at}(Y_{(j)}))\theta_{co}(Y_{(j)}) + \chi_{at}(Y_{(j)}) \theta_{at}(Y_{(j)}))}{1-\chi_{nt}(Y_{(j)})} \right\}.  
\end{align*}
}
\normalsize
Since there is no closed form solution for $(\hat{\bm \psi}, \hat{\bm \xi})$,  we instead find the asymptotically equivalent estimators $(\breve{\bm \tau}, \breve{\bm \rho})$. These equivalent estimators will be used for examining the large-sample performance of the binomial likelihood ratio test $T_n$. First, we need to define several new quantities to describe $(\breve{\bm \tau}, \breve{\bm \rho})$. Table~\ref{tab:quantity_summary} summarizes these quantities. 

\begin{table}[h]
	\centering
	\caption{Definition of new quantities}\label{tab:quantity_summary}
	\begin{tabular}{ll}
		\toprule
		Quantity & Definition \\
		$E_{1}(t)$ & $(F_{co}(t) - F_{nt}(t)) - \frac{\phi_{at}}{1-\phi_{nt}}(F_{co}(t) - F_{at}(t))$\\ 
		$E_2(t)$ & $\frac{\phi_{nt}}{1-\phi_{at}}(F_{co}(t) - F_{nt}(t)) - (F_{co}(t) - F_{nt}(t))$\\
		$Q_{zd}(t)$ & $\frac{1}{F_{zd}(t)(1-F_{zd}(t))}$ for $z, d \in \{0, 1\}$ \\ 
		$Q(t)$ & $\frac{\phi_1 (1-\phi_{at})}{Q_{00}(t)} + \frac{\phi_1 \phi_{at}}{Q_{01}(t)} + \frac{(1-\phi_1)\phi_{nt}}{Q_{10}(t)} + \frac{(1-\phi_1)(1-\phi_{nt})}{Q_{11}(t)}$ \\
		$C_{00}(t)$ & $\left( \frac{\phi_1 (1-\phi_{at})}{Q_{00}(t)} \right)/Q(t)$ \\
		$C_{01}(t)$ & $\left( \frac{\phi_1 \phi_{at}}{Q_{01}(t)} \right)/Q(t)$ \\
		$C_{10}(t)$ & $\left( \frac{(1-\phi_1)\phi_{nt}}{Q_{10}(t)} \right)/Q(t)$ \\
		$C_{11}(t)$ & $\left( \frac{(1-\phi_1)(1-\phi_{nt})}{Q_{11}(t)} \right)/Q(t)$ \\
		$r_{at}(t)$ & $\frac{1}{Q(t)} \frac{\phi_1(1-\phi_1)^2}{\phi_{at}(1-\phi_{at})}E_{1}(t)$ \\
		$r_{nt}(t)$ & $\frac{1}{Q(t)} \frac{\phi_1^2(1-\phi_1)}{\phi_{nt}(1-\phi_{nt})}E_{2}(t)$\\		
		$det(t)$ & $\frac{\phi_1(1-\phi_1)}{\phi_{nt}(1-\phi_{nt})\phi_{at}(1-\phi_{at})} + \frac{1}{Q(t)} \frac{\phi_1(1-\phi_1)^2}{\phi_{at}(1-\phi_{at})}E_{1}^2 (t) + \frac{1}{Q(t)} \frac{\phi_1^2(1-\phi_1)}{\phi_{nt}(1-\phi_{nt})}E_{2}^2 (t)$,\\
		& or equivalently $\frac{\phi_1(1-\phi_1)}{\phi_{nt}(1-\phi_{nt})\phi_{at}(1-\phi_{at})} + r_{at}(t)E_{1}(t) + r_{nt}(t)E_{2}(t)$
	\end{tabular}
\end{table}

Finally, we can define the equivalent estimator $\breve{\bm \tau} = (\breve{\tau}_{co}, \breve{\tau}_{nt}, \breve{\tau}_{at})$ based on both true values and observable quantities. Observable quantitiles are the empirical distribution functions $\overline{F}_{zd}(t)$, $\overline{F}_0(t)$, $\overline{F}_1(t)$ and plug-in estimators $(\breve{\phi}_{nt}, \breve{\phi}_{at})$. The estimator $\breve{\bm \tau}$ is defined as follows: 
\begin{align*}
	&\sqrt n (\breve{\tau}_{co}(t) - F_{co}(t))\\
	&=\sqrt n \frac{1}{1-\phi_{nt}-\phi_{at}} \Bigg[ (1-\breve{\phi}_{at})(\overline{F}_{00}(t) - F_{00}(t)) - \breve{\phi}_{nt}(\overline{F}_{10}(t) - F_{10}(t)) \\
	&- (C_{00}(t)+C_{10}(t))\left\{ (\overline{F}_0(t) - \overline{F}_1(t)) - \frac{r_{at}(t)E_1(t) + r_{nt}(t)E_2(t)}{det(t)}\right\}\\
	&+(F_{co}(t)-F_{nt}(t))\left\{ (\breve{\phi}_{nt}-\phi_{nt}) + \frac{\phi_{nt}}{1-\phi_{at}}(\breve{\phi}_{at}-\phi_{at}) - \frac{r_{at}(t) + \frac{\phi_{nt}}{1-\phi_{at}} r_{nt}(t)}{det(t)} \right\} \Bigg] \\
	\text{or} & =\sqrt n \frac{1}{1-\phi_{nt}-\phi_{at}} \Bigg[ (1-\breve{\phi}_{nt})(\overline{F}_{11}(t) - F_{11}(t)) - \breve{\phi}_{at}(\overline{F}_{01}(t) - F_{01}(t)) \\
	&+ (C_{01}(t)+C_{11}(t))\left\{ (\overline{F}_0(t) - \overline{F}_1(t)) - \frac{r_{at}(t)E_1(t) + r_{nt}(t)E_2(t)}{det(t)}\right\}\\
	&+(F_{co}(t)-F_{at}(t))\left\{ \frac{\phi_{at}}{1-\phi_{nt}}(\breve{\phi}_{nt}-\phi_{nt}) + (\breve{\phi}_{at}-\phi_{at}) - \frac{\frac{\phi_{at}}{1-\phi_{nt}} r_{at}(t) +  r_{nt}(t)}{det(t)} \right\} \Bigg]
\end{align*}
and 
\begin{align}\label{eq:null_equi_estimator_theta}
	\sqrt n (\breve{\tau}_{nt}(t) - F_{nt}(t))  &= \sqrt n \frac{1}{\phi_{nt}} \bigg[ \breve{\phi}_{nt} (\overline{F}_{10}(t) - F_{10}(t)) \nonumber \\
	&+ C_{10}(t) \left\{(\overline{F}_0(t) - \overline{F}_1(t)) - \frac{r_{at}(t)E_1(t) + r_{nt}(t)E_2(t)}{det(t)} \right\}\bigg] \nonumber\\
	\sqrt n (\breve{\tau}_{at}(t) - F_{at}(t))  &= \sqrt n \frac{1}{\phi_{at}} \bigg[ \breve{\phi}_{at} (\overline{F}_{01}(t) - F_{01}(t)) \nonumber\\
	&- C_{01}(t) \left\{(\overline{F}_0(t) - \overline{F}_1(t)) - \frac{r_{at}(t)E_1(t) + r_{nt}(t)E_2(t)}{det(t)} \right\}\bigg].
\end{align}
Furthermore, the estimator $\breve{\bm \rho}$ is defined as
\begin{align}\label{eq:null_equi_estimator_chi}
	\sqrt n (\breve{\rho}_{nt}(t) - \phi_{nt}) &= \sqrt n \left( (\breve{\phi}_{nt} - \phi_{nt}) - (\overline{F}_0(t) - \overline{F}_1(t)) \frac{r_{at}(t)}{det(t)} \right) \nonumber\\
	\sqrt n (\breve{\rho}_{at}(t) - \phi_{at}) &= \sqrt n \left( (\breve{\phi}_{at} - \phi_{at}) - (\overline{F}_0(t) - \overline{F}_1(t)) \frac{r_{nt}(t)}{det(t)} \right) 
\end{align}

Then, we can prove the following proposition that shows that both $(\hat{\bm \psi}, \hat{\bm \xi})$ and $(\breve{\bm \tau}, \breve{\bm \rho})$ are asymptotically equivalent: 
\begin{proposition}Fix $0< \kappa < 1$. Under the null hypothesis $H_0$, the maximum binomial likelihood estimators $\hat{\bm \psi}$ and $\hat{\bm \xi}$ satisfy
	\begin{align*}
	\frac{1}{|I_\kappa|}\sum_{j \in  I_\kappa} \left\|
	\sqrt n \left\{ \hat{\bm \psi}(Y_{(j)}) - \breve{\bm \tau}(Y_{(j)}) \right\}
	\right\|_2^2 &=o_P(1),\\
	\frac{1}{|I_\kappa|}\sum_{j \in  I_\kappa} \left\|
	\sqrt n \left\{ \hat{\bm \xi}(Y_{(j)}) - \breve{\bm \rho}(Y_{(j)}) \right\}
	\right\|_2^2 &=o_P(1).
	\end{align*} 
	\label{proposition:nullBL}
\end{proposition}

\subsection{Proof of Proposition \ref{proposition:nullBL}} 

To begin with define $$(\breve{\bm \psi}, \breve{\bm \xi})=\argmax_{\bm \theta \in \bm \vartheta_{0}, \bm \chi \in \bm \varphi} \M_n(\bm \theta, \bm \chi),$$ where $\bm \vartheta_0 = \left\{(\theta_{co}, \theta_{nt}, \theta_{at}):  \theta_{co}, \theta_{nt}, \theta_{at} \in \R^\R \right\}$, is the unrestricted null parameter space. In this case there is no-closed form expression of $(\breve{\bm \psi}, \breve{\bm \xi})$. However, by the asymptotic expansion of the sample null objective function we can find an asymptotically equivalent formula for $\breve{\bm \psi}$ and $\breve{\bm \xi}$.

\begin{lemma}\label{lm:nullplugin} Let $\breve{\bm \psi}$ and $\breve{\bm \tau}=(\breve \tau_{co}, \breve \tau_{nt}, \breve \tau_{at})$ be as defined above. Then 
	\begin{align}\label{eq:nullBL}
	\frac{1}{|I_{\kappa}|}\sum_{j \in I_\kappa} \left\|\sqrt n \left\{ \breve{\bm \psi}(Y_{(j)})- \breve{\bm \tau}(Y_{(j)}) \right\} \right\|^2_2 &= o_P(1), \nonumber \\
	\frac{1}{|I_{\kappa}|}\sum_{j \in I_\kappa} \left\|\sqrt n \left\{ \breve{\bm \xi}(Y_{(j)})- \breve{\bm \rho}(Y_{(j)}) \right\} \right\|^2_2 &= o_P(1)
	\end{align}
	whenever $|| \breve{\bm \psi}(Y_{(j)})- {\bm F_0}(Y_{(j)})||_2=o_P(1)$, for every $j \in I_\kappa$, where $\bm F_0=(F_{co}, F_{nt}, F_{at})$, with $F_{co}^{(0)}=F_{co}^{(1)}=F_{co}$, is the vector of true distribution functions under the null.
\end{lemma}

\begin{proof}
	We have
	\begin{align}
	&( \M_n - \M)(\breve{\bm \psi}, \breve{\bm \xi}) - (\M_n - \M)(\bm F_0, \bm\phi) \nonumber \\
	&= \frac{1}{|I_{\kappa}|} \sum_{z, d\in \{0, 1\}} \sum_{j \in I_\kappa} \frac{n_{zd}}{n} \frac{ (\overline F_{zd}(Y_{(j)}) - F_{zd}(Y_{(j)}) )}{F_{zd}(Y_{(j)})(1-F_{zd}(Y_{(j)}))}(\hat{\psi}_{zd}(Y_{(j)})- F_{zd}(Y_{(j)})) + O_P(n^{-\frac{3}{2}}) \nonumber\\
	&= \frac{1}{|I_{\kappa}|} \sum_{z, d\in \{0, 1\}} \sum_{j \in I_\kappa} Q_{zd}(Y_{(j)})(\hat{\psi}_{zd}(Y_{(j)})- F_{zd}(Y_{(j)})) + O_P(n^{-\frac{3}{2}}), \nonumber
	\end{align}
	where $Q_{zd}(Y_{(j)})=\frac{n_{zd}}{n} \cdot \frac{1}{F_{zd}(Y_{(j)})(1-F_{zd}(Y_{(j)}))}$. 
	
	Now, under the null hypothesis, $F_{co}^{(0)}=F_{co}^{(1)}=F_{co}$, we can re-group the terms in the above sum in terms of $\hat{\psi}_{co}(Y_{(j)}) - F_{co}(Y_{(j)}), \hat{\psi}_{nt}(Y_{(j)}) - F_{nt}(Y_{(j)})$ and $\hat{\psi}_{at}(Y_{(j)}) - F_{at}(Y_{(j)})$, to get 
	\begin{align}\label{eq:Mdiffnull1}
	&\sqrt n \big\{ (\M_n - \M)(\breve{\bm \psi}, \breve{\bm \xi}) - (\M_n - \M)(\bm{F}_0, \bm\phi) \big\} \nonumber \\
	&=\frac{\sqrt n}{|I_{\kappa}|} \left\{ \sum_{j \in I_\kappa} \begin{pmatrix}
	\breve{\xi}_{nt}(Y_{(j)}) - \phi_{nt} \\
	\breve{\xi}_{at}(Y_{(j)}) - \phi_{at} \\
	\breve{\psi}_{co}(Y_{(j)}) - F_{co}(Y_{(j)}) \\
	\breve{\psi}_{nt}(Y_{(j)}) - F_{nt}(Y_{(j)}) \\
	\breve{\psi}_{at}(Y_{(j)}) - F_{at}(Y_{(j)}) \\
	\end{pmatrix}^{\T}  \bm Z_{n}(Y_{(j)}) \right\}+ O_P(n^{-\frac{3}{2}}), 
	\end{align}
	where $\bm{Z}_n(t) = \begin{pmatrix}
	\bm{Z}_{n1}(t) \\ \bm{Z}_{n2}(t)
	\end{pmatrix}$ is a $5 \times 1$ matrix with 
	\begin{align}\label{eq:Znt}
		\bm Z_{n1}(t)  &= \frac{\sqrt n}{|I_{\kappa}|} \begin{pmatrix}
	\frac{\phi_1(\breve{\phi}_{nt} - \phi_{nt})}{\phi_{nt}(1-\phi_{nt})} - \frac{1}{1-\phi_{at}}\frac{n_{00}}{n} (F_{co}(t) - F_{nt}(t))(\overline{F}_{00}(t) - F_{00}(t)) Q_{00}(t) \\
	~~~~~~~~~~~~ - \frac{\phi_{at}}{(1-\phi_{nt})^2} \frac{n_{11}}{n} (F_{co}(t)-F_{at}(t))(\overline{F}_{11}(t)-F_{11}(t))Q_{11}(t) \\
	\frac{(1-\phi_1)(\breve{\phi}_{at} - \phi_{at})}{\phi_{at}(1-\phi_{at})} - \frac{\phi_{nt}}{(1-\phi_{at})^2}\frac{n_{00}}{n} (F_{co}(t) - F_{nt}(t))(\overline{F}_{00}(t) - F_{00}(t)) Q_{00}(t) \\
	~~~~~~~~~~~~ - \frac{1}{1-\phi_{nt}} \frac{n_{11}}{n} (F_{co}(t)-F_{at}(t))(\overline{F}_{11}(t)-F_{11}(t))Q_{11}(t)
	\end{pmatrix}, \\
	\bm Z_{n2}(t)  &= \frac{\sqrt n}{|I_{\kappa}|} \begin{pmatrix}
	\frac{1-\phi_{nt}-\phi_{at}}{1-\phi_{at}} \frac{n_{00}}{n} (\overline{F}_{00}(t) - F_{00}(t))Q_{00}(t) + \frac{1-\phi_{nt}-\phi_{at}}{1-\phi_{nt}} \frac{n_{11}}{n} (\overline{F}_{11}(t) - F_{11}(t))Q_{11}(t)\\
	\frac{\phi_{nt}}{1-\phi_{at}} \frac{n_{00}}{n} (\overline{F}_{00}(t) - F_{00}(t))Q_{00}(t) + \frac{n_{10}}{n} (\overline{F}_{10}(t) - F_{10}(t))Q_{10}(t) \\
	\frac{\phi_{at}}{1-\phi_{nt}} \frac{n_{11}}{n} (\overline{F}_{11}(t) - F_{11}(t))Q_{11}(t) + \frac{n_{01}}{n} (\overline{F}_{01}(t) - F_{01}(t))Q_{01}(t)
	\end{pmatrix}.
	\end{align}

	Next, denote by $\bm V_n$ the Hessian matrix of $\M(\bm \theta, \bm \chi)$ at the point $\left(\bm F_0(Y_{(j)}), \bm \phi(Y_{(j)}) \right)$ for $j \in I_{\kappa}$. Note that the Hessian matrix is block diagonal 
	\begin{align}\label{eq:V}
	\bm V_n=\diag(\bm V_n(Y_{(j)}))_{j \in I_\kappa},
	\end{align}
	where $\bm V_n(\cdot): \R \rightarrow \R^{5\times 5}$ is given by the following:
	$$
	\bm{V}_n(t) = \begin{bmatrix}
		\frac{\partial^2 \mathbb{M}(\bm\theta, \bm\chi)}{(\partial \chi_{nt}(t))^2} & \frac{\partial^2 \mathbb{M}(\bm\theta, \bm\chi)}{\partial \chi_{nt} \partial \chi_{at}(t)} & \frac{\partial^2 \mathbb{M}(\bm\theta, \bm\chi)}{\partial \chi_{nt}(t) \partial \theta_{co}(t)} & \frac{\partial^2 \mathbb{M}(\bm\theta, \bm\chi)}{\partial \chi_{nt}(t) \partial \theta_{nt}(t)} & \frac{\partial^2 \mathbb{M}(\bm\theta, \bm\chi)}{\partial \chi_{nt}(t) \partial \theta_{at}(t)} \\
		
		\frac{\partial^2 \mathbb{M}(\bm\theta, \bm\chi)}{\partial \chi_{at}(t) \partial \chi_{nt}(t)}   & \frac{\partial^2 \mathbb{M}(\bm\theta, \bm\chi)}{(\partial \chi_{at}(t))^2} & \frac{\partial^2 \mathbb{M}(\bm\theta, \bm\chi)}{\partial \chi_{at}(t) \partial \theta_{co}(t)} & \frac{\partial^2 \mathbb{M}(\bm\theta, \bm\chi)}{\partial \chi_{at}(t) \partial \theta_{nt}(t)} & \frac{\partial^2 \mathbb{M}(\bm\theta, \bm\chi)}{\partial \chi_{at}(t) \partial \theta_{at}(t)} \\
		
		\frac{\partial^2 \mathbb{M}(\bm\theta, \bm\chi)}{\partial \theta_{co}(t) \partial \chi_{nt}(t)}   & \frac{\partial^2 \mathbb{M}(\bm\theta, \bm\chi)}{\partial \theta_{co}(t) \partial \chi_{nt}(t)} & \frac{\partial^2 \mathbb{M}(\bm\theta, \bm\chi)}{(\partial \theta_{co}(t))^2 } & \frac{\partial^2 \mathbb{M}(\bm\theta, \bm\chi)}{\partial \theta_{co}(t) \partial \theta_{nt}(t)} & \frac{\partial^2 \mathbb{M}(\bm\theta, \bm\chi)}{\partial \theta_{co}(t) \partial \theta_{at}(t)} \\
		
		\frac{\partial^2 \mathbb{M}(\bm\theta, \bm\chi)}{\partial \theta_{nt}(t) \partial \chi_{nt}(t)}   & \frac{\partial^2 \mathbb{M}(\bm\theta, \bm\chi)}{\partial \theta_{nt}(t) \partial \chi_{nt}(t)} & \frac{\partial^2 \mathbb{M}(\bm\theta, \bm\chi)}{\partial \theta_{nt}(t) \partial \theta_{co}(t)} & \frac{\partial^2 \mathbb{M}(\bm\theta, \bm\chi)}{(\partial \theta_{nt}(t))^2 } & \frac{\partial^2 \mathbb{M}(\bm\theta, \bm\chi)}{\partial \theta_{nt}(t) \partial \theta_{at}(t)} \\
		
		\frac{\partial^2 \mathbb{M}(\bm\theta, \bm\chi)}{\partial \theta_{at}(t) \partial \chi_{nt}(t)}   & \frac{\partial^2 \mathbb{M}(\bm\theta, \bm\chi)}{\partial \theta_{at}(t) \partial \chi_{nt}(t)} & \frac{\partial^2 \mathbb{M}(\bm\theta, \bm\chi)}{\partial \theta_{at}(t) \partial \theta_{co}(t)} & \frac{\partial^2 \mathbb{M}(\bm\theta, \bm\chi)}{\partial \theta_{at}(t) \partial \theta_{nt}(t)} & \frac{\partial^2 \mathbb{M}(\bm\theta, \bm\chi)}{(\partial \theta_{at}(t))^2 }  \\
	\end{bmatrix}_{(\bm \theta, \bm \chi) = (\bm{F}_0, \bm \phi)}
	$$
	
	For each $\bm{V}_n(t)$, we simply make partitions such as $-\bm{V}_n(t) = \begin{pmatrix}
	A(t) & B(t) \\ B(t)^{\T} & C(t)
	\end{pmatrix}$
	where $A(t)$ is the upper-left $2 \times 2$ matrix, $B(t)$ is the upper right $2 \times 3$ matrix, and $C(t)$ is the lower right $3\times 3$ matrix. Both $A(t)$ and $C(t)$ are symmetric. The three matrices have the forms, 
	\begin{align*}
	A(t) &= \frac{\sqrt n}{|I_{\kappa}|}\begin{pmatrix}
	\frac{\phi_1}{\phi_{nt}(1-\phi_{nt})} + \frac{1-\phi_1}{1-\phi_{at}} \frac{(F_{co}(t) - F_{nt}(t))^2}{F_{00}(t)(1-F_{00}(t))}    &   \frac{(1-\phi_1)\phi_{nt}}{(1-\phi_{at})^2} \frac{(F_{co}(t) - F_{nt}(t))^2}{F_{00}(t)(1-F_{00}(t))} \\
	~~~~~+ \frac{\phi_1 \phi_{at}^2}{(1-\phi_{nt})^3} \frac{(F_{co}(t) - F_{at}(t))^2}{F_{11}(t)(1-F_{11}(t))} & ~~~~~+ \frac{\phi_1 \phi_{at}}{(1-\phi_{nt})^2} \frac{(F_{co}(t) - F_{at}(t))^2}{F_{11}(t)(1-F_{11}(t))} \\
	&\frac{1-\phi_1}{\phi_{at}(1-\phi_{at})} + \frac{(1-\phi_1)\phi_{nt}^2}{(1-\phi_{at})^3} \frac{(F_{co}(t) - F_{nt}(t))^2}{F_{00}(t)(1-F_{00}(t))} \\
	& ~~~~~+ \frac{\phi_1 }{1-\phi_{nt}} \frac{(F_{co}(t) - F_{at}(t))^2}{F_{11}(t)(1-F_{11}(t))} 
	\end{pmatrix}, \\
	B(t) &= \frac{\sqrt n}{|I_{\kappa}|} \begin{bmatrix}
	\frac{(1-\phi_1)(1-\phi_{nt}-\phi_{at})}{1-\phi_{at}}\cdot \frac{-(F_{co}(t) - F_{nt}(t))}{F_{00}(t)(1-F_{00}(t))}  & \frac{(1-\phi_1)\phi_{nt}}{1-\phi_{at}}   & \frac{\phi_1 \phi_{at}^2}{(1-\phi_{nt})^2} \\
	~~~~~+ \frac{\phi_1(1-\phi_{nt}-\phi_{at})\phi_{nt}}{(1-\phi_{nt})^2} \cdot \frac{-(F_{co}(t) - F_{at}(t))}{F_{11}(t)(1-F_{11}(t))} & ~~~~~ \times \frac{-(F_{co}(t) - F_{nt}(t))}{F_{00}(t)(1-F_{00}(t))}& ~~~~~ \times  \frac{-(F_{co}(t) - F_{at}(t))}{F_{11}(t)(1-F_{11}(t))} \\
	\frac{(1-\phi_1)(1-\phi_{nt}-\phi_{at})\phi_{nt}}{(1-\phi_{at})^2}\cdot \frac{-(F_{co}(t) - F_{nt}(t))}{F_{00}(t)(1-F_{00}(t))} & \frac{(1-\phi_1)\phi_{nt}^2}{(1-\phi_{at})^2}  & \frac{\phi_1 \phi_{at}}{(1-\phi_{nt})}  \\\
	~~~~~ + \frac{\phi_1(1-\phi_{nt}-\phi_{at})}{(1-\phi_{nt})} \cdot \frac{-(F_{co}(t) - F_{at}(t))}{F_{11}(t)(1-F_{11}(t))} & ~~~~~\times \frac{-(F_{co}(t) - F_{nt}(t))}{F_{00}(t)(1-F_{00}(t))} &  ~~~~~\times \frac{-(F_{co}(t) - F_{at}(t))}{F_{11}(t)(1-F_{11}(t))}
	\end{bmatrix}, \\
	C(t) &= \frac{\sqrt n}{|I_{\kappa}|} \begin{bmatrix}
	\frac{(1-\phi_1)(1-\phi_{nt}-\phi_{at})^2}{1-\phi_{at}}Q_{00}(t)& \frac{(1-\phi_1)(1-\phi_{nt}-\phi_{at})\phi_{nt}}{1-\phi_{at}}Q_{00}(t) & \frac{\phi_1(1-\phi_{nt}-\phi_{at})\phi_{at}}{1-\phi_{nt}}Q_{11}(t) \\
	~~~~~+ \frac{\phi_1(1-\phi_{nt}-\phi_{at})^2}{(1-\phi_{nt})} Q_{11}(t)  & & \\
	& \frac{(1-\phi_1)\phi_{nt}^2}{1-\phi_{at}}Q_{00}(t)  & 0 \\
	& ~~~~~+ \phi_1 \phi_{nt} Q_{10}(t) & \\
	& &  \frac{\phi_1\phi_{at}^2}{1-\phi_{nt}}Q_{11}(t) \\
	& & ~~~~~ + (1-\phi_1)\phi_{at} Q_{01}(t)
	\end{bmatrix}
	\end{align*}

	Now, by a second order Taylor expansion of $\M(\breve{\bm \psi}, \breve{\bm\xi}) - \M(\bm F_0,  \bm\phi)$ around the true values $(\bm{F}_0, \bm\phi)$ gives, 
	\begin{align*}
	&\M(\breve{\bm \psi}, \breve{\bm\xi})- \M(\bm F_0,\bm\phi)\\
	&=\frac{1}{2}\cdot \sum_{j \in I_\kappa}\begin{pmatrix}
	\breve{\xi}_{nt}(Y_{(j)}) - \phi_{nt} \\
	\breve{\xi}_{at}(Y_{(j)}) - \phi_{at} \\
	\breve{\psi}_{co}(Y_{(j)}) - F_{co}(Y_{(j)}) \\
	\breve{\psi}_{nt}(Y_{(j)}) - F_{nt}(Y_{(j)}) \\
	\breve{\psi}_{at}(Y_{(j)}) - F_{at}(Y_{(j)}) \\
	\end{pmatrix}^{\T} \bm V(Y_{(j)}) \begin{pmatrix}
	\breve{\xi}_{nt}(Y_{(j)}) - \phi_{nt} \\
	\breve{\xi}_{at}(Y_{(j)}) - \phi_{at} \\
	\breve{\psi}_{co}(Y_{(j)}) - F_{co}(Y_{(j)}) \\
	\breve{\psi}_{nt}(Y_{(j)}) - F_{nt}(Y_{(j)}) \\
	\breve{\psi}_{at}(Y_{(j)}) - F_{at}(Y_{(j)}) \\
	\end{pmatrix} + O_P(n^{-\frac{3}{2}}),
	\end{align*} 
	since the gradient of $\M(\bm \theta, \bm \chi)$ at the point $(\bm{F}_0, \bm\phi)$ is zero (by arguments similar to the proof of Lemma \ref{lemma:Mnopt}). Then from \eqref{eq:Mdiffnull1} 
	\begin{align}\label{Mdiffnull2}
	&\M_n(\breve{\bm \psi}, \breve{\bm \xi}) - \M_n (\bm F_0, {\bm \phi}) =\frac{\sqrt n }{|I_{\kappa}|} \left\{ \sum_{j \in I_\kappa}\begin{pmatrix}
	\breve{\xi}_{nt}(Y_{(j)}) - \phi_{nt} \\
	\breve{\xi}_{at}(Y_{(j)}) - \phi_{at} \\
	\breve{\psi}_{co}(Y_{(j)}) - F_{co}(Y_{(j)}) \\
	\breve{\psi}_{nt}(Y_{(j)}) - F_{nt}(Y_{(j)}) \\
	\breve{\psi}_{at}(Y_{(j)}) - F_{at}(Y_{(j)}) \\
	\end{pmatrix}^{\T} \bm Z_{n}(Y_{(j)}) \right\} \nonumber \\
	& + \frac{1}{2} \frac{n}{|I_{\kappa}|}\cdot \sum_{j \in I_\kappa}\begin{pmatrix}
	\breve{\xi}_{nt}(Y_{(j)}) - \phi_{nt} \\
	\breve{\xi}_{at}(Y_{(j)}) - \phi_{at} \\
	\breve{\psi}_{co}(Y_{(j)}) - F_{co}(Y_{(j)}) \\
	\breve{\psi}_{nt}(Y_{(j)}) - F_{nt}(Y_{(j)}) \\
	\breve{\psi}_{at}(Y_{(j)}) - F_{at}(Y_{(j)}) \\
	\end{pmatrix}^{\T} \bm V_n(Y_{(j)}) \begin{pmatrix}
	\breve{\xi}_{nt}(Y_{(j)}) - \phi_{nt} \\
	\breve{\xi}_{at}(Y_{(j)}) - \phi_{at} \\
	\breve{\psi}_{co}(Y_{(j)}) - F_{co}(Y_{(j)}) \\
	\breve{\psi}_{nt}(Y_{(j)}) - F_{nt}(Y_{(j)}) \\
	\breve{\psi}_{at}(Y_{(j)}) - F_{at}(Y_{(j)}) \\
	\end{pmatrix} + O_P(n^{-\frac{3}{2}}), 
	\end{align}
	
	Similarly, replacing $(\breve{\bm \xi}, \breve{\bm \psi})$ by $(\bm \phi, \bm F_0)-n^{-\frac{1}{2}}\bm V_n^{-1} \bm Z_n=(\breve{\bm \rho}, \breve{\bm \tau})$ (by Lemma \ref{lm:nullexp} below),  in \eqref{Mdiffnull2}, where $\bm Z_n=(\bm Z_n(Y_{(j)}))_{j \in I_\kappa}$ and $\bm V_n=\diag(\bm V_n(Y_{(j)}))_{j \in I_\kappa}$, gives 
	\begin{align} \label{eq:diffnormalexpansion}
	\M_n(\breve{\bm \tau}, \breve{\bm \rho}) - \M_n(\bm{F}_0, \bm \phi) &= -\frac{1}{2}\cdot \frac{1}{|I_{\kappa}|}\sum_{j \in I_{\kappa}}\bm  Z_n(Y_{(j)})^{\T} \bm V_n(Y_{(j)})^{-1}\bm  Z_n(Y_{(j)})+ O_P(n^{-\frac{3}{2}}),
	\end{align}
	since $\frac{1}{|I_{\kappa}|} \sum_{j \in I_\kappa} ||\breve{\bm \tau}(Y_{(j)})||^2_2=O_P(1/ n)$. This implies, subtracting \eqref{eq:diffnormalexpansion} from \eqref{Mdiffnull2} gives, 
	\begin{align} \label{eq:BLqf}
	&\M_n(\breve{\bm \psi}, \breve{\bm \xi}) - \M_n(\breve{\bm \tau}, \breve{\bm \rho}) \nonumber\\
	&=\frac{1}{2}\frac{n}{|I_{\kappa}|}\sum_{j \in I_\kappa} \begin{pmatrix}
	\breve{\xi}_{nt}(Y_{(j)}) - \breve{\rho}_{nt}(Y_{(j)}) \\
	\breve{\xi}_{at}(Y_{(j)}) - \breve{\rho}_{at}(Y_{(j)}) \\
	\breve{\psi}_{co}(Y_{(j)}) - \breve{\tau}_{co}(Y_{(j)}) \\
	\breve{\psi}_{nt}(Y_{(j)}) - \breve{\tau}_{nt}(Y_{(j)}) \\
	\breve{\psi}_{at}(Y_{(j)}) - \breve{\tau}_{at}(Y_{(j)}) \\
	\end{pmatrix}^{\T}\bm V_n(Y_{(j)}) \begin{pmatrix}
	\breve{\xi}_{nt}(Y_{(j)}) - \breve{\rho}_{nt}(Y_{(j)}) \\
	\breve{\xi}_{at}(Y_{(j)}) - \breve{\rho}_{at}(Y_{(j)}) \\
	\breve{\psi}_{co}(Y_{(j)}) - \breve{\tau}_{co}(Y_{(j)}) \\
	\breve{\psi}_{nt}(Y_{(j)}) - \breve{\tau}_{nt}(Y_{(j)}) \\
	\breve{\psi}_{at}(Y_{(j)}) - \breve{\tau}_{at}(Y_{(j)}) \\
	\end{pmatrix}+ O_P(n^{-\frac{3}{2}}). 
	\end{align}
	
	Now, since $\M_n(\breve{\bm \psi}, \breve{\bm \xi}) - \M_n(\breve{\bm \tau}, \breve{\bm \rho}) \geq 0$ and $\sup_{j \in I_\kappa}||-\bm V(Y_{(j)})^{-1}||_\infty=O_P(n)$ (seen from Lemma \ref{lm:nullexp} below), the result in \eqref{eq:nullBL} follows.\footnote{For a symmetric matrix $\bm A$, denote by $||\bm A||_\infty$ the maximum eigenvalue of $\bm A$.} 
\end{proof}

\begin{lemma}\label{lm:nullexp}For $t \in (0, 1)$, 
	\begin{align*}
	-\bm V_n^{-1}(t)\bm Z_n(t)=\sqrt n 
	\begin{pmatrix}
	\breve{\rho}_{nt}(t) - \phi_{nt} \\
	\breve{\rho}_{at}(t) - \phi_{at} \\ 
	\breve \tau_{co}(t) - F_{co}(t) \\
	\breve \tau_{nt}(t) - F_{nt}(t)\\
	\breve \tau_{at}(t) - F_{at}(t)
	\end{pmatrix}, 
	\end{align*} 
	where $\breve{\bm \tau} = (\breve{\tau}_{co}, \breve{\tau}_{nt}, \breve{\tau}_{at})$ is defined in \eqref{eq:null_equi_estimator_theta} and $\breve{\bm \rho} = (\breve{\rho}_{nt}, \breve{\rho}_{at})$ is defined in \eqref{eq:null_equi_estimator_chi}. 
\end{lemma}

\begin{proof} Recall $\bm V_n(t)$ from \eqref{eq:V} and $-\bm{V}_n$ that has four sub-matrices $A(t), B(t), B(t)^{\T}$ and $C(t)$. Then, the inverse matrix $-\bm V^{-1}_n (t)$ can be computed as:
$$
-\bm V^{-1}_n (t) = \begin{pmatrix}
A^{*}(t) & -A^{*}(t) B(t)C(t)^{-1} \\
-C(t)^{-1} B(t)^{\T} A^{*}(t) & C(t)^{-1} + C(t)^{-1}B(t)^{\T}A^{*}(t)B(t)C(t)^{-1}
\end{pmatrix}
$$
where $A^{*}(t) = (A(t) - B(t)C(t)^{-1}B(t)^{\T})^{-1}$. Also, the multiplication $-\bm{V}_n(t) \bm{Z}_n(t)$ can be represented by, 
	$$
	-\bm V^{-1}_n (t) \bm{Z}_n(t) =\begin{pmatrix}
	A^*(t) (Z_{n1}(t) - B(t)C(t)^{-1}Z_{n2}(t)) \\
	C(t)^{-1}Z_{n2}(t) - C(t)^{-1}B(t)^{\T}A^*(t) (Z_{n1}(t) - B(t)C(t)^{-1}Z_{n2}(t))
	\end{pmatrix}.
	$$
	Then, the proof can be completed by direct multiplication. 
\end{proof}

The proof Proposition \ref{proposition:nullBL} can now be completed by arguments similar to the proof of Proposition \ref{proposition:pavabl}. We outline the steps below, omitting the details: 

\begin{itemize}
	
	\item[--]To begin with define, $\tilde {\bm \tau}=(\tilde \tau_{co}, \tilde \tau_{nt}, \tilde \tau_{at})$, as follows: 
	\begin{align}\label{eq:inctau}
	\tilde \tau_{s}=\arg\min_{\theta \in \cP([0, 1]^\R)}\sum_{j \in I_\kappa} (\breve \tau_{s}(Y_{(j)})-\theta(Y_{(j)}))^2,
	\end{align}
	where $s \in \{co, nt, at\}$. Then as in Proposition \ref{proposition:pavabl}, it can shown that $$
	\frac{1}{|I_{\kappa}|}\sum_{j \in I_\kappa} \left\|\sqrt n \left\{\tilde{\bm \tau}(Y_{(j)})- \breve{\bm \tau}(Y_{(j)})\right\} \right\|^2_2=o_P(n^{-\frac{1}{4}}).$$
	Also, we can define $\tilde{\bm \rho} = (\tilde{\rho}_{nt}, \tilde{\rho}_{at})$ by using truncation $(\breve{\rho}_{nt}, \breve{\rho}_{at})$ with the interval $[0,1]$. Specifically, $$
	\tilde{\rho}_{nt}(t) = \begin{cases}
	0 & \text{if } \breve{\rho}_{nt}(t) < 0 \\
	\breve{\rho}_{nt}(t) & \text{if } 0 \geq \breve{\rho}_{nt}(t) \leq 1 \\
	1 & \text{if } \breve{\rho}_{nt}(t) > 1
	\end{cases}, \quad \tilde{\rho}_{at}(t) = \begin{cases}
	0 & \text{if } \breve{\rho}_{at}(t) < 0 \\
	\breve{\rho}_{at}(t) & \text{if } 0 \geq \breve{\rho}_{at}(t) \leq 1 \\
	1 & \text{if } \breve{\rho}_{at}(t) > 1
	\end{cases}.
	$$
	Then, we have   
	$$
	\frac{1}{|I_{\kappa}|}\sum_{j \in I_\kappa} \left\|\sqrt n \left\{\tilde{\bm \rho}(Y_{(j)})- \breve{\bm \rho}(Y_{(j)})\right\} \right\|^2_2=o_P(n^{-\frac{1}{4}}).
	$$
	This implies $\M_n(\breve{\bm \tau}, \breve{\bm \rho}) - \M_n(\hat{\bm \psi}, \hat{\bm \xi}) \leq \M_n(\breve{\bm \tau}, \breve{\bm \rho}) - \M_n(\tilde{\bm \tau}, \tilde{\bm \rho})=o_P(n^{-\frac{5}{4}})$.

	\item[--] Then as in the proof of Lemma \ref{lm:BLpluginsqrtn} in Section \ref{sec:pfBLplugin}, we can have $$
	\frac{1}{|I_{\kappa}|}\sum_{j \in I_\kappa} \left\| \sqrt n \left\{\hat{\bm \psi}(Y_{(j)})-  \breve{\bm \tau}(Y_{(j)})\right\}\right\|_2^2 =o_P(n^{-\frac{1}{4}}).
	$$
	Similarly, we have $$
	\frac{1}{|I_{\kappa}|}\sum_{j \in I_\kappa} \left\| \sqrt n \left\{\hat{\bm \xi}(Y_{(j)})-  \breve{\bm \rho}(Y_{(j)})\right\}\right\|_2^2 =o_P(n^{-\frac{1}{4}}),$$
	which completes the proof of Proposition \ref{proposition:nullBL}.  
	
\end{itemize}

\section{Proof of Theorems 2 and 3}
\label{sec:pfblrt}

\subsection{Proof of Theorem 2}
\label{subsec:thm2}

In this section we derive the asymptotic distribution of the binomial likelihood ratio statistic. The binomial likelihood ratio test statistic is defined as
$$T_n=2 \left( \max_{\bm \theta \in \bm \vartheta_{+}, \bm \chi \in \bm \varphi_+} \ell_{\bm{Y}, \bm{D} | \bm{Z}}(\bm \theta, \bm \chi)-\max_{\bm \theta \in \bm \vartheta_{+, 0}, \bm \chi \in \bm \varphi_+}\ell_{\bm{Y}, \bm{D} | \bm{Z}}(\bm \theta, \bm \chi) \right),$$
Denote 
\begin{align}\label{eq:BLRT}
(\hat {\bm \theta}, \hat{\bm \chi})=\argmax_{\bm \theta \in \bm \vartheta_{+}, \bm \chi \in \bm \varphi_+} \ell_{\bm{Y}, \bm{D} | \bm{Z}}(\bm \theta, \bm \chi), \quad (\hat {\bm \psi}, \hat{\bm \xi})=\argmax_{\bm \theta \in \bm \vartheta_{+, 0}, \bm \chi \in \bm \varphi_{+}} \ell_{\bm{Y}, \bm{D} | \bm{Z}}(\bm \theta, \bm \chi),
\end{align} 
where $\hat {\bm \theta}(t)=(\hat \theta_{co}^{(0)}(t), \hat \theta_{nt}(t), \hat \theta_{co}^{(1)}(t), \hat \theta_{at}(t))$ and $\hat {\bm \psi}(t)=(\hat \psi_{co}(t), \hat \psi_{nt}(t), \hat \psi_{at}(t))$. Furthermore, we denote
\begin{align*}
\hat{\theta}_{00}(t) &= \frac{(1-\hat{\chi}_{nt}(t)-\hat{\chi}_{at}(t))\hat{\theta}_{co}(t) + \hat{\chi}_{nt}(t)\hat{\theta}_{nt}(t)}{1-\hat{\chi}_{at}(t)}, \\
\hat{\theta}_{01}(t) &= \hat{\theta}_{at}(t), \quad \hat{\theta}_{10}(t) = \hat{\theta}_{nt}(t), \\
\hat{\theta}_{11}(t) &= \frac{(1-\hat{\chi}_{nt}(t)-\hat{\chi}_{at}(t))\hat{\theta}_{co}(t) + \hat{\chi}_{at}(t)\hat{\theta}_{at}(t)}{1-\hat{\chi}_{nt}(t)}, \\
\hat{\psi}_{00}(t) &= \frac{(1-\hat{\xi}_{nt}(t)-\hat{\xi}_{at}(t))\hat{\psi}_{co}(t) + \hat{\xi}_{nt}(t)\hat{\psi}_{nt}(t)}{1-\hat{\xi}_{at}(t)}, \\
\hat{\psi}_{01}(t) &= \hat{\psi}_{at}(t), \quad \hat{\psi}_{10}(t) = \hat{\psi}_{nt}(t), \\
\hat{\psi}_{11}(t) &= \frac{(1-\hat{\xi}_{nt}(t)-\hat{\xi}_{at}(t))\hat{\psi}_{co}(t) + \hat{\xi}_{at}(t)\hat{\psi}_{at}(t)}{1-\hat{\xi}_{nt}(t)}.
\end{align*}

We also define $\breve{\tau}_{zd}(t)$ for $z, d \in \{0,1\}$ based on $\breve{\bm \tau} = (\breve{\tau}_{co}, \breve{\tau}_{nt}, \breve{\tau}_{at})$ and $\breve{\bm \rho} = (\breve{\rho}_{nt}, \breve{\rho}_{at})$ as follows:
\begin{align}\label{eq:equi_estimator_tau_zd}
	\breve{\tau}_{00}(t) &= \overline{F}_{00}(t) - \frac{1}{1-\phi_{at}}C_{00}(t)\left\{(\overline{F}_0(t) - \overline{F}_1(t)) - \frac{r_{at}(t)E_1(t) + r_{nt}(t)E_2(t)}{det}\right\} \nonumber\\
	\breve{\tau}_{01}(t) &= \overline{F}_{01}(t) - \frac{1}{\phi_{at}}C_{01}(t)\left\{(\overline{F}_0(t) - \overline{F}_1(t)) - \frac{r_{at}(t)E_1(t) + r_{nt}(t)E_2(t)}{det}\right\} \nonumber\\
	\breve{\tau}_{10}(t) &= \overline{F}_{10}(t) + \frac{1}{\phi_{nt}}C_{10}(t)\left\{(\overline{F}_0(t) - \overline{F}_1(t)) - \frac{r_{at}(t)E_1(t) + r_{nt}(t)E_2(t)}{det}\right\} \nonumber\\
	\breve{\tau}_{11}(t) &= \overline{F}_{11}(t) + \frac{1}{1-\phi_{at}}C_{11}(t)\left\{(\overline{F}_0(t) - \overline{F}_1(t)) - \frac{r_{at}(t)E_1(t) + r_{nt}(t)E_2(t)}{det}\right\}
\end{align} 

We now have the following lemma, which shows that, under the null, the restricted (i.e., $\hat{\psi}_{zd}$) and unrestricted (i.e., $\hat{\theta}_{zd}$) maximum binomial likelihood estimators are asymptotically close to the estimators $\breve{\tau}_{zd}$ and $\overline{F}_{zd}$ respectively. 

\begin{lemma}\label{lm:blplugin} Under the null $H_0$, the following holds:
	\begin{align}\label{eq:blpluginnI}
	\frac{1}{|I_{\kappa}|}\sum_{j \in I_\kappa} \sum_{z, d \in \{0, 1\}} \left\{ \sqrt n \left(\hat{\theta}_{zd}(Y_{(j)})-\overline F_{zd}(Y_{(j)}) \right)\right\}^2= o_P(1), 
	\end{align}
	and 
	\begin{align}\label{eq:blpluginnII} 
	\frac{1}{|I_{\kappa}|}\sum_{j \in I_\kappa} \sum_{z, d \in \{0, 1\}} \left\{ \sqrt n \left( \hat{\psi}_{zd}(Y_{(j)})-\breve \tau_{zd}(Y_{(j)}) \right) \right\}^2= o_P(1).
	\end{align}
\end{lemma}

\begin{proof} The result in \eqref{eq:blpluginnI} can be shown by arguments similar to the proof Lemma of \ref{lm:BLpluginsqrtn}. Recall that  Lemma \ref{lm:BLpluginsqrtn}  shows that $\frac{1}{|I_{\kappa}|}\sum_{j \in I_\kappa} \sum_{z, d \in \{0, 1\}} \left\{\sqrt n(\hat{F}_{zd}(Y_{(j)})-\overline F_{zd}(Y_{(j)}) \right\}^2= o_P(1)$, where $\hat{\bm F}$ is the maximum binomial likelihood estimate of $\bm F$, the vector of true distribution functions, when the proportion of the compliance classes are estimated by maximizing the binomial likelihood function. On the other hand, $\hat{\bm \theta}$ is the maximum binomial likelihood estimate of $\bm F$, when the proportion of the compliance classes are estimated by the plug-in estimates. Nevertheless, the proof of Lemma \ref{lm:BLpluginsqrtn} can be repeated verbatim to show \eqref{eq:blpluginnI}. 
	
	The result in \eqref{eq:blpluginnII} follows from Proposition \ref{proposition:nullBL} and the definition of $\{\breve \tau_{zd}\}_{z, d \in \{0, 1\}}$. 
\end{proof}

Using this lemma leading term of the asymptotic expansion of the binomial likelihood ratio test can be derived as follows: 

\begin{lemma}Let $\{\breve \tau_{zd}\}_{z, d \in \{0, 1\}}$ be as defined above. Then the binomial likelihood ratio test statistic satisfies
	\begin{align}
	T_n &= \frac{1}{|I_{\kappa}|}\sum_{j \in I_\kappa} \Bigg[ \sum_{z, d \in \{0, 1\}} n_{zd}\left\{  \frac{(\breve{\tau}_{zd}(Y_{(j)})-\overline F_{zd}(Y_{(j)}))^2}{\overline F_{zd}(Y_{(j)})(1-\overline F_{zd}(Y_{(j)}))} \right\}  \\
	&+ n_0 \frac{(\breve{\rho}_{at}(Y_{(j)}) - \breve{\phi}_{at})^2}{\breve{\phi}_{at}(1-\breve{\phi}_{at})} + n_1 \frac{(\breve{\rho}_{nt}(Y_{(j)}) - \breve{\phi}_{nt})^2}{\breve{\phi}_{nt}(1-\breve{\phi}_{nt})} \Bigg] + o_P(1)
	\label{blrt_approx}
	\end{align}
\end{lemma}

\begin{proof} Recall that the definitions of  $(\hat {\bm \theta}, \hat{\bm \chi})$ and $(\hat {\bm \psi}, \hat{\bm \xi})$ from \eqref{eq:BLRT}. Then, the binomial likelihood ratio test can be rewritten as, 
	\begin{align}\label{eq:blrt_I}
	T_n&= 2\left( \ell_{\bm{Y}, \bm{D} | \bm{Z}}(\hat{\bm \theta}, \hat{\bm\chi}) - \ell_{\bm{Y}, \bm{D} | \bm{Z}}(\hat{\bm \psi}, \hat{\bm\xi}) \right) \nonumber \\
	&= \frac{2n}{|I_{\kappa}|}\sum_{j \in I_\kappa} \left[ \left\{\sum_{\{z, d\}\in \{0, 1\}}  T_{zd}(Y_{(j)}|\hat {\theta}_{zd})- T_{zd}(Y_{(j)}|\hat{\psi}_{zd}) \right\} + \left\{\sum_{s \in \{nt, at\}} T_{s}(Y_{(j)}|\hat{\chi}_{s}) - T_{s}(Y_{(j)}|\hat{\xi}_{s})\right\} \right],
	\end{align}
	where 
	\begin{align*}
	T_{zd}(Y_{(j)} |\hat {\theta}_{zd})=&\frac{n_{zd}}{n}  \left\{ \overline F_{zd}(Y_{(j)}) \log \hat{\theta}_{zd}(Y_{(j)})  +(1- \overline F_{zd}(Y_{(j)}))  \log (1-\hat{\theta}_{zd}(Y_{(j)})) \right\} \\
	T_{zd}(Y_{(j)}|\hat{\psi}_{zd})=&\frac{n_{zd}}{n} \left\{ \overline F_{zd}(Y_{(j)}) \log \hat{\psi}_{zd}(Y_{(j)})  +(1- \overline F_{zd}(Y_{(j)}))  \log (1-\hat{\psi}_{zd}(Y_{(j)})) \right\}.
	\end{align*}
	and for $s \in \{nt, at\}$, 
	\begin{align*}
	T_{nt}(Y_{(j)} |\hat{\chi}_{nt})=&\frac{n_{1}}{n}  \left\{ \breve{\phi}_{nt}(Y_{(j)}) \log \hat{\chi}_{nt}(Y_{(j)})  +(1- \breve{\phi}_{nt}(Y_{(j)}))  \log (1-\hat{\chi}_{nt}(Y_{(j)})) \right\} \\
	T_{at}(Y_{(j)} |\hat{\chi}_{at})=&\frac{n_{0}}{n}  \left\{ \breve{\phi}_{at}(Y_{(j)}) \log \hat{\chi}_{at}(Y_{(j)})  +(1- \breve{\phi}_{at}(Y_{(j)}))  \log (1-\hat{\chi}_{at}(Y_{(j)})) \right\} \\
	T_{nt}(Y_{(j)} |\hat{\xi}_{nt})=&\frac{n_{1}}{n}  \left\{ \breve{\phi}_{nt}(Y_{(j)}) \log \hat{\xi}_{nt}(Y_{(j)})  +(1- \breve{\phi}_{nt}(Y_{(j)}))  \log (1-\hat{\xi}_{nt}(Y_{(j)})) \right\} \\
	T_{at}(Y_{(j)} |\hat{\xi}_{at})=&\frac{n_{0}}{n}  \left\{ \breve{\phi}_{at}(Y_{(j)}) \log \hat{\xi}_{at}(Y_{(j)})  +(1- \breve{\phi}_{at}(Y_{(j)}))  \log (1-\hat{\xi}_{at}(Y_{(j)})) \right\}.
	\end{align*}
	
	Recall the definition of the (negative) binary entropy function $I(x)=x\log x+(1-x)\log (1-x)$. Then, note that 
	\begin{align}\label{eq:difference}
	&T_{zd}(Y_{(j)}|\hat {\bm\theta})-I(\overline F_{zd}(Y_{(j)})) \nonumber \\
	&=\frac{n_{zd}}{n}\left\{ \overline F_{zd}(Y_{(j)}) \log \frac{\hat{\theta}_{zd}(Y_{(j)})}{\overline F_{zd}(Y_{(j)})}  +(1- \overline F_{zd}(Y_{(j)}))  \log \frac{1-\hat{\theta}_{zd}(Y_{(j)})}{1-\overline F_{zd}(Y_{(j)})} \right\} \nonumber \\
	&=R_{zd}^{(j)}, 
	\end{align}
	where 
	$$R_{zd}^{(j)}=\frac{n_{zd}}{n} \cdot \frac{(\hat{\theta}_{zd}(Y_{(j)})-\overline F_{zd}(Y_{(j)}))^2}{4} \bigg\{ \frac{\overline F_{zd}(Y_{(j)})}{(\omega_{zd}(Y_{(j)}))^2} - \frac{1-\overline F_{zd}(Y_{(j)})}{ (1-\omega_{zd}(Y_{(j)}))^2} \bigg\},$$
	and $\omega_{zd}(Y_{(j)}) \in [\overline F_{zd}(Y_{(j)})\wedge \hat\theta_{zd}(Y_{(j)}), \hat\theta_{zd}(Y_{(j)})\vee \overline F_{zd}(Y_{(j)})]$.

	Note that $\omega_{zd}(Y_{(j)}) \geq \overline F_{zd}(Y_{(\ceil{n\kappa})})\wedge \hat \theta_{zd}(Y_{(\ceil{n\kappa})})$ and $\overline F_{zd}(Y_{(j)}) \leq \overline F_{zd}(Y_{(\ceil{n(1-\kappa)})})$. Therefore, 
	$$\frac{\overline F_{zd}(Y_{(j)})}{(\omega_{zd}(Y_{(j)}))^2}  \leq \frac{\overline F_{zd}(Y_{(\ceil{n(1-\kappa)})})}{\overline F_{zd}(Y_{(n\kappa)})\wedge \hat \theta_{zd}(Y_{(n\kappa)})} = O_P(1),$$ since $\overline F_{zd}(Y_{(\ceil{n\kappa})})\pto H_{zd}^{-1}(\kappa)$, $\overline F_{zd}(Y_{(\ceil{n(1-\kappa)})})\pto H_{zd}^{-1}(1-\kappa)$ using Observation \ref{obs:orderstatistics}, and $|\hat \theta_{zd}(Y_{(\ceil{n\kappa})})- \overline F_{zd}(Y_{(\ceil{n\kappa})})|=o_P(1)$ by  Lemma \ref{lm:blplugin}. Similarly, 
	$$\frac{1-\overline F_{zd}(Y_{(j)})}{ (1-\omega_{zd}(Y_{(j)}))^2} =O_P(1).$$ 
	Therefore,  
	\begin{align} \label{eq:RuvI}
	\sum_{j \in I_\kappa}|R_{zd}^{(j)}|  & \leq O_P(1)\sum_{j \in I_\kappa} |\hat \theta_{zd}(Y_{(j)})- \overline F_{zd}(Y_{(j)})|^2 \nonumber \\
	& \leq O_P(1) \sum_{j \in I_\kappa} |\overline F_{zd}(Y_{(j)})-\hat \theta_{zd}(Y_{(j)})|^2= o_P(1),
	\end{align}
	by \eqref{eq:blpluginnI}. Therefore,  by \eqref{eq:difference}, 
	\begin{align}\label{eq:TuvI}
	\frac{1}{|I_{\kappa}|} \sum_{\{z, d\}\in \{0, 1\}} \sum_{j \in I_\kappa}  T_{zd}(Y_{(j)}|\hat {\bm\theta})-I(\overline F_{zd}(Y_{(j)})) 
	&=  o_P(1/n).
	\end{align}
	
	Similarly, by a second order Taylor approximation, 
	\begin{align}\label{eq:TuvII}
	T_{zd}(Y_{(j)}|\hat {\bm\psi})-I(\overline F_{zd}(Y_{(j)})) & = \frac{n_{zd}}{n} \cdot \frac{1}{2}\cdot \frac{(\hat{\psi}_{zd}(Y_{(j)})-\overline F_{zd}(Y_{(j)}))^2}{\overline F_{zd}(Y_{(j)})(1-\overline F_{zd}(Y_{(j)}))} + W_{zd}^{(j)}, 
	\end{align}
	where 
	$$W_{zd}^{(j)}=\frac{n_{zd}}{n} \cdot \frac{(\hat{\psi}_{zd}(Y_{(j)})-\overline F_{zd}(Y_{(j)}))^3}{6} \bigg\{ \frac{\overline F_{zd}(Y_{(j)})}{(\omega_{zd}(Y_{(j)}))^3} - \frac{1-\overline F_{zd}(Y_{(j)})}{ (1-\omega_{zd}(Y_{(j)}))^3} \bigg\},$$
	and $\omega_{zd}(Y_{(j)}) \in [\overline F_{zd}(Y_{(j)})\wedge \hat\psi_{zd}(Y_{(j)}), \hat\psi_{zd}(Y_{(j)})\vee \overline F_{zd}(Y_{(j)})]$. Now, as in \eqref{eq:RuvI}, 
	\begin{align} \label{eq:Wuv}
	\frac{1}{|I_{\kappa}|} \sum_{j \in I_\kappa}|W_{zd}^{(j)}|  & \leq O_P(1)\frac{1}{|I_{\kappa}|}\sum_{j \in I_\kappa} |\hat \psi_{zd}(Y_{(j)})- \overline F_{zd}(Y_{(j)})|^3 \nonumber \\
	& \leq O_P(1) \frac{1}{|I_{\kappa}|} \sum_{j \in I_\kappa} |\overline F_{zd}(Y_{(j)})-\hat \psi_{zd}(Y_{(j)})|^3 \nonumber \\
	& \leq O_P(1) \frac{1}{|I_{\kappa}|} \sum_{j \in I_\kappa} |\overline F_{zd}(Y_{(j)})-\breve \tau_{zd}(Y_{(j)})|^3 +  O_P(1) \frac{1}{|I_{\kappa}|} \sum_{j \in I_\kappa} |\hat \psi_{zd}(Y_{(j)})-\breve \tau_{zd}(Y_{(j)})|^3 \nonumber \\
	& \leq O_P(n^{-\frac{3}{2}}) +  O_P(1) \frac{1}{|I_{\kappa}|} \left(\sum_{j \in I_\kappa} |\hat \psi_{zd}(Y_{(j)})-\breve \tau_{zd}(Y_{(j)})|^2 \right)^{\frac{3}{2}} \nonumber \\
	&=o_P(1/n), 
	\end{align}
	using $\sup_t|\overline F_{zd}(t)-\breve \tau_{zd}(t)|=O_P(1/\sqrt n)$ for the first term, and Cauchy-Schwarz followed by \eqref{eq:blpluginnII} in the second term. 
	
	Therefore, combing \eqref{eq:TuvII} with \eqref{eq:Wuv} above gives, 
	\begin{align}\label{eq:TuvIII}
	& \frac{2n}{|I_{\kappa}|} \sum_{\{z, d\}\in \{0, 1\}} \sum_{j \in I_\kappa}  T_{zd}(Y_{(j)}|\hat {\bm\psi})-I(\overline F_{zd}(Y_{(j)})) \nonumber \\
	& =\frac{2n}{|I_{\kappa}|} \sum_{\{z, d\}\in \{0, 1\}} \sum_{j \in I_\kappa}  \left\{ \frac{n_{zd}}{n} \cdot \frac{1}{2}\cdot \frac{(\hat{\psi}_{zd}(Y_{(j)})-\overline F_{zd}(Y_{(j)}))^2}{\overline F_{zd}(Y_{(j)})(1-\overline F_{zd}(Y_{(j)}))} +  o_P(1/n)\right\} \nonumber \\
	& =\frac{1}{|I_{\kappa}|} \left\{ \sum_{\{z, d\}\in \{0, 1\}} \sum_{j \in I_\kappa}  n_{zd}\cdot \frac{(\breve{\tau}_{zd}(Y_{(j)})-\overline F_{zd}(Y_{(j)}))^2}{\overline F_{zd}(Y_{(j)})(1-\overline F_{zd}(Y_{(j)}))} \right\} +  o_P(1), 
	\end{align}
	where the last step uses triangle inequality and \eqref{eq:blpluginnII}. 
	
	Similarly, we can deduce 
	\begin{align*}
	&\frac{2n}{|I_{\kappa}|}\sum_{j \in I_\kappa} \sum_{s \in \{nt, at\}} T_{s}(Y_{(j)}|\hat{\chi}_{s}) - T_{s}(Y_{(j)}|\hat{\xi}_{s}) \\
	&= \frac{1}{|I_{\kappa}|} \left\{ \sum_{j \in I_\kappa} n_0 \frac{(\breve{\rho}_{at}(Y_{(j)}) - \breve{\phi}_{at})^2}{\breve{\phi}_{at}(1-\breve{\phi}_{at})} + n_1 \frac{(\breve{\rho}_{nt}(Y_{(j)}) - \breve{\phi}_{nt})^2}{\breve{\phi}_{nt}(1-\breve{\phi}_{nt})} \right\} + o_P(1). 
	\end{align*}
	By combining this and \eqref{eq:TuvIII} with \eqref{eq:blrt_I}, the proof can be completed. 
\end{proof}

The proof of Theorem 2 can now be completed by simplifying the RHS of \eqref{blrt_approx}. The following shows some steps for simplifying this equation. First, for each $j$, the sum $\sum_{z, d \in \{0, 1\}} n_{zd} \left\{\frac{(\breve{\tau}_{zd}(Y_{(j)})-\overline F_{zd}(Y_{(j)}))^2}{\overline F_{zd}(Y_{(j)})(1-\overline F_{zd}(Y_{(j)}))} \right\}$ can be simplified as
\begin{align*}
&\sum_{z, d \in \{0, 1\}} n_{zd} \left\{\frac{(\breve{\tau}_{zd}(Y_{(j)})-\overline F_{zd}(Y_{(j)}))^2}{\overline F_{zd}(Y_{(j)})(1-\overline F_{zd}(Y_{(j)}))} \right\} \\
&= \sum_{z, d \in \{0. 1\}} \eta_{zd} \left\{ \sqrt n (\breve{\tau}_{zd}(Y_{(j)}) - \overline{F}_{zd}(Y_{(j)})) \right\}^2 Q_{zd}(Y_{(j)}) + o_P(1)\\
&= n \frac{\phi_1(1-\phi_1)}{Q(Y_{(j)})} \left\{(\overline{F}_0(t) - \overline{F}_1(t)) - \frac{r_{at}(t)E_1(t) + r_{nt}(t)E_2(t)}{det} \right\}^2 +o_P(1)\\
&= n \frac{\phi_1(1-\phi_1)}{Q(Y_{(j)})} (\overline{F}_0(Y_{(j)}) - \overline{F}_1 (Y_{(j)}))^2 \left( \frac{1}{det(Y_{(j)})} \right)^2 \left( \frac{\phi_1(1-\phi_1)}{\phi_{nt}(1-\phi_{nt})\phi_{at}(1-\phi_{at})} \right)^2 +o_P(1)
\end{align*}
and the other sum is 
\begin{align*}
& n_0 \frac{(\breve{\rho}_{at}(Y_{(j)}) - \breve{\phi}_{at})^2}{\breve{\phi}_{at}(1-\breve{\phi}_{at})} + n_1 \frac{(\breve{\rho}_{nt}(Y_{(j)}) - \breve{\phi}_{nt})^2}{\breve{\phi}_{nt}(1-\breve{\phi}_{nt})} \\
&= (1-\phi_1) \frac{\left\{ \sqrt n(\breve{\rho}_{at}(Y_{(j)}) - \breve{\phi}_{at} ) \right\}^2}{\phi_{at}(1-\phi_{at})}  + \phi_1 \frac{\left\{ \sqrt n(\breve{\rho}_{nt}(Y_{(j)}) - \breve{\phi}_{nt} ) \right\}^2}{\phi_{nt}(1-\phi_{nt})} +o_P(1)\\
&= n \frac{\phi_1(1-\phi_1)}{Q(Y_{(j)})} (\overline{F}_0(Y_{(j)}) - \overline{F}_1 (Y_{(j)}))^2 \left( \frac{1}{det(Y_{(j)})} \right)^2 \\
&~~~~~ \times \left( \frac{\phi_1(1-\phi_1)}{\phi_{nt}(1-\phi_{nt})\phi_{at}(1-\phi_{at})} \right)\left\{ \frac{\phi_1(1-\phi_1)^2}{\phi_{at}(1-\phi_{at})} \frac{E_1^2(Y_{(j)})}{Q(Y_{(j)})} +  \frac{\phi_1^2(1-\phi_1)}{\phi_{nt}(1-\phi_{nt})} \frac{E_2^2(Y_{(j)})}{Q(Y_{(j)})} \right\} + o_P(1)
\end{align*}

Therefore, $T_n$ is 
\begin{align*}
&T_n \\
&= \frac{1}{|I_{\kappa}|} \sum_{j \in I_\kappa} n \frac{\phi_1(1-\phi_1)}{Q(Y_{(j)})} (\overline{F}_0(Y_{(j)}) - \overline{F}_1 (Y_{(j)}))^2 \left( \frac{1}{det(Y_{(j)})} \right)^2 \left( \frac{\phi_1(1-\phi_1)}{\phi_{nt}(1-\phi_{nt})\phi_{at}(1-\phi_{at})} \right) + o_P(1)\\
&= \frac{1}{|I_{\kappa}|} \sum_{j \in I_\kappa} n \cdot \phi_1(1-\phi_1)(\overline{F}_0(Y_{(j)}) - \overline{F}_1 (Y_{(j)}))^2 \frac{1}{\frac{\phi_{nt}(1-\phi_{nt})\phi_{at}(1-\phi_{at})}{\phi_1(1-\phi_1)} + Q(Y_{(j)}) det(Y_{(j)})} + o_P(1)\\
&= \frac{1}{|I_{\kappa}|} \sum_{j \in I_\kappa} \frac{n_0 n_1}{n}(\overline{F}_0(Y_{(j)}) - \overline{F}_1 (Y_{(j)}))^2 \\
&~~~~~\times \frac{1}{Q(Y_{(j)}) + (1-\phi_1)\phi_{nt}(1-\phi_{nt})E_1^2(Y_{(j)}) +\phi_1 \phi_{at}(1-\phi_{at})E_2^2(Y_{(j)}) } + o_P(1)
\end{align*}

Since this denominator of the last term is 
\begin{align*}
& Q(Y_{(j)}) + (1-\phi_1)\phi_{nt}(1-\phi_{nt})E_1^2(Y_{(j)}) +\phi_1 \phi_{at}(1-\phi_{at})E_2^2(Y_{(j)}) \\
&= \left\{\frac{\phi_1(1-\phi_{at})}{Q_{00}(Y_{(j)})} + \frac{\phi_1 \phi_{at}}{Q_{01}(Y_{(j)})} + \phi_1 \phi_{at}(1-\phi_{at}) E_2^2(Y_{(j)}) \right\} \\
&~~~~~+ \left\{\frac{(1-\phi_1)\phi_{nt}}{Q_{10}(Y_{(j)})} + \frac{(1-\phi_1)(1- \phi_{at})}{Q_{11}(Y_{(j)})} +(1-\phi_1) \phi_{nt}(1-\phi_{nt}) E_1^2(Y_{(j)}) \right\} \\
&= \phi_1 \{ (1-\phi_{at})F_{00}(Y_{(j)}) + \phi_{at}F_{01}(Y_{(j)})\}\{1-(1-\phi_{at})F_{00}(Y_{(j)}) - \phi_{at}F_{01}(Y_{(j)})\} \\
&~~~~~+ (1-\phi_1)\{ (1-\phi_{nt})F_{11}(Y_{(j)}) + \phi_{nt}F_{10}(Y_{(j)})\}\{1-(1-\phi_{nt})F_{11}(Y_{(j)}) - \phi_{nt}F_{10}(Y_{(j)})\} \\
&= \phi_1 F(Y_{(j)})(1-F(Y_{(j)})) + (1-\phi_1) F(Y_{(j)})(1-F(Y_{(j)})) \\
&= F(Y_{(j)})(1-F(Y_{(j)}))
\end{align*}

By substituting $F(Y_{(j)})(1-F(Y_{(j)}))$ in the denominator, we have 
\begin{align*}
T_n &= \frac{1}{|I_{\kappa}|} \sum_{j \in I_\kappa} \frac{n_0 n_1}{n} \frac{(\overline{F}_0(Y_{(j)}) - \overline{F}_1 (Y_{(j)}))^2}{F(Y_{(j)})(1-F(Y_{(j)}))} + o_P(1) \\
&= \frac{1}{|I_{\kappa}|} \sum_{j \in I_\kappa} \frac{n_0 n_1}{n} \frac{(\overline{F}_0(Y_{(j)}) - \overline{F}_1 (Y_{(j)}))^2}{\overline{H}(Y_{(j)})(1-\overline{H}(Y_{(j)}))} + o_P(1)\\
\text{or equivalently} & \frac{n_0 n_1}{n} \int_{\overline{J}_{\kappa}}\frac{(\overline{F}_0(t) - \overline{F}_1 (t)^2}{\overline{H}(t)(1-\overline{H}(t))} {\rm d}\overline{H}(t) + o_P(1)
\end{align*}
where $\overline{H}(t) = (n_0 \overline{F}_0(t) + n_1 \overline{F}_1(t))/n$.

\subsection{Proof of Theorem 3}
\label{subsec:thm3}

In the definition of the simple version binomial likelihood, the function $J(\overline{F}_0(t_j), \theta_0(t_j))$ is maximized when $\theta_0(t_j) = \overline{F}_0(t_j)$. Therefore, $(\overline{F}_0, \overline{F}_1) = \argmax_{\theta_0, \theta_1 \in \in \cP([0, 1]^\R)}\ell^{simple}_{\bm{Y}|\bm{Z}}(\bm\theta)$. Similarly, when $\theta_0=\theta_1$, the simple version binomial likelihood is $(1/m) \sum_{j=1}^{m} nJ(\overline{H}(t_j), \theta_0(t_j))$, and the maximum is attained at $\theta_0 = \overline{H}$.  This completes the proof of the first part. For the second part, by using the fact that
$$
x \log\left(\frac{x}{y}\right) + (1-x) \log\left(\frac{1-x}{1-y}\right) = \frac{1}{2} \frac{(x-y)^2}{y(1-y)} + o(|x-y|^2),
$$
the test statistic $T_n^{simple}$ with the knots $\bm{t} = (Y_{(1)}, \ldots, Y_{(n)})$ can be written as 
\begin{align*}
T_n^{simple} &= \frac{1}{n} \sum_{j=1}^{n} n_0\frac{(\overline{F}_0(Y_{(j)}) - \overline{H}(Y_{(j)}))^2}{\overline{H}(Y_{(j)}) (1-\overline{H}(Y_{(j)}))} + n_1 \frac{(\overline{F}_1(Y_{(j)}) - \overline{H}(Y_{(j)}))^2}{\overline{H}(Y_{(j)}) (1-\overline{H}(Y_{(j)}))} + o_P(1) \\
&=  \frac{n_0 n_1}{n} \int_{-\infty}^{\infty} \frac{\bar{F}_{0}(t) - \bar{F}_{1}(t)}{\bar{H}(t) (1-\bar{H}(t))} {\rm d} \bar{H}(t) +o_P(1),
\end{align*}
completing the proof. 

\section{Proofs from Section 5}
\label{sec:algoapp}

In this section we recall the well-known pool-adjacent violators algorithm \citep{barlow1972, deleeuw2009},  elaborate on our proposed EM-PAVA algorithm, and prove Proposition 2.

\subsection{The pool-adjacent-violators algorithm (PAVA)}
\label{sec:pava}

The PAVA takes input a vector  $\bm{z}=(u_1, \ldots , u_n)$ and an a weight vector $\bm{w} = (w_1, \ldots, w_n)$, and returns another vector $\text{PAVA}_{\bm w}(\bm u)=(\hat u_1, \ldots , \hat u_n)$ such that 
\begin{equation}
\text{PAVA}_{\bm w}(\bm u)=\text{arg} \min_{v_1 \leq v_2 \leq \cdots \leq v_n} \sum_{i=1}^{n} w_i  (u_i-v_i)^2.
\label{eq:pavaobj}
\end{equation}

The weighted PAVA is as follows: To begin with set $\hat u_a=u_a$ for all $ a \in [n]$. 

\begin{enumerate}
	
	\item[Step 1.] If $\hat u_1 \leq \hat u_2$, move to Step 2. Otherwise, $\hat u_1 > \hat u_2$ in which case the values are updated as $$\hat u_1=\hat u_2\leftarrow \frac{w_1 u_1 + w_2 u_2}{w_1 + w_2},$$ the weighted average of the original values of $\{u_1, u_2 \}$. Then, move to Step 2. Note that the first step does not update the points from the third to the last, that is, $\hat u_a =u_a$, for $a \in [3, n]$. 
	
	\item[Step 2.] For the $a$-th point, compare $\hat u_a$ with $\hat u_{a+1}$. If $\hat u_a \leq \hat u_{a+1}$, then $\hat u_a$ remains the same and the algorithm moves to the next point. If $\hat u_a > \hat u_{a+1}$, then $\hat u_a=\hat u_{a+1}\leftarrow \frac{w_a \hat u_a + w_{a+1} \hat u_{a+1}}{w_1 + w_2}$, the weighted average of $\{\hat u_a, \hat u_{a+1}\}$. Then new value is compared with $\hat u_{a-1}$. If the required monotonicity assumption is achieved, that is, $\hat u_{a-1} \leq \hat u_a$, then the algorithm moves to the $(a+1)$-th point. Otherwise, $\hat u_{a-1} > \hat u_a$, in which case $\hat u_{a-1}=\hat u_a=\hat u_{a+1}$ is updated by the weighted average of $\{\hat u_{a-1}, \hat u_a, \hat u_{a+1}\}$. This repeated until a sequence the partial sequence $\hat u_1, \ldots, \hat u_a$ is non-decreasing. Then the algorithm moves to the $(a+1)$-th point. 
\end{enumerate}

It is well known that the output $\text{PAVA}_{\bm w}(\bm{z})=(\hat u_1, \hat u_2, \ldots, \hat u_n)$ of the above algorithm is non-decreasing and is the solution for the optimization problem~(\ref{eq:pavaobj}). For example, suppose $\bm{z}=(3,2,1)$ and $\bm w=(1/3, 1/3, 1/3)$. Then PAVA updates $\bm u$ in the following order,
$$(3,2,1) \rightarrow (5/2, 5/2, 1) \rightarrow (5/2, 7/4, 7/4) \rightarrow (2,2,2).$$
For our experiments, we have used the \texttt{pava} function in the $\mathrm{R}$ package \texttt{Iso}, which implements weighted PAVA described.

\subsection{Proofs from the expectation step of the EM-PAVA Algorithm}
\label{sec:pfE}

To begin recall the definition of $\overline{K}_{zs}^{ij}=\{Z_i=z, S_i=s \>|\> t_j\}$. The complete-data binomial log-likelihood can be re-written as follows:
\begin{align}\label{eq:loglc}
\log \overline L(\bm \theta, \bm \chi|\overline \cD_n)= \frac{1}{m} \sum_{i=1}^{n}\sum_{j=1}^{m} \sum_{z \in\{0, 1\}} \sum_{s \in \{co, nt, at\}}S_{z, s}(\overline{\cD}_{n,i}, t_j),
\end{align} 
where $\overline{\cD}_{n,i} = (Z_i, S_i, D_i, Y_i)$ and 
{\small
\begin{align*}
S_{z, co}(\overline{\cD}_{n,i}, t_j) &=\bm 1 \{\overline{K}_{z, co}^{ij}\} \log \chi_{co}(t_j)  + \bm 1\{Y_i \leq t_j, \overline{K}_{z, co}^{ij} \} \log \theta_{co}^{(z)}(t_j) +\bm 1 \{ Y_i > t_j , \overline{K}_{u, co}^{ij}\} \log (1-\theta_{co}^{(z)}(t_j)), \\
S_{z, nt}(\overline{\cD}_{n,i}, t_j) &=\bm 1 \{\overline{K}_{z, nt}^{ij}\} \log \chi_{nt}(t_j)  + \bm 1\{Y_i \leq t_j, \overline{K}_{z, nt}^{ij} \} \log \theta_{nt}(t_j) +\bm 1 \{ Y_i > t_j, \overline{K}_{u, nt}^{ij}\} \log (1-\theta_{nt}(t_j)), \\
S_{z, at}(\overline{\cD}_{n,i}, t_j) &=\bm 1 \{\overline{K}_{z, at}^{ij}\} \log \chi_{at}(t_j)  + \bm 1\{Y_i \leq t_j, \overline{K}_{z, at}^{ij} \} \log \theta_{at}(t_j) +\bm 1 \{ Y_i > t_j, \overline{K}_{u, at}^{ij}\} \log (1-\theta_{at}(t_j)). 
\end{align*}}
\normalsize
We have 
\begin{align}\label{expected_log_likelihood}
Q_k(\bm\theta,   \bm\chi|\hat{\bm{\theta}}_{(k)},  \hat{\bm{\chi}}_{(k)})  = \frac{1}{m}\sum_{j=1}^{m} \sum_{z \in\{0,1\}} \sum_{s \in \{co, nt, at\}}  Q_{z, s}(t_j) ,
\end{align}
where, for $z \in \{0, 1\}$, $Q_{z, s}(t_j)= \E_{\bm{\hat \theta}_{(k)},  \bm{\hat \chi}_{(k)}}\left(\sum_{i=1}^{n} S_{z, s}(Y_i, t_j)|\cD_n\right)$. To compute \eqref{expected_log_likelihood}, we need to compute the following probabilities: 
\begin{align*}
u_0^{(k)}(t_j) &= \pr_{\bm{\hat \theta}_{(k)},  \bm{\hat \chi}_{(k)}}(S_i=co \> | \> Z_i=0, D_i=0, Y_i \leq t_j)  \\
u_1^{(k)}(t_j) &= \pr_{\bm{\hat \theta}_{(k)},  \bm{\hat \chi}_{(k)}}(S_i=co \> | \> Z_i=1, D_i=1, Y_i \leq t_j) \\
v_0^{(k)}(t_j) &= \pr_{\bm{\hat \theta}_{(k)},  \bm{\hat \chi}_{(k)}}(S_i=co \> | \> Z_i=0, D_i=0, Y_i > t_j)  \\
v_1^{(k)}(t_j) &= \pr_{\bm{\hat \theta}_{(k)},  \bm{\hat \chi}_{(k)}}(S_i=co \> | \> Z_i=1, D_i=1, Y_i > t_j).
\end{align*}

\begin{lemma}\label{lm:Estep} Let $u_{0}^{(k)}(t_j) , u_{1}^{(k)}(t_j) , v_{0}^{(k)}(t_j), v_{0}^{(k)}(t_j)$ be as defined above. Then, for $z \in \{0, 1\}$, 
	\begin{align*}
	Q_{z, co}(t_j)= &n_{zz}\Bigg\{\left(\overline F_{zz}(t_j) u_z^{(k)}(t_j) + (1-\overline F_{zz}(t_j) )v_z^{(k)}(t_j) \right) \log(1-{\chi}_{nt}(t_j)-{\chi}_{at}(t_j)) \\
	& ~+ \overline F_{zz}(t_j)  u_z^{(k)}(t_j)   \log \theta_{co}^{(z)}(t_j)  + (1-\overline F_{zz}(t_j))  v_z^{(k)}(t_j) \log (1-\theta_{co}^{(z)}(t_j)) \Bigg\}.  
	\end{align*}
	Similarly, 
	\begin{align*}
	Q_{0, nt}(t_j)&= n_{00} \Bigg\{\left( \overline F_{00}(t_j)(1-u_0^{(k)}(t_j)) + (1-\overline F_{00}(t_j)(1-v_0^{(k)}(t_j) ) \right) \log{\chi}_{nt}(t_j) \\
	& ~ +  \overline F_{00}(t_j) (1-u_0^{(k)}(t_j))   \log \theta_{nt}(t_j) +  (1-\overline F_{00}(t_j) (1-v_0^{(k)}(t_j)) \log (1-\theta_{nt}(t_j))\Bigg\}, 
	\end{align*} 
	and 
	\begin{align*}
	Q_{1, at}(t_j)&=n_{11}  \Bigg\{\left( \overline F_{11}(t_j)(1-u_1^{(k)}(t_j)) + (1-\overline F_{11}(t_j)(1-v_1^{(k)}(t_j))  \right)  \log{\chi}_{at}(t_j) \\
	& + \overline F_{11}(t_j) (1-u_1^{(k)}(t_j))   \log \theta_{at}(t_j)  + (1-\overline F_{11}(t_j) (1-v_1^{(k)}(t_j)) \log (1-\theta_{at}(t_j)) \Bigg\}. 
	\end{align*}
	Finally, $Q_{1, nt}(t_j)= n_{10} \log\chi_{nt}(t_j) +  n_{10} J(\overline F_{10}(t_j), \theta_{nt}(t_j))$ and $Q_{0, at}(t_j)= n_{01} \log\chi_{at}(t_j) +  n_{01}J (\overline F_{01}(t_j) , \theta_{at}(t_j))$, where $J(x, y)=x \log y + (1-x) \log (1-y)$. 
\end{lemma}

The proof of the above lemma is an easy consequence of Lemma \ref{lm:br} below. This completes the proof of  the expectation step of the EM algorithm, at the $(m+1)$-th iteration.

\begin{lemma}\label{lm:br}For every integer $m \geq 1$,
	\begin{align*}
	u_0^{(k)}(t_j) &= \frac{ \left\{ \frac{1}{|I_{\kappa}|} \sum_{j \in I_{\kappa}}\hat \chi_{co, (k)}(t_j) \right\} \hat \theta_{co,  (k)}^{(0)}(t_j)}{\left\{ \frac{1}{|I_{\kappa}|} \sum_{j \in I_{\kappa}}\hat \chi_{co, (k)}(t_j) \right\} \hat \theta_{co,  (k)}^{(0)}(t_j) + \left\{ \frac{1}{|I_{\kappa}|} \sum_{j \in I_{\kappa}}\hat \chi_{nt, (k)}(t_j) \right\} \hat \theta_{nt, (k)}(t_j)}, \\
	u_1^{(k)}(t_j) &= \frac{\left\{ \frac{1}{|I_{\kappa}|} \sum_{j \in I_{\kappa}}\hat \chi_{co, (k)}(t_j) \right\} \hat \theta_{co,  (k)}^{(1)}(t_j)}{\left\{ \frac{1}{|I_{\kappa}|} \sum_{j \in I_{\kappa}}\hat \chi_{co, (k)}(t_j)\right\} \hat \theta_{co,  (k)}^{(1)}(t_j) + \left\{ \frac{1}{|I_{\kappa}|} \sum_{j \in I_{\kappa}}\hat \chi_{at, (k)}(t_j)\right\} \hat \theta_{at, (k)}(t_j)}, \\
	v_0^{(k)}(t_j) &= \frac{\left\{ \frac{1}{|I_{\kappa}|} \sum_{j \in I_{\kappa}}\hat \chi_{co, (k)}(t_j)\right\} (1- \hat \theta_{co,  (k)}^{(0)}(t_j) )}{\left\{ \frac{1}{|I_{\kappa}|} \sum_{j \in I_{\kappa}}\hat \chi_{co, (k)}(t_j)\right\}( 1-\hat \theta_{co,  (k)}^{(0)}(t_j)) + \left\{ \frac{1}{|I_{\kappa}|} \sum_{j \in I_{\kappa}}\hat \chi_{nt, (k)}(t_j) \right\} ( 1- \hat \theta_{nt, (k)}(t_j) ) }, \\
	v_1^{(k)}(t_j) &= \frac{\left\{ \frac{1}{|I_{\kappa}|} \sum_{j \in I_{\kappa}}\hat \chi_{co, (k)}(t_j)\right\} (1- \hat \theta_{co,  (k)}^{(1)}(t_j) )}{\left\{ \frac{1}{|I_{\kappa}|} \sum_{j \in I_{\kappa}}\hat \chi_{co, (k)}(t_j)\right\}( 1-\hat \theta_{co,  (k)}^{(1)}(t_j)) + \left\{ \frac{1}{|I_{\kappa}|} \sum_{j \in I_{\kappa}}\hat \chi_{at, (k)}(t_j)\right\} ( 1- \hat \theta_{at,  (k)}(t_j) )}. 
	\end{align*}
	where $\hat{\chi}_{co, (k)}(t_j) = 1-\hat{\chi}_{at, (k)}(t_j)-\hat{\chi}_{nt, (k)}(t_j)$. 
\end{lemma}

\begin{proof}Throughout the proof, we denote $\pr= \pr_{\bm{\hat \theta}_{(k)},  \bm{\hat \chi}_{(k)}}$ for notational simplicity.  
	
	To begin with, note that 
	\begin{align}\label{eq:beta0}
	u_0^{(k)}(t_j)&=\pr(S_i=co \> | \> Z_i=0, D_i=0, Y_i \leq t_j) \nonumber \\
	&= \frac{\pr(S_i=co, Z_i=0, D_i=0, Y_i \leq t_j)}{\pr(S_i=co, Z_i=0, D_i=0, Y_i \leq t_j) + \pr(S_i=nt, Z_i=0, D_i=0, Y_i \leq t_j)}.
	\end{align}
	
	Now, 
	\begin{align}\label{eq:n}
	&\frac{\pr(S_i=co, Z_i=0, D_i=0, Y_i \leq t_j)}{\pr(Z_i=0, D_i=0)} \nonumber \\
	&=\pr( Y_i \leq t_j \> | \> S_i=co, Z_i=0, D_i=0) \cdot \pr(S_i=co \> | \> Z_i=0, D_i=0) \nonumber \\
	&=\frac{ \frac{1}{|I_{\kappa}|} \sum_{j \in I_{\kappa}} \hat \chi_{co, (k)}(t_j)}{\frac{1}{|I_{\kappa}|} \sum_{j \in I_{\kappa}}\hat \chi_{co, (k)}(t_j)+\frac{1}{|I_{\kappa}|} \sum_{j \in I_{\kappa}}\hat \chi_{nt, (k)}(t_j)} \pr( Y_i \leq t_j \> | \> S_i=co, Z_i=0)  \nonumber \\
	&=\frac{\frac{1}{|I_{\kappa}|} \sum_{j \in I_{\kappa}}\hat \chi_{co, (k)}(t_j)}{\frac{1}{|I_{\kappa}|} \sum_{j \in I_{\kappa}}\hat \chi_{co, (k)}(t_j)+\frac{1}{|I_{\kappa}|} \sum_{j \in I_{\kappa}}\hat \chi_{nt, (k)}(t_j)}\hat \theta_{co,  (k)}^{(0)}(t_j).
	\end{align}
	Moreover, 
	\begin{align}\label{eq:d}
	& \frac{\pr(S_i=co, Z_i=0, D_i=0, Y_i \leq t_j) + \pr(S_i=nt, Z_i=0, D_i=0, Y_i \leq t_j)}{\pr(Z_i=0, D_i=0)} \nonumber \\
	&=\pr( Y_i \leq t_j \> | \> S_i=co, Z_i=0, D_i=0) \cdot \pr(S_i=co \> | \> Z_i=0, D_i=0) \nonumber \\
	&+ \pr( Y_i \leq t_j \> | \> S_i=nt, Z_i=0, D_i=0) \cdot \pr(S_i=nt \> | \> Z_i=0, D_i=0) \nonumber \\
	&=\frac{\frac{1}{|I_{\kappa}|} \sum_{j \in I_{\kappa}}\hat \chi_{co, (k)}(t_j)}{\frac{1}{|I_{\kappa}|} \sum_{j \in I_{\kappa}}\hat \chi_{co, (k)}(t_j)+\frac{1}{|I_{\kappa}|} \sum_{j \in I_{\kappa}}\hat \chi_{nt, (k)}(t_j)}\hat \theta_{co,  (k)}^{(0)}(t_j) \\
	&~~~~~+ \frac{\frac{1}{|I_{\kappa}|} \sum_{j \in I_{\kappa}}\hat \chi_{nt, (k)}(t_j)}{{\frac{1}{|I_{\kappa}|} \sum_{j \in I_{\kappa}}\hat \chi_{co, (k)}(t_j)+\frac{1}{|I_{\kappa}|} \sum_{j \in I_{\kappa}}\hat \chi_{nt, (k)}(t_j)}}\hat \theta_{nt, (k)}(t_j). 
	\end{align}
	Substituting \eqref{eq:n} and \eqref{eq:d} in \eqref{eq:beta0} the identity for $u_0^{(k)}(t_j)$ follows. The other identities can be proved similarly. 
\end{proof}

\subsection{Proofs from the maximization step of the EM-PAVA Algorithm}
\label{sec:pfM}

In the maximization step, unrestricted maximizers of the expectation are defined as $$(\breve{\bm\theta}_{(k+1)}, \breve{\bm \chi}_{(k+1)})=\arg\max_{\bm\theta \in \bm \vartheta,  \bm\chi \in \bm\varphi}Q_{k}(\bm\theta,  \bm\chi| \hat{\bm\theta}_{(k)},  \hat{\bm \chi}_{(k)}),$$ 
where $\breve{\bm \theta}_{(k+1)}(t)=(\breve \theta_{co, (k+1)}^{(0)}(t), \breve  \theta_{nt, (k+1)}(t), \theta_{co, (k+1)}^{(1)}(t), \breve \theta_{at, (k+1)}(t))$ and $\breve{\bm \chi}_{(k+1)}(t)=(\breve \chi_{nt, (k+1)}(t), \breve \chi_{at, (k+1)}(t))$. 

%

\begin{lemma}\label{lm:mstep} Let $u_{0}^{(k)}(t_j) , u_{1}^{(k)}(t_j) , v_{0}^{(k)}(t_j), v_{0}^{(k)}(t_j)$ be as in Lemma \ref{lm:br}.  Then
	\begin{align*}
	\breve{\theta}_{co, (k+1)}^{(0)}(t_j) &= \frac{\overline F_{00}(t_j)   u_0^{(k)}(t_j)}{\overline F_{00}(t_j)   u_0^{(k)}(t_j) + (1-\overline F_{00}(t_j))   v_0^{(k)}(t_j)}, \\
	\breve{\theta}_{nt, (k+1)}(t_j) &= \frac{n_{00}\overline F_{00}(t_j)   (1-u_0^{(k)}(t_j)) + n_{10}\overline F_{10}(t_j)}{n_{00}\overline F_{00}(t_j) (1-u_0^{(k)}(t_j)) + n_{00}(1-\overline F_{00}(t_j))(1-v_0^{(k)}(t_j)) + n_{10}}, \\
	\breve{\theta}_{co, (k+1)}^{(1)}(t_j) &= \frac{ \overline F_{11}(t_j)   u_1^{(k)}(t_j)}{ \overline F_{11}(t_j)   u_1^{(k)}(t_j) +  (1-\overline F_{11}(t_j))   v_1^{(k)}(t_j)}, \\
	\breve{\theta}_{at, (k+1)}(t_j) &= \frac{n_{11}\overline F_{11}(t_j)   (1-u_1^{(k)}(t_j)) + n_{01}\overline F_{01}(t_j)}{n_{11}\overline F_{11}(t_j) (1-u_1^{(k)}(t_j)) + n_{11}(1-\overline F_{11}(t_j))(1-v_1^{(k)}(t_j)) + n_{01}};
	\end{align*}
	and 
	\begin{align*}
	\breve \chi_{nt, (k+1)}(t_j) &= \frac{1}{n}\left\{n_{00}\overline F_{00}(t_j)(1-u_0^{(k)}(t_j)) + n_{00}(1-\overline F_{00}(t_j))(1-v_0^{(k)}(t_j)) + n_{10}\right\}  \\
	\breve \chi_{at, (k+1)}(t_j) &= \frac{1}{n}\left\{ n_{01} + n_{11}\overline F_{11}(t_j)(1-u_1^{(k)}(t_j)) + n_{11}(1-\overline F_{11}(t_j))(1-v_1^{(k)}(t_j))\right\}.
	\end{align*}
	Moreover, 
	$$\breve{\bm \chi}_{(k+1)}(t_j)=(\breve \chi_{nt, (k+1)}(t_j) , \breve \chi_{at, (k+1)}(t_j) ) \in \bm\varphi_{+} = [0, 1]^2_+.$$
	That is, $\breve \chi_{nt, (k+1)}(t_j) , \breve \chi_{at, (k+1)}(t_j) \in [0, 1]$ and $0 \leq \breve\chi_{nt, (k+1)}(t_j) + \breve\chi_{at, (k+1)}(t_j) \leq 1$. 
\end{lemma}

\begin{proof}This follows from Lemma \ref{lm:Estep}, by solving the first-order conditions obtained by taking the gradient of the $Q_{k}(\bm\theta,  \bm\chi| \hat{\bm\theta}_{(k)},  \bm{\hat \chi}_{(k)})$ with respect to $\bm \theta(t_j)$ and $\bm \chi(t_j)$, and equating it to zero.
	
	To see, $\hat{\bm \chi}_{(k+1)}(t_j) \in [0, 1]^2_+$, note that $\chi_{nt, (k+1)}(t_j)$ and $\chi_{at, (k+1)}(t_j)$ are obtained by maximizing with respect to $a, b$ a function of the form $x\log(a)+y \log(j)+z\log(1-a-b)$, for some non-negative quantities $x,y,z$. Clearly, this is maximized when $a=x/(x+y+z)$, $b=y/(x+y+z)$, which satisfy the requited constraints: $a, b \in [0, 1]$ and $0 \leq a+b \leq 1 $.
\end{proof}

To ensure the monotonicity constraint we apply the PAVA with the following weights to the vector $\breve{\bm\theta}_{(k+1)}$, which is computed in the above lemma: 
\begin{align}\label{weight}
w_{co, (k+1)}^{(0)}(t_j) &= n_{00}\overline F_{00}(t_j)  u_0^{(k)}(t_j) + n_{00}(1-\overline F_{00}(t_j))  v_0^{(k)}(t_j), \nonumber\\
w_{nt,  (k+1)}(t_j) &= n_{00}\overline F_{00}(t_j)(1-u_0^{(k)}(t_j)) + n_{00}(1-\overline F_{00}(t_j))(1-v_0^{(k)}(t_j)) + n_{10}, \nonumber\\
w_{at,  (k+1)}(t_j) &= n_{11}\overline F_{11}(t_j)(1-u_1^{(k)}(t_j)) + n_{11}(1-\overline F_{11}(t_j))(1-v_1^{(k)}(t_j)) + n_{01}, \nonumber \\
w_{co,  (k+1)}^{(1)}(t_j) &= n_{11}\overline F_{11}(t_j)  u_1^{(k)}(t_j) + n_{11}(1-\overline F_{11}(t_j))  v_1^{(k)}(t_j). 
\end{align}
This completes the description of the EM-PAVA algorithm. Proposition 2, which is proved below, shows that this procedure indeed maximizes $Q_m(\bm\theta,  \bm\chi| \hat{\bm\theta}_{(k)},  \hat{\bm \chi}_{(k)})$ over the restricted parameter space $\bm \vartheta_+ \times \bm \varphi_+$. \\

\noindent \textbf{\textit{Proof of Proposition 2}.} A collection of $f_1, f_2, \ldots, f_n: \R \rightarrow \R$ is said to be {\it nice} with respect to a given weight vector $\bm w=(w_1, w_2, \ldots, w_n)$ if the following hold:
\begin{itemize}
	\item[--] there exists $\tilde{\bm \theta}=(\tilde{\theta}_1, \ldots ,\tilde{\theta}_n)$ such that $\bar{\theta}_{ab} = \arg \max \sum_{s=a}^{b} f_s(\theta)$ can be represented as the weighted average of $(\tilde{\theta}_i, \ldots ,\tilde{\theta}_b)$, that is, 
	$$\bar{\theta}_{ab} = \frac{\sum_{s=a}^{b} w_s \tilde{\theta}_i}{\sum_{s=a}^{b} w_s} \quad \forall~ a \leq b,$$ 
	\item[--]$\sum_{s=a}^{b} f_s(\theta)$ is strictly increasing when $\theta \leq \bar{\theta}_{ab}$ and is strictly decreasing when $\theta > \bar{\theta}_{ab}$.
\end{itemize}

We will use the following well-known result about maximizing the sum of nice functions under the  monotonicity constraint.

\begin{lemma}\label{lm:mpava}\citep{ma2015} Let $f_1, f_2, \ldots, f_n: \R \rightarrow \R$ be collection of functions, and $\breve z_i=\arg\max_{z \in \R} f_i(z)$. If this collection of functions is {nice} with respect to a given weight vector $\bm w=(w_1, w_2, \ldots, w_n)$, then  
	$$\arg \max_{ z_1 \leq \ldots \leq z_n} \sum_{s=1}^{n} f_s(z_s)=\mathrm{PAVA}_{\bm w}(\breve{z}_1,\ldots,\breve{z}_n),$$
	where the PAVA uses the weight vector $\bm w$.
\end{lemma}

Since $\hat{\bm \chi}_{(k+1)} = \breve{\bm\chi}_{(k+1)} \in \bm \varphi_+$ (Lemma \ref{lm:mstep}), it suffices to show that $\hat{\bm \theta}_{(k+1)} \in \bm \vartheta_+$ and it is the restricted maximum. Note that the estimates 
$$\breve{\theta}_{co, (k+1)}^{(0)}(t_j), \breve{\theta}_{nt, (k+1)}(t_j), \breve{\theta}_{co, (k+1)}^{(1)}(t_j), \breve{\theta}_{at, (k+1)}(t_j) \in [0, 1],$$ for each $j$.  Therefore, the PAVA estimates $\hat{\bm \theta}_{(k+1)} \in \bm \vartheta_+$. Next, to apply Lemma \ref{lm:mpava} above, define the following four functions $f_{1j}, f_{2j}, f_{3j}, f_{4j}$: 
\begin{align*}
f_{1j}(\theta_{1j}) &= n_{00}\overline F_{00}(t_j)  u_0^{(k)} \log\theta_{1j} +  n_{00}(1-\overline F_{00}(t_j)  v_0^{(k)} \log(1-\theta_{1j})\nonumber\\
f_{2j}(\theta_{2j})  &= \left\{ n_{00}\overline F_{00}(t_j)  (1-u_0^{(k)}) + n_{10}\overline F_{10}(t_j)\right\}\log\theta_{2j} \nonumber\\
&~~~~~~+ \left\{n_{00}(1-\overline F_{00}(t_j)  (1-v_0^{(k)}) + n_{10}(1-\overline F_{10}(t_j)) \right\}\log(1-\theta_{2j}) \nonumber\\
f_{3j}(\theta_{3j}) &= n_{11}\overline F_{11}(t_j)  u_1^{(k)} \log\theta_{3j} + n_{11}(1-\overline F_{11}(t_j)  v_1^{(k)} \log(1-\theta_{3j})\nonumber\\
f_{4j}(\theta_{4j}) &= \left\{ n_{11}\overline F_{11}(t_j)  (1-u_1^{(k)}) + n_{01}\overline F_{01}(t_j)\right\}\log\theta_{4j} \nonumber\\
&~~~~~~+ \left\{n_{11}(1-\overline F_{11}(t_j)  (1-v_1^{(k)}) + n_{01}(1-\overline F_{01}(t_j)) \right\} \log(1-\theta_{4j}).
\end{align*}
where $\theta_{1j}=\theta_{co}^{(0)}(t_j), \theta_{2j}=\theta_{nt}(t_j), \theta_{3j}=\theta_{co}^{(1)}(t_j)$ and $\theta_{4j}=\theta_{at}(t_j)$. Then, from \eqref{expected_log_likelihood} and Lemma \ref{lm:Estep}, it follows that  
$$Q_{k}(\bm\theta,  \bm\chi|\hat{\bm\theta}_{(k)},  \hat{\bm \chi}_{(k)})= C(\bm \chi)+ \frac{1}{m}\sum_{s=1}^{4} \sum_{j=1}^{m} f_{sj}(\theta_{sj}),
$$
where $C(\bm\chi)$ is a function depending only of $\bm \chi$. Therefore, maximizing $Q_{k}(\bm\theta,  \bm\chi|\hat{\bm\theta}_{(k)},  \hat{\bm \chi}_{(k)})$ is equivalent to maximizing $ \sum_{j=1}^{m}  f_{sj}(\theta_{sj})$ for each $s$. Now,  for each $s$, it is easy to see that the functions $ f_{sj}(\theta_{sj})$, for $j$, satisfy the condition in Lemma~\ref{lm:mpava} with weights as in \eqref{weight}, and, hence the proof of  Proposition 2 follows.

\section{Additional Simulation Results}

\subsection{Size of the different tests as the sample size varies}

In this subsection, we empirically compute the estimated sizes of $T_n$ (asymp.), $T_n^{simple}$ (asymp.) and $T_{KS}$ for different values of sample size $n$. In particular, we use $n = (500, 1000, 1500, 2000, 3000, 4000)$. As in  Section 5.2 of the main manuscript, we considered two settings (far, close)  for the distributions $F_{nt}$ and $F_{at}$,
\begin{align*}
\text{close} &: F_{co}^{(0)} = F_{co}^{(1)} \sim N(0,1), \quad F_{nt} \sim N(-1, 1) \quad \text{and} \quad  F_{at} \sim N(1,1) \\
\text{far} &: F_{co}^{(0)} = F_{co}^{(1)} \sim N(0,1), \quad F_{nt} \sim N(-2, 1) \quad \text{and} \quad  F_{at} \sim N(2,1).
\end{align*}

\begin{table}
\centering
\caption{Size for various $n$ values}{%
	\begin{tabular}{lrccc}
	\hline
		$(\mu_{nt}, \mu_{at})$ & $n$ & $T_{n}$ & $T_{n}^{simple}$  &  $T_{KS}$\\
		\hline 
		$(-1,1)$ & 500 & 0.0391 & 0.0498 & 0.0395 \\
		& 1000 & 0.0410 & 0.0476 & 0.0475 \\
		& 1500 & 0.0423 & 0.0483 & 0.0447 \\
		& 2000 & 0.0436 & 0.0488 & 0.0519 \\
		& 3000 & 0.0484 & 0.0530 & 0.0518 \\
		& 4000 & 0.0487 & 0.0517 & 0.0498 \\[0.2cm]
		$(-2,2)$ & 500 & 0.0270 & 0.0512 & 0.0406 \\
		& 1000 & 0.0285 & 0.0472 & 0.0460 \\
		& 1500 & 0.0319 & 0.0515 & 0.0475 \\
		& 2000 & 0.0341 & 0.0504 & 0.0504\\
		& 3000 & 0.0341 & 0.0512 & 0.0470 \\
		& 4000 & 0.0364 & 0.0516 & 0.0524 \\
		\hline
	\end{tabular}}
\label{tab:size1}	
\end{table}

Table \ref{tab:size1} shows the estimated sizes for the two settings from 10,000 Monte Carlo simulations.  The column of $T_n$ in Table~\ref{tab:size1} shows that the estimated size approaches to the nominal level 0.05 as $n$ increases. When $(\mu_{nt}, \mu_{at}) = (-1,1)$, the speed of the convergence is faster, than when $(\mu_{nt}, \mu_{at}) = (-2,2)$. For the same set of $n$ values, it seems like the other tests, $T_n^{simple}$ and $T_{KS}$, have their sizes close to 0.05 over all $n$ values, and even for small $n$ values, the sizes are near 0.5. 

\subsection{Power of the different tests}

In this subsection, we compute the estimated power for different tests, including the bootstrap version of the (full) binomial likelihood ratio test $T_n$. We conducted 1000 Monte Carlo simulations, and for each simulation $B=1000$ bootstrapped samples were used. For $n=300$ and $(\phi_{co}, \phi_{nt}, \phi_{at}) = (1/3, 1/3, 1/3)$, we first consider the following setting: 
$$
\text{Situation 1}: F_{co}^{(0)} \sim N(0,1), F_{co}^{(1)} \sim N(0, 3), \quad F_{nt} \sim N(\mu_{nt}, 1), F_{at} \sim N(\mu_{at}, 1).$$
The results are given in Table~\ref{tab:power1}, which shows that $T_n$ (boot.) has the best performance overall. Interestingly, as the difference between $F_{nt}$ and $F_{at}$ increases,  the estimated powers for the BLRTs decreases. Figure~\ref{fig:dist} shows the true distribution functions of $F_0(t)$ and $F_1(t)$. Note that as the difference between $\mu_{nt}$ and $\mu_{at}$ increases, the difference of the two distributions is more concentrated at the center, which leads to increasing power for $T_{KS}$ and decreasing power for the BLRTs. The same pattern was discussed in the main manuscript, as well. 

\begin{table}[h]
\centering
	\caption{Power when variances are different (Situation 1)}{%
		\begin{tabular}{ccccc}
		\hline
			$(\mu_{nt}, \mu_{at})$ & $T_n$ (asymp.) & $T_{n}$ (boot.) & $T_{n}^{simple}$  &  $T_{KS}$\\
			\hline
			(0, 0) & 0.754 & 0.767 & 0.749 & 0.244 \\
			(-1, 1) & 0.761 & 0.784 & 0.754 & 0.324 \\
			(-2, 2) & 0.668 & 0.732 & 0.614 & 0.383 \\
			(-3, 3) & 0.393 & 0.493 & 0.359 & 0.338 \\ 
			\hline
	\end{tabular}}
	\label{tab:power1} 
\end{table}
\begin{figure}[h]
	\centering
	\includegraphics[width=130mm]{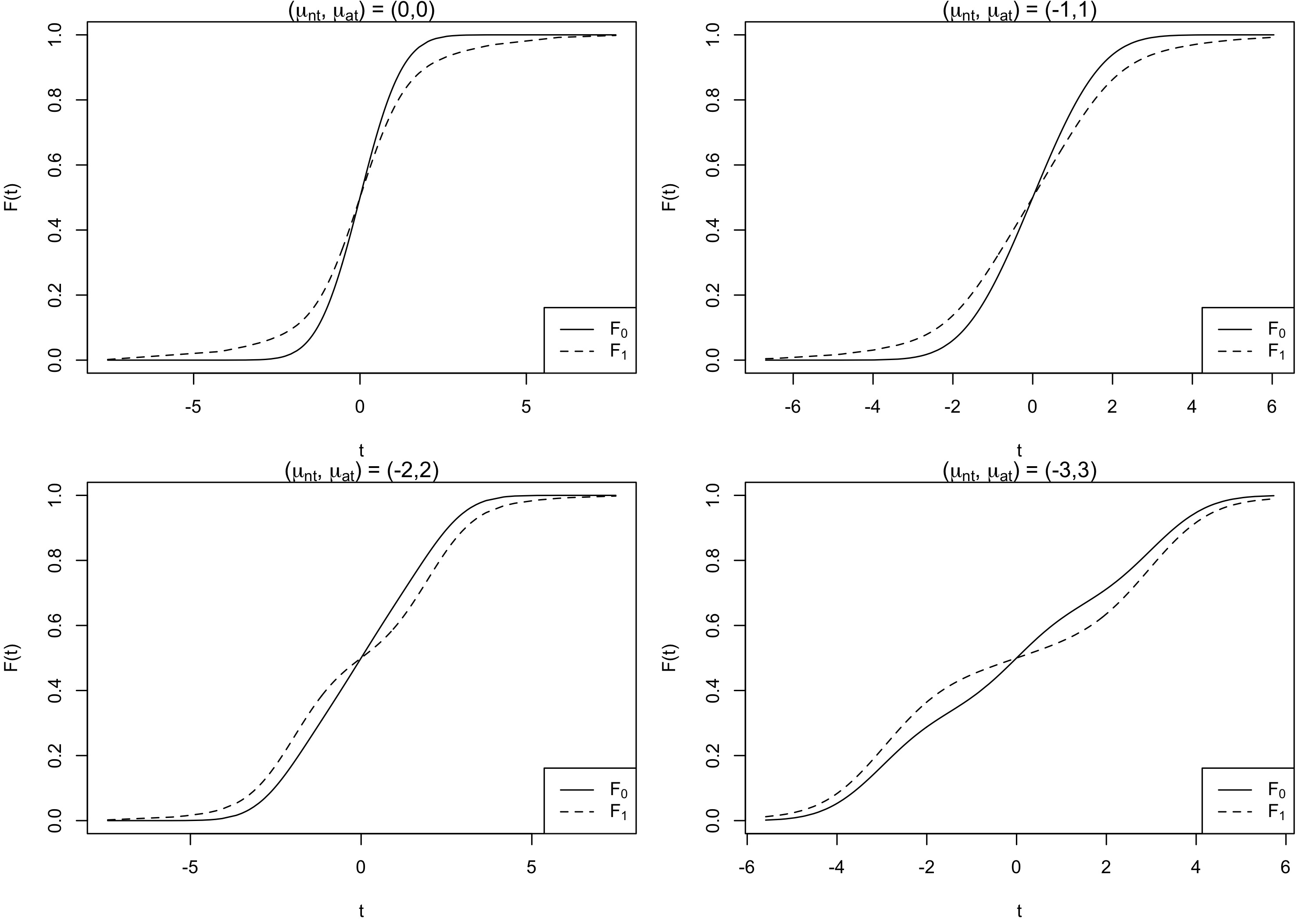}
	\caption{Comparison of $F_0(t)$  and $F_1(t)$ for different pairs of $(\mu_{nt}, \mu_{at})$.}
	\label{fig:dist}
\end{figure}

Next, we consider examples with non-normal distributions. In particular, we consider the following two cases: 
\begin{align*}
\text{Situation 2:}
\begin{array}{ccc}
F_{co}^{(0)} & \sim & p_{co}^{(0)} \cdot \bm{1}(t=0)+ (1-p_{co}^{(0)} ) \cdot \mathrm{Exp}(\mu_{co}^{(0)}),  \\
F_{co}^{(1)} & \sim & p_{co}^{(1)} \cdot \bm{1}(t=0)+ (1-p_{co}^{(1)} ) \cdot\mathrm{Exp}(\mu_{co}^{(1)}), \\
F_{nt} & \sim & p_{nt} \cdot \bm{1}(t=0)+ (1-p_{nt} ) \cdot\mathrm{Exp}(\mu_{nt}), \\
F_{at} & \sim & p_{at} \cdot \bm{1}(t=0)+ (1-p_{at} ) \cdot\mathrm{Exp}(\mu_{at}).
\end{array}
\end{align*}
\begin{align*}
\text{Situation 3:} 
\begin{array}{ccc}
F_{co}^{(0)} & \sim & p_{co}^{(0)} \cdot \bm{1}(t=0)+ (1-p_{co}^{(0)} ) \cdot  \text{lognor}(\mu_{co}^{(0)},1) ,  \\
F_{co}^{(1)} & \sim & p_{co}^{(1)} \cdot \bm{1}(t=0)+ (1-p_{co}^{(1)} ) \cdot \text{lognor}(\mu_{co}^{(1)},1) , \\
F_{nt} & \sim & p_{nt} \cdot \bm{1}(t=0)+ (1-p_{nt} ) \cdot \text{lognor}(\mu_{nt},1) , \\
F_{at} & \sim & p_{at} \cdot \bm{1}(t=0)+ (1-p_{at} ) \cdot \text{lognor}(\mu_{at},1), 
\end{array}
\end{align*}
where $F \sim \text{lognor}(\mu, 1)$ means $\log F \sim N(\mu, 1)$. Note that in both the situations the we have semi-continuous outcomes, that is, a mixture with a continuous component (Exponential or Log-normal) and a discrete component (a point mass at zero with a positive probability). 

\begin{table}[h]
\small 
\centering
	\caption{Power for non-normal distributions (Situations 2 and 3)}
	{%
		\begin{tabular}{lcccccc}
		\hline 
			& $(p_{co}^{(0)}, p_{co}^{(1)}, p_{nt}, p_{at})$ & $(\mu_{co}^{(0)}, \mu_{co}^{(1)}, \mu_{nt}, \mu_{at})$ & $T_n$ (asymp.) & $T_{n}$ (boot.) & $T_{n}^{simple}$  &  $T_{KS}$\\ 
			\hline 
			Situation 2 & (0, 0, 0, 0) & (1, 2, 1, 2) & 0.379 & 0.424 & 0.385 & 0.317 \\
			& (0, 0, 0, 0) & (1, 2, 0.1, 10) & 0.194 & 0.298 & 0.171 & 0.229 \\
			& (0.1, 0.2, 0.1, 0.1) & (1, 2, 0.1, 10) & 0.128 & 0.165 & 0.104 & 0.143 \\
			& (0.1, 0.3, 0.1, 0.1) & (1, 2, 0.1, 10) & 0.232 & 0.289 & 0.200 & 0.156 \\ [0.2cm] \\ 
			Situation 3 & (0, 0, 0, 0) & (-0.5, 0.5, -1, 1) & 0.658 & 0.718 & 0.635 & 0.601 \\ 
			& (0, 0, 0, 0) & (-0.5, 0.5, -2, 2) & 0.461 & 0.557 & 0.401 & 0.516 \\
			& (0.1, 0.2, 0.1, 0.1) & (-0.5, 0.5, -2, 2) & 0.245 & 0.312 & 0.182 & 0.293 \\
			& (0.1, 0.3, 0.1, 0.1) & (-0.5, 0.5, -2, 2) & 0.311 & 0.355 & 0.247 & 0.218 \\ 
			\hline
	\end{tabular}}
	\label{tab:power2}
\end{table}

Table~\ref{tab:power2} shows the estimated powers in the above two cases for $n=300$. Again, $T_n$ (boot.) overall has the highest power. As before, we see that when $\mu_{nt}$ and $\mu_{at}$ are distant, $T_n$ can get lower power than $T_{KS}$. However, even in this case, when $p_{co}^{(0)}$ and $p_{co}^{(1)}$ are far from each other, $T_n$ can be more powerful than $T_{KS}$. The difference in the proportion of zeros implies that the distributional difference is in the left tail, and, as a result, the BLRTs can detect such differences effectively.